\documentclass[11pt,a4paper]{article}
\usepackage[a4paper,hmargin=2.8cm,vmargin=3cm]{geometry}
\usepackage[utf8x]{inputenc}
\usepackage[bookmarksnumbered=true]{hyperref}
\usepackage{amsmath,amsthm,amsfonts,amssymb}
\usepackage{wasysym} %needed for \apprle
\let\oldmarginpar\marginpar
\renewcommand\marginpar[1]{\-\oldmarginpar[\raggedleft\footnotesize #1]%
  {\raggedright\footnotesize #1}}

\newtheorem{thm}{Theorem}[section]  % use thm for Theorems to keep numbering consistent

\newtheorem{proposition}[thm]{Proposition}
\newtheorem{lem}[thm]{Lemma}
\newtheorem{lemma}[thm]{Lemma}

\newcommand{\done}{}
\newcommand{\ech}[2]{#2}	%'exponential-bounds-change': replaces (e.g.) old norm-bounds by new exp ones (first arg by second), easily undoable
\newcommand{\ekt}{e^{K\lvert t\rvert}}	%exp-bound

\newcommand{\bx}{{\bf x}}

\newcommand{\cU}{{\cal U}}

\newcommand{\bR}{{\mathbb R}}

\newcommand{\bN}{{\mathbb N}}

\newcommand{\tr}{\mbox{Tr}}

\newcommand{\wt}{\widetilde}

\newcommand{\cF}{{\cal F}}

\newcommand{\cE}{{\cal E}}

\newcommand{\cK}{{\cal K}}
\newcommand{\cH}{{\cal H}}
\newcommand{\cL}{{\cal L}}
\newcommand{\cN}{{\cal N}}

\newcommand{\R}{\mathbb{R}}
\newcommand{\N}{\mathcal{N}}

\newcommand{\ad}{\operatorname{ad}}	% abbreviation for iterated commutators
\newcommand{\di}{{d}}		% differential (for integrals)
\newcommand{\Ncal}{\mathcal{N}}		% calligraphic N
\newcommand{\Hcal}{\mathcal{H}}		% calligraphic H
\newcommand{\scal}[2]{\big<#1,#2\big>} % scalaer product
\newcommand{\cc}[1]{\overline{#1}}	% complex conjugate
\newcommand{\Rbb}{\mathbb{R}}		% real numbers
\newcommand{\Nbb}{\mathbb{N}}		% natural numbers
\renewcommand{\Im}{\operatorname{Im}\,} 	%ImaginaryPart
\newcommand{\norm}[1]{\lVert#1\rVert}	%Norm
\newcommand{\ph}{\varphi_t^{(N)}}	% solution of N-dependent Hartree equation
\newcommand{\phdot}{\dot{\varphi}_t^{(N)}}	% time derivative of solution of N-dependent Hartree equation
\newcommand{\phddot}{\ddot{\varphi}_t^{(N)}}	% second time derivative of solution of N-de Hartree equaution
\newcommand{\be}[1]{\begin{equation}\label{eq:#1}}	%begin equation with label
\newcommand{\ee}{\end{equation}}
\newcommand{\bd}{\begin{displaymath}}			% abbreviation begin displaymath
\newcommand{\ed}{\end{displaymath}}

\newcommand{\todo}[1]{}

\title{Quantitative Derivation of the Gross-Pitaevskii Equation} 

\author{Niels Benedikter, Gustavo de Oliveira\thanks{Supported by ERC Grant MAQD 240518 }  \, and Benjamin Schlein\thanks{Partially supported by ERC Grant MAQD 240518} \\ \\ Institute of Applied Mathematics, University of Bonn\\ Endenicher Allee 60, 53115 Bonn, Germany}

\begin{document}
\maketitle

\begin{abstract}
Starting from first principle many-body quantum dynamics, we show that the dynamics of 
Bose-Einstein condensates can be approximated by the time-dependent nonlinear 
Gross-Pitaevskii equation, giving a bound on the rate of the convergence. Initial data are constructed on the bosonic Fock space applying an appropriate Bogoliubov transformation on a coherent state with expected number of particles $N$. The Bogoliubov transformation plays a crucial role; it produces the correct microscopic correlations among the particles. Our analysis shows that, on the level of the one particle reduced density, the form of the initial data is preserved by the many-body evolution, up to a small error which vanishes as $N^{-1/2}$ in the limit of large $N$.  
\end{abstract}

\section{Introduction and main results}
\label{s:intro}

A Bose-Einstein condensate is a state of matter of a gas of bosons where a macroscopic fraction of the particles occupy the same one-particle state. The existence of Bose-Einstein condensation at small temperature was first predicted in 1925 by Bose and Einstein, who considered gases of non-interacting bosons. Seventy years later, in 1995, the existence of Bose-Einstein condensates 
was then confirmed by experiments; see \cite{BEC1,BEC2}. Since then, Bose-Einstein condensates have attracted a lot of attention in theoretical and in experimental physics. In particular, they have been used to explore fundamental questions in quantum mechanics, such as the emergence of interference, decoherence, superfluidity and quantized vortices. In experiments, condensates are initially trapped by strong magnetic fields and cooled down at extremely low temperatures (in the nano-kelvin scale). Then, the traps are released and the subsequent evolution of the condensate is observed. The goal of this paper is to study the dynamics of initially trapped Bose-Einstein condensates. In particular, starting from many-body quantum dynamics, we 
show rigorously that the evolution of the condensate can be described, in certain regimes, by the time-dependent Gross-Pitaevskii equation. 

%This gives a mathematical description of the experiments discussed above.

\bigskip

{\it The model.} We consider a trapped gas of $N$ bosons, described by the Hamilton operator 
\begin{equation}\label{eq:ham-0}
H^{\text{trap}}_N = \sum_{j=1}^N \left(-\Delta_{x_j} +V_{\text{ext}} (x_j)\right) + \sum_{i<j}^N N^2 V (N (x_i - x_j)) 
\end{equation}
acting on the Hilbert space $L^2_s (\bR^{3N}, d\bx)$, the subspace of $L^2 (\bR^{3N}, d\bx)$ consisting of all functions symmetric with respect to permutations of the $N$ particles. The external potential $V_{\text{ext}}$ confines the particles inside the trap. The interaction potential $V$ is assumed to be non-negative, spherically symmetric and decaying sufficiently fast at infinity. We denote by $a_0$ the scattering length of the potential $V$. To define the scattering length, we consider the solution of the zero energy scattering equation 
\begin{equation}\label{eq:0en-0} \left( -\Delta + \frac{1}{2} V \right) f = 0 \end{equation}
with the boundary condition $f (x) \to 1$ as $|x| \to \infty$. The scattering length is then given by 
\begin{equation}\label{eq:8pia} 8 \pi a_0 = \int dx \, V(x) f(x) \, . \end{equation}
It is easy to check that, for large $|x|$, 
\begin{equation}\label{eq:f} f(x) = 1- \frac{a_0}{|x|} + O (|x|^{-2}) \, . \end{equation}
By simple scaling, we find then
\[ \left( -\Delta + \frac{N^2}{2} V(N x) \right) f (Nx) = 0 \]
which implies that the scattering length of the rescaled potential $N^2 V(N.)$ appearing in (\ref{eq:ham-0}) is given by $a=a_0/N$. 

\bigskip

{\it Ground state properties.} Let \[ E_N = \min_{\substack{\psi_N \in L^2_s (\bR^{3N}) \\ \| \psi_N \| = 1}} \langle \psi_N , H^{\text{trap}}_N \psi_N \rangle \]
denote the ground state energy of (\ref{eq:ham-0}). It was proven in \cite{LSY} that
\[ \lim_{N \to \infty} \frac{E_N}{N} = \min_{\substack{\varphi \in L^2 (\bR^3) \\ \| \varphi \|_2 = 1}} \cE_{\text{GP}} (\varphi) \]
where 
\begin{equation}\label{eq:GP-en} \cE_{\text{GP}} (\varphi) = \int dx \, \left[ |\nabla \varphi (x)|^2 + V_{\text{ext}} (x) |\varphi (x)|^2 + 4 \pi a_0 |\varphi (x)|^4 \right] \end{equation}
is the Gross-Pitaevskii energy functional. Hence, in the leading order, the ground state energy per particle depends on the interaction potential $V$ only through its scattering length $a_0$. In \cite{LS}, it was also shown that the ground state of (\ref{eq:ham-0}) exhibits complete Bose-Einstein condensation in the minimizer of the Gross-Pitaevskii energy functional (\ref{eq:GP-en}),
in the sense that 
\begin{equation}\label{eq:cond} \gamma_N^{(1)} \to | \phi_{\text{GP}} \rangle \langle \phi_{\text{GP}}| \end{equation}
where $|\phi_{\text{GP}} \rangle \langle \phi_{\text{GP}} |$ is the orthogonal projection onto the (normalized) 
minimizer $\phi_{\text{GP}} \in L^2 (\bR^3)$ of (\ref{eq:GP-en}) and where $\gamma^{(1)}_N$ denotes the one-particle reduced density associated with the ground state $\psi_N \in L^2_s (\bR^{3N})$ of (\ref{eq:ham-0}), which is defined as the non-negative trace class operator with integral kernel
\begin{equation}\label{eq:one-red} \gamma^{(1)}_N (x;x') := \int dx_2 \dots dx_N \psi_N (x , x_2, \dots  ,x_N) \overline{\psi}_N (x' , x_2 , \dots , x_N). \end{equation}
We assume here that $\| \psi_N \|  =1$, and therefore that $\tr \, \gamma^{(1)}_N = 1$. The convergence in (\ref{eq:cond}), which hold in the trace-class topology, implies that, in the ground state of (\ref{eq:ham-0}), all particles, up to a fraction which vanishes in the limit of large $N$, are in the same one particle state, described by the orbital $\phi_{\text{GP}}$. The results of \cite{LSY,LS} show that the Gross-Pitaevskii theory correctly describes the ground state properties of the Hamiltonian (\ref{eq:ham-0}). 

\bigskip

{\it Time evolution.} 
%Let us suppose now that the bosonic system is prepared in the ground state of (\ref{eq:ham-0}); this %is achieved, in the experiments described above, by cooling down the gas to extremely low %temperatures. What happens if the traps are switched off? In this case, the system immediately start 
When the traps are switched off, the system starts to evolve, the dynamics being governed by the $N$-particle Schr\"odinger equation 
\begin{equation}\label{eq:schr} i \partial_t \psi_{N,t} = H_N \psi_{N,t} \end{equation}
with the translation invariant Hamiltonian
\begin{equation}\label{eq:ham-1} H_N = \sum_{j=1}^N -\Delta_{x_j} +
\sum_{i<j}^N N^2 V(N (x_i - x_j)). \end{equation}
It turns out that the time evolution of an initial data exhibiting complete condensation can be described, in the limit of large $N$, by the Gross-Pitaevskii theory. In fact, the following result was established in 
\cite{ESY0,ESY1,ESY2,ESY3,ESY4}, and, in a slightly different form, in \cite{P}. Consider a family $\psi_N \in L^2_s (\bR^{3N})$ with bounded energy per particle
\[ \langle \psi_N , H_N \psi_N \rangle \leq C N \]
and exhibiting complete condensation in a one-particle state $\varphi \in H^1 (\bR^3)$, in the sense that the one-particle reduced density $\gamma^{(1)}_N$ associated with $\psi_N$ (defined as in (\ref{eq:one-red})) satisfies
\[ \gamma_N^{(1)} \to |\varphi \rangle \langle \varphi| \] 
as $N \to \infty$. Then, the solution $\psi_{N,t}$ of the Schr\"odinger equation (\ref{eq:schr}) still exhibits complete Bose-Einstein condensation, in the sense that the reduced one-particle density $\gamma_{N,t}^{(1)}$ associated with $\psi_{N,t}$ satisfies 
\begin{equation}\label{eq:conv-GP} \gamma_{N,t}^{(1)} \to |\varphi_t \rangle \langle \varphi_t| \end{equation}
as $N \to \infty$, where $\varphi_t$ is the solution of the time-dependent Gross-Pitaevskii equation
\begin{equation}\label{eq:GP-0} i\partial_t \varphi_t = -\Delta \varphi_t +
8\pi a_0 |\varphi_t|^2 \varphi_t. \end{equation}
This result establishes the stability of complete Bose-Einstein condensation with respect to the time-evolution, and the fact that the condensate wave function evolves according to (\ref{eq:GP-0}). It justifies therefore the use of the Gross-Pitaevskii equation (\ref{eq:GP-0}) to predict and describe the outcome of the experiments discussed above. 

\bigskip

{\it Mean field regime.} The method used in \cite{ESY0,ESY1,ESY2,ESY3,ESY4} 
to prove (\ref{eq:conv-GP}) relies on techniques first developed to understand the mean-field limit of many-body quantum dynamics. This regime is achieved when considering the time-evolution 
\begin{equation}\label{eq:mf-schr}
i\partial_t \psi_{N,t} = H_N^{\text{mf}} \psi_{N,t} 
\end{equation}
generated by the Hamiltonian
\begin{equation}\label{eq:ham-mf} H_N^{\text{mf}} = \sum_{j=1}^N -\Delta_{x_j} + \frac{1}{N} \sum_{i<j}^N V (x_i -x_j) \end{equation}
in the limit of large $N$. Also in this case, under appropriate assumptions on the potential $V$, 
complete condensation is preserved by the time-evolution. Here, the evolution of the condensate wave function is governed by the Hartree equation
\begin{equation}\label{eq:hartree-0} i\partial_t \varphi_t = -\Delta \varphi_t + \left( V * |\varphi_t|^2\right) \varphi_t \,. \end{equation}
The first rigorous derivation of (\ref{eq:hartree-0}) was obtained in \cite{S} for bounded interaction potential, i.e.\ under the assumption $\| V \|_\infty < \infty$. The basic idea in \cite{S} was to study directly the time-evolution of the family of reduced densities $\gamma^{(k)}_{N,t}$, $k =1 ,2, \dots , N$, defined similarly to (\ref{eq:one-red}).  {F}rom the Schr\"odinger equation (\ref{eq:mf-schr}), it is easy to derive a hierarchy of $N$ coupled equations, known as the BBGKY hierarchy, for the evolution of the densities $\gamma^{(k)}_{N,t}$. As $N \to \infty$, the BBGKY hierarchy converges, at least formally, towards an infinite hierarchy of coupled equations, which is solved by products of the solution of the Hartree equation (\ref{eq:hartree-0}). The problem of proving the convergence towards the Hartree dynamics essentially reduces to showing the uniqueness of the solution of the infinite hierarchy. This general scheme, first introduced in \cite{S} for bounded interactions, was later extended to potentials with Coulomb singularities in \cite{EY} (and in \cite{ES}, for bosons with relativistic dispersions). In \cite{ESY1,ESY2,ESY3,ESY4}, the same strategy was then applied to analyze the dynamics generated by the Hamiltonian (\ref{eq:ham-1}) and to obtain a rigorous derivation of the Gross-Pitaevskii equation (\ref{eq:GP-0}). 

Writing (\ref{eq:ham-1}) as
\begin{equation}\label{eq:ham-11} H_N = \sum_{j=1}^N -\Delta_{x_j} + \frac{1}{N} \sum_{i<j}^N N^3 V (N (x_i - x_j)) \end{equation}
one can interpret the Gross-Pitaevskii regime as a very singular mean-field limit, where the interaction converges, as $N \to \infty$, towards a delta-function. {F}rom the point of view of physics, however, these two regimes are very different. While in the mean field limit particles experience a large number of weak collisions, the evolution generated by (\ref{eq:ham-0}) is characterized by very rare and strong interactions (two particles only interact when they are very close, at distances of order $N^{-1}$ from each others, which is much smaller than the typical interparticle distance, of order $N^{-1/3}$). As a consequence, it turns out that correlations among particles, which are negligible in the mean field limit, play a crucial role in the Gross-Pitaevskii regime. To explain this point, let us consider the evolution of the one-particle reduced density $\gamma^{(1)}_{N,t}$ which is governed by the equation 
\begin{equation}\label{eq:BBGKY1} \begin{split} 
i\partial_t & \gamma^{(1)}_{N,t} (x;x') \\ = \; & \left( -\Delta_x + \Delta_{x'} \right) \gamma^{(1)}_{N,t} (x;x')  
\\ &+ \int dx_2 \, \left( (N-1) N^2 V (N (x-x_2) - (N-1) N^2 V(N (x'-x_2)) \right) \gamma^{(2)}_{N,t} (x,x_2 ; x',x_2) \end{split} \end{equation}
where $\gamma^{(2)}_{N,t}$ is the two-particle reduced density, defined similarly to 
(\ref{eq:one-red}), with the normalization $\tr \; \gamma^{(2)}_{N,t} = 1$. If we assume the initial 
state to exhibit complete condensation and we accept that condensation is preserved by the time evolution, we should expect $\gamma^{(1)}_{N,t}$ and $\gamma^{(2)}_{N,t}$ to be approximately factorized. In $\gamma^{(2)}_{N,t}$, however, we also want to take into account the correlations between the two particles. Describing the correlations through the solution of the zero-energy scattering equation (\ref{eq:0en-0}), we are led to the ansatz
\begin{equation}\label{eq:ans-mod} \begin{split} 
\gamma^{(1)}_{N,t} (x;x') &\simeq \varphi_t (x) \overline{\varphi}_t (x'), \\
\gamma^{(2)}_{N,t} (x_1, x_2; x'_1, x'_2) &\simeq  f (N (x_1 - x_2)) f (N
(x'_1- x'_2)) \, \varphi_t (x_1)  \varphi_t (x_2) \overline{\varphi}_t
(x'_1) \overline{\varphi}_t (x'_2).
\end{split} \end{equation}
Plugging this into (\ref{eq:BBGKY1}), we obtain a new self-consistent equation for $\varphi_t$,  given by
\begin{equation}\label{eq:gp-mod0} i\partial_t \varphi_t = -\Delta \varphi_t
+ \left(N^2 (N-1) V(N.) f(N.) * |\varphi_t|^2 \right) \varphi_t.
\end{equation}
{F}rom (\ref{eq:8pia}), we have $N^2 (N-1) V(Nx) f(Nx) \simeq N^3 V(Nx) f(Nx) \simeq 8\pi a_0 \delta (x)$ as $N \to \infty$; therefore, in this limit, $\varphi_t$ must be a solution of the Gross-Pitaevskii equation (\ref{eq:GP-0}). We observe here that the presence of the factor $f(N.)$, describing the correlations among the particles, is crucial in this argument to understand the emergence of the scattering length in  (\ref{eq:GP-0}). We conclude that any derivation of the Gross-Pitaevskii equation (\ref{eq:GP-0}) must take into account the singular correlation structure developed by the many body evolution. In fact, understanding the correlations and adapting the techniques of \cite{S,EY,ES} to deal with them (for example, when proving a priori regularity of the limiting densities) was one of the main challenges faced in \cite{ESY0,ESY1,ESY2,ESY3,ESY4}. As we will discuss shortly, correlations also play a major role in the approach presented in this paper. 

\bigskip

{\it The coherent states approach.} A drawback of the approach used in \cite{ESY0,ESY1,ESY2,ESY3,ESY4} is the fact that it does not give any control on the rate of the convergence of the many-body quantum evolution towards the limiting dynamics (described by the Hartree or the Gross-Pitaevskii equation). 
%This is a consequence of the abstract compactness argument used to establish the convergence of %the sequence $\{ \gamma^{(k)}_{N,t} \}_{k=1}^N$ of reduced densities. 
%For mean field dynamics with bounded potentials, it is possible to avoid the compactness argument %by expanding directly the solution of the BBGKY hierarchy (\ref{eq:BBGKY}).
%(instead of expanding the simpler infinite hierarchy to prove the uniqueness of its solution). 
%This was in fact the approach originally used in \cite{S}. 
%Even so, the expansion can only be shown to converge for short times, and the bounds on the rate %of convergence obtained in this way deteriorate very fast. 
%. Iteration of the argument gives
%convergence for all times, but the control on the rate of convergence deteriorates very fast, leading %to bounds of the form  
%\[ \tr\, \left| \gamma^{(1)}_{N,t} - |\varphi_t \rangle \langle \varphi_t| \right| \leq C N^{-\frac{1}{2^t}} \] 
%which can be easily proven for bounded potentials, following the ideas of \cite{S}. 
%For the dynamics generated by (\ref{eq:ham-1}), no such bound can be extracted (or, at least, has %been extracted so far) with the approach presented above. 
The problem of controlling the rate of convergence is by no means of purely academic interest. In experiments, the number of particles $N$ is large ($N \simeq 10^3$ in typical samples of Bose-Einstein condensates), but of course finite. For fixed $N$, the statement (\ref{eq:conv-GP}) is empty. Only an explicit  bound on the difference $\gamma^{(1)}_{N,t} - |\varphi_t \rangle \langle \varphi_t|$ can tell us whether the limiting (Hartree or Gross-Pitaevskii) dynamics is a good approximation for the many-body evolution. In \cite{RS}, a different approach to the study of the many-body quantum dynamics in the mean field regime was developed, starting from ideas introduced in a slightly different context in \cite{H,GV}. 
To explain this approach, we consider the bosonic Fock space 
\[ \cF = \bigoplus_{n \geq 0} L^2_s (\bR^{3n} , dx_1 \dots  dx_n). \]
The idea here is that on $\cF$ we can describe states with variable number
of particles. The normalized Fock-state $\psi = \{ \psi^{(n)} \}_{n\geq 1}$
contains $n$ particles with probability $\| \psi^{(n)} \|^2$. For $f \in L^2
(\bR^3)$ we introduce, as usual, creation and annihilation operators $a^*
(f)$ and $a(f)$ which act on $\cF$ by creating and annihilating a particle
with wave function $f$, respectively (precise definitions and basic
properties are given in Section \ref{sec:fock}). For $x \in \bR^3$, we can
also introduce operator valued distributions $a_x^*$ and $a_x$, creating and annihilating a particle at $x$. In terms of these distributions, we define the mean field Hamiltonian
\[ \cH^{\text{mf}}_N = \int dx \nabla_x a_x^* \nabla_x a_x + \frac{1}{2N}
\int dxdy \, V(x-y) a_x^* a_y^*a_y a_x. \]
The operator $\cH_N$ commutes with the number of particles operator $\cN$, whose action on $\psi = \{ \psi^{(n)} \}_{n\geq 0} \in \cF$ is given by $(\cN \psi)^{(n)} = n \psi^{(n)}$. Moreover, when restricted to the sector with exactly $N$ particles, $\cH^{\text{mf}}_N$ coincides exactly with the mean field Hamiltonian (\ref{eq:ham-mf}) introduced above. The advantage of working on the Fock space is that we have more freedom in the choice of the initial state. We will use this freedom, by choosing coherent initial states. The coherent state with wave function $f \in L^2 (\bR^3)$ is defined as the vector $W(f) \Omega \in \cF$, where $\Omega = \{ 1, 0, \dots \}$ is the Fock vacuum, and $W(f)= \exp (a^* (f) - a(f))$ is the Weyl operator with wave function $f$.  A simple computation shows that
\[ W(f) \Omega = e^{-\| f \|_2^2/2} \sum_{j \geq 0} \frac{(a^* (f))^j}{j!}
\Omega = e^{-\|f \|_2^2/2} \left\{ 1, f , \frac{f^{\otimes 2}}{\sqrt{2!}} ,
\dots \right\}. \]
Coherent states do not have a fixed number of particles. However, the expected number of particles is given by $\| f \|_2^2$. To recover the mean field regime discussed above, we consider therefore the time evolution of an initial coherent state $W(\sqrt{N} \varphi) \Omega$, for $\varphi \in L^2 (\bR^3)$ with $\| \varphi \|_2 = 1$. It turns out that coherent states have especially nice algebraic properties (due to the fact that they are eigenvectors of all annihilation operators). Making use of these properties, one can show that, for large $N$, the mean field evolution $\Psi_{N,t} = e^{-i\cH^{\text{mf}}_N t} W(\sqrt{N} \varphi) \Omega$ of an initial coherent state will again be approximately coherent, of the form $W(\sqrt{N} \varphi_t) \Omega$, where $\varphi_t$ is the solution of the Hartree equation (\ref{eq:hartree-0}) with initial data $\varphi_{t=0} = \varphi$.

Moreover, it is possible to express the difference between the reduced densities associated with the evolved coherent state and the orthogonal projection onto $\varphi_t$ in terms of a so called fluctuation dynamics, defined as the two-parameter group of unitary transformations
\begin{equation}\label{eq:cU-mf} \cU^{\text{mf}} (t;s) = W^* (\sqrt{N}
\varphi_t) e^{-i \cH^{\text{mf}}_N (t-s)} W(\sqrt{N} \varphi_s). \end{equation}
Essentially, the problem of bounding the rate of the convergence of the
many-body evolution towards the Hartree dynamics in the mean field limit
reduces to the problem of controlling the growth of the expectation of the
number of particles operator $\cN$ with respect to $\cU^{\text{mf}} (t;s)$,
uniformly in $N$. To obtain such estimates, we observe that the fluctuation dynamics satisfies a Schr\"odinger type equation
\[ i\partial_t \cU^{\text{mf}} (t;s) = \cL^{\text{mf}}_N (t) \cU^{\text{mf}} (t;s) \]
with $\cU^{\text{mf}} (s;s) = 1$ and with the time dependent generator
\begin{equation}\label{eq:cL-mf} \begin{split} 
\cL^{\text{mf}}_N (t) = \; & [i\partial_t W^* (\sqrt{N} \varphi_t)] W(\sqrt{N} \varphi_t) + W^* (\sqrt{N} \varphi_t) \cH^{\text{mf}}_N W(\sqrt{N} \varphi_t) \\  = \; &\int dx \nabla_x a_x^* \nabla_x a_x + \int dx ( V*|\varphi_t|^2) (x) a_x^* a_x + \int dx dy V(x-y) \varphi_t (x) \overline{\varphi}_t (y) \, a_x^* a_y \\ &+  \int dx dy V(x-y)  \left(\varphi_t (x) \varphi_t (y) \, a_x^* a^*_y + \overline{\varphi}_t (x) \overline{\varphi}_t (y) \, a_x a_y \right) \\
&+\frac{1}{\sqrt{N}} \int dx dy V(x-y) a_x^* \left( \varphi_t (y) a_y^* +
\overline{\varphi}_t (y) a_y \right) + \frac{1}{2N} \int dx dy V(x-y) a_x^*
a_y^* a_y a_x. \end{split} \end{equation}
The terms on the third and fourth line do not commute with the number of particles operator (because the number of creation operator does not match the number of annihilation operators). As a consequence, the fluctuation dynamics $\cU^{\text{mf}} (t;s)$ does not leave the number of particles invariant. This is not surprising since $\cU^{\text{mf}} (t;s)$ describes fluctuations around the Hartree evolution, which are expected to grow with time. Still, under the assumption that the interaction $V$ contains at most Coulomb singularities, bounds of the form
\[ \langle \psi , \cU^{\text{mf}} (t;0)^* \cN \cU^{\text{mf}} (t;0) \psi \rangle \leq C e^{K |t|} \langle \psi, (\cN+1)^{4} \psi \rangle \]
for the growth of the number of particles operator (and, actually, also for all its power) were proven in \cite{RS}. As a corollary, estimates of the form
\begin{equation}\label{eq:bds-mf} \tr \, \left| \Gamma_{N,t}^{(1)} - |\varphi_t \rangle \langle \varphi_t| \right| \leq \frac{C e^{K|t|}}{N} \,  \end{equation}
on the rate of convergence towards the Hartree dynamics followed. Here $\Gamma_{N,t}^{(1)}$ denotes the one-particle reduced density associated with the evolved state $\Psi_{N,t} = e^{-i\cH_N t} W(\sqrt{N} \varphi) \Omega$. The one-particle reduced density associated with a Fock space vector $\Psi$ is defined by the integral kernel
\begin{equation}\label{eq:GNt} \Gamma_{\Psi}^{(1)} (x,y) = \frac{1}{\langle
\Psi , \cN \Psi \rangle} \left\langle \Psi , a_y^* a_x \Psi \right\rangle. \end{equation}
It is simple to check that, for $N$-particle states, (\ref{eq:GNt}) coincides with the definition (\ref{eq:one-red}). 
The analysis of the time evolution of initial coherent states is also useful to study the dynamics of initial data with a fixed number of particles $N$. For $N$-particle initial data $\psi_N \in L^2 (\bR^{3N})$ obtained projecting down a Fock-state of the form $W(\sqrt{N} \varphi) \Psi$ onto the $N$-particle sector, bounds of the form (\ref{eq:bds-mf}) were established in \cite{CLS}, extending the ideas developed in \cite{RS}, for arbitrary $\Psi \in \cF$ with finite number of particles and energy (note that this class of $N$-particle states include factorized wave functions of the form $\varphi^{\otimes N}$, which are obtained choosing $\psi = \Omega$). Using techniques similar to those proposed in \cite{P}, bounds for the rate of convergence towards the nonlinear Hartree dynamics were also obtained in \cite{KP2009}, allowing also for potential with strong singularities. A different point of view on the mean field limit was given in \cite{FKS}, where the convergence towards the Hartree dynamics was interpreted as a Egorov-type theorem.

\bigskip

{\it Bogoliubov transformations and the Gross-Pitaevskii regime.} It seems natural to ask whether the coherent state approach introduced in \cite{RS} can be used to obtain a rigorous derivation of the Gross-Pitaevskii equation (\ref{eq:GP-0}), providing at the same time bounds on the rate of the convergence of the form (\ref{eq:bds-mf}). A major difficulty in this program is immediately clear. 
{F}rom (\ref{eq:cL-mf}), we notice that there are two contributions to the generator $\cL_N (t)$, one 
 arising from the derivative of the Weyl operator $W^* (\sqrt{N}
 \varphi_t)$, the other from the derivative of $e^{-i (t-s) \cH_N}$. The
 second contribution, given by $W^* (\sqrt{N} \varphi_t) \cH^{\text{mf}}_N
 W(\sqrt{N} \varphi)$, can be computed recalling that Weyl operators act as shifts on creation and annihilation operators (see (\ref{eq:W3})). It turns out that this contribution contains a term, linear in creation and annihilation operators, having the form
\begin{equation}\label{eq:lin-mf}  \sqrt{N} \int dx \left[ -\Delta \varphi_t (x) + (V*|\varphi_t|^2) (x) \varphi_t (x) \right] a_x^* + \text{h.c.} \end{equation}
This term is large (of the order $N^{1/2}$) and does not commute with the number of particles operator. With such a term in the generator, it would be impossible to show uniform (in $N$) bounds for the growth of the number of particles. In the mean field regime, however, (\ref{eq:lin-mf}) is exactly 
canceled by the contribution proportional to the derivative of $W^* (\sqrt{N} \varphi_t)$, which contain the term
\[ \begin{split} -\sqrt{N} \int dx \, & ( i \partial_t \varphi_t (x) ) a_x^* - \text{h.c.} = -
\sqrt{N} \int dx \left[ -\Delta \varphi_t (x) + (V*|\varphi_t|^2) (x) \varphi_t (x) \right] a_x^* - \text{h.c.} 
\end{split}
\]
where we used the equation (\ref{eq:hartree-0}). As a result, the generator $\cL_N (t)$ on the r.h.s.\ of (\ref{eq:cLN}) contains only terms which, at least formally, are either order one, or smaller.

To adapt this approach to the Gross-Pitaevskii regime, we define the Fock-space Hamiltonian
\begin{equation}\label{eq:cHN} \cH_N = \int dx \, \nabla_x a_x^* \nabla_x a_x + \frac{1}{2} \int dx dy N^2 V(N (x-y)) a_x^* a_y^* a_y a_x \end{equation}
and, following (\ref{eq:cU-mf}), we naively introduce the fluctuation dynamics
\begin{equation}\label{eq:cU-mfGP} 
\cU^{GP} (t;s) = W^* (\sqrt{N} \varphi^{(N)}_t) e^{-i (t-s) \cH_N} W(\sqrt{N} \varphi^{(N)}_t) \end{equation}
where, for technical reasons which will be clear later, we choose $\varphi^{(N)}_t$ to solve the modified Gross-Pitaevskii equation 
\begin{equation}\label{eq:mod-GP0} i\partial_t \varphi^{(N)}_t (x) = -\Delta \varphi^{(N)}_t (x) + (N^3 V(N.) f(N.) * |\varphi^{(N)}_t|^2) (x) \varphi^{(N)}_t (x) \end{equation}
where $f$ is the solution of the zero-energy scattering equation (\ref{eq:0en-0}) (we use the notation $\varphi_t^{(N)}$ to distinguish the solution of (\ref{eq:mod-GP0}) from the solution $\varphi_t$ of the limiting Gross-Pitaevskii equation (\ref{eq:GP-0})). Since the solution of (\ref{eq:mod-GP0}) can be shown to converge towards the solution of (\ref{eq:GP-0}), with an error of the order $N^{-1}$, control of the fluctuations around (\ref{eq:mod-GP0}) also implies control of the fluctuations around (\ref{eq:GP-0}). Let $\cL_N^{\text{GP}}$ denote the generator of (\ref{eq:cU-mfGP}), given by sum of the two contributions $(i\partial_t W^* (\sqrt{N} \varphi^{(N)}_t)) W (\sqrt{N} \varphi^{(N)}_t)$ and $W^* (\sqrt{N} \varphi^{(N)}_t) \cH_N W (\sqrt{N} \varphi^{(N)}_t)$. As in the mean field regime, the term $W^*(\sqrt{N} \varphi^{(N)}_t) \cH_N W(\sqrt{N} \varphi^{(N)}_t)$ contains a large contribution, linear in creation and annihilation operators, given by
\[ \sqrt{N} \int dx \left[ -\Delta \varphi^{(N)}_t (x) + (N^3 V (N.) * |\varphi^{(N)}_t|^2) (x) \varphi^{(N)}_t (x) \right] a_x^* + \text{h.c.} \] 
The term $(i\partial_t W^* (\sqrt{N} \varphi^{(N)}_t)) W (\sqrt{N} \varphi^{(N)}_t)$, on the other hand, contains the linear term
\[ \begin{split} 
- \sqrt{N} \int dx \big[ (i\partial_t &\varphi^{(N)}_t (x)) a_x^* + \text{h.c.}\big]  \\
&=- \sqrt{N} \int dx \left[ -\Delta \varphi^{(N)}_t (x) + (N^3 V(N.) f(N.) |\varphi^{(N)}_t (x)|^2 \varphi^{(N)}_t (x) \right] a_x^* - \text{h.c.}
\end{split}\]
In contrast with the mean field regime discussed above, here, because of the factor $f$, there is no complete cancellation between the two large linear terms. Hence, the generator $\cL^{\text{GP}}_N (t)$ of (\ref{eq:cU-mfGP}) contains a large contribution, linear in the creation and annihilation operators, of the form
\begin{equation}\label{eq:lin-N} \sqrt{N} \int dx \left( N^3 V(N.) (1- f(N.)) * |\varphi^{(N)}_t|^2 \right) (x) \left( \varphi^{(N)}_t (x) a_x^* + \text{h.c.} \right)  \end{equation}
and it seems impossible to obtain uniform (in $N$) bound on the growth of the number of particles w.r.t. (\ref{eq:cU-mfGP}). The reason for this failure is that we are trying to control fluctuations around the wrong evolution. When we approximate $e^{-it \cH_N} W(\sqrt{N} \varphi) \psi$ by an evolved coherent state $W(\sqrt{N} \varphi^{(N)}_t) \psi$, we are completely neglecting the correlation structure developed by the many-body evolution. As a result, fluctuations around the coherent approximation $W(\sqrt{N} \varphi^{(N)}_t)$ are too strong to be bounded uniformly in $N$. 
Since correlations are, in first order, an effect of two-body interactions, we are going to approximate them using a unitary operator of the form
\[ T(k) = \exp\left(\frac{1}{2} \int dx dy \left(k(x,y) a_x^* a_y^* - \overline{k} (x,y) a_x a_y \right)\right) \]
for an appropriate $k \in L^2 (\bR^3 \times \bR^3)$, which will be interpreted as the integral kernel of a Hilbert-Schmidt operator (again denoted by $k$) on $L^2 (\bR^3)$. The operator $T(k)$ acts on creation and annihilation operators as a Bogoliubov transformation (see Sect. \ref{sec:bt} for precise definitions and basic properties):
\[ \begin{split} T^* (k) a (f) T (k) &= a(\text{ch} (k) (f)) + a^* (\text{sh} (k) (\overline{f})) \\
T^*(k) a^* (f) T (k) &= a^* (\text{ch} (k) (f)) + a^* (\text{sh} (\overline{k}) (\overline{f}))
\end{split} \]
where $\text{ch} (k)$ and $\text{sh} (k)$ are the bounded operators on $L^2 (\bR^3)$ defined by the (absolute convergent) series
\[ \begin{split} 
\text{ch} (k) & = \sum_{n\geq 0} \frac{1}{(2n)!} (k\overline{k})^n \qquad \text{and } \quad 
\text{sh} (k)   = \sum_{n \geq 0} \frac{1}{(2n+1)!} (k\overline{k})^n k \end{split} \]
where products of $k$ and $\overline{k}$ have to be understood in the sense of operators. Inspired by the analysis of \cite{ESY1,ESY2,ESY3,ESY4}, where correlations were successfully described by the solution of the zero-energy scattering equation (\ref{eq:f}), we define the (time-dependent) kernel 
\begin{equation}\label{eq:kt-def0} k_t (x,y) = - N w (N(x-y)) \varphi^{(N)}_t (x) \varphi^{(N)}_t (y) \end{equation}
where $\varphi^{(N)}_t$ is the solution of the modified Gross-Pitaevskii equation (\ref{eq:mod-GP0}), and where
\[ w (x) = 1 - f(x) \]
with $f$ being the solution of (\ref{eq:0en-0}). We will consider an initial state of the form $W(\sqrt{N} \varphi) T(k_0) \psi$, for $\psi$ with bounded number of particles and bounded energy, and we will approximate its time evolution by $W(\sqrt{N} \varphi_t) T(k_t) \psi$, leading therefore to the fluctuation dynamics 
\begin{equation}\label{eq:cU-0} \cU (t;s) = T^* (k_t) W^* (\sqrt{N}
\varphi_t) e^{-i(t-s) \cH_N} W(\sqrt{N} \varphi_s) T(k_s). \end{equation}
In this way, the approximating dynamics takes into account the correct correlation structure, and we have a better chance to bound the fluctuations. In fact, we will show in Section \ref{s:growthoffluct} that it is indeed possible to obtain a uniform (in $N$) control for the growth of the number of particles w.r.t. the dynamics (\ref{eq:cU-0}). We will see in Section \ref{sec:gen-fd} that the large linear contribution (\ref{eq:lin-N}), which will still appear in the generator of (\ref{eq:cU-0}), is compensated by a contribution arising from sandwiching the cubic term in $W^* (\sqrt{N} \varphi_t^{(N)}) \cH_N W(\sqrt{N} \varphi_t^{(N)})$ between $T(k_t)$ and $T^* (k_t)$. In fact, after conjugating this expression with the Bogoliubov transformation, some of the (cubic) terms will not be in normal order (a creation operator will lie on the right of an annihilation operator). By the canonical commutation relations, normal ordering produces then terms which are linear in the creation and annihilation operators; some of these terms cancel exactly the large contribution (\ref{eq:lin-N}). Other important cancellations will emerge between the quadratic and the not normally ordered quartic terms; see Section \ref{sec:gen-fd} for the details. The control of the growth of the number of particles w.r.t. (\ref{eq:cU-0}) will imply convergence of the reduced densities associated with the fully evolved Fock state $e^{-i \cH_N t} W(\sqrt{N} \varphi_t) T(k_t) \psi$ towards the orthogonal projection onto the solution of the Gross-Pitaevskii equation (\ref{eq:GP-0}), with a bound on the rate of the convergence. Note that the idea of adding an operator with a quadratic expression in creations and annihilation operators in the exponent to the Weyl operator generating the coherent state was used in the mean-field regime in \cite{GMM2010} to obtain norm-approximations of the many-body dynamics. 

\bigskip

{\it The main theorem.} We are now ready to state our main result. 
\begin{thm}\label{thm:main}
Let $\varphi \in H^4 (\bR^3)$, with $\| \varphi \|_2 =1$. Let $\cH_N$ be the
Hamilton operator defined in (\ref{eq:cHN}), with a non-negative and
spherically symmetric potential $V \in L^1\cap L^3 (\bR^3, (1+|x|^6) dx)$. Let $\psi \in \cF$ (possibly depending on $N$) be such that
\begin{equation}\label{eq:ass-thm1} \langle \psi , \cN \psi \rangle ,  \frac{1}{N} \langle \psi , \cN^2 \psi \rangle, \langle \psi, \cH_N \psi \rangle \leq D \end{equation}
for a constant $D>0$. Let $\Gamma_{N,t}^{(1)}$ denote the one-particle reduced density associated with the evolved state $e^{-i\cH_N t} W(\sqrt{N} \varphi) T(k_0) \psi$, as defined in (\ref{eq:GNt}). Then there exist constants $C, c_1, c_2 >0$, depending only on $V, \| \varphi \|_{H^4}$ and on the constant $D$ appearing in (\ref{eq:ass-thm1}), such that
\begin{equation}\label{eq:mt} \tr \left| \Gamma_{N,t}^{(1)} - |\varphi_t\rangle \langle \varphi_t| \right| \leq \frac{C \exp (c_1 \exp (c_2 |t|))}{N^{1/2}} \end{equation}
for all $t \in \bR$ and $N \in \bN$. Here $\varphi_t$ denotes the solution of the time-dependent Gross-Pitaevskii equation
\begin{equation}\label{eq:GP} i\partial_t \varphi_t = -\Delta \varphi_t + 8 \pi a_0 |\varphi_t|^2 \varphi_t \end{equation}
with the initial condition $\varphi_{t=0} =\varphi$.
\end{thm}

{\it Remarks.}
\begin{itemize}
\item[(i)] Let us point out that we insert the correct correlation structure in the initial data. Our result implies the approximate stability of states of the form $W(\sqrt{N} \varphi) T(k_0) \psi$ with respect to the many-body evolution (in the sense that the evolution of $W(\sqrt{N} \varphi) T(k_0) \psi$ has approximately the same form, just with evolved $\varphi_t$, up to a small error). It does not imply, on the other hand, that the correlation structure is produced by the time-evolution. This is in contrast with the results of \cite{ESY1,ESY2,ESY3,ESY4}, which can also be applied to completely factorized initial data. It remains unclear, however, if it is possible to obtain convergence with a $N^{-1/2}$ rate (or with any rate) for initial data with 
no correlations (the problem of the creation of correlation was studied in \cite{EMS}). 
\item[(ii)] The time dependence on the r.h.s.\ of (\ref{eq:mt}) deteriorates fast for large $t$. This however is just a consequence of the fact that, in general, high Sobolev norms of the solution of (\ref{eq:GP}) can grow exponentially fast. Assuming a uniform bound for $\| \varphi_t \|_{H^4}$, the time dependence on the r.h.s.\ of (\ref{eq:mt}) can be replaced by $C \exp (K |t|)$. 
%In fact, if one assumes the $L^\infty$-norm of $\varphi_t$ and of its time-derivative $\dot\varphi_t$ to %vanish sufficiently fast as $t \to \infty$ (at an integrable rate), one should even obtain bounds which %are uniform in $t$. For the mean field regime, this was observed in \cite{KP2009}.
\item[(iii)] To simplify a little bit the computations, we did not include an external potential in the Hamiltonian (\ref{eq:cHN}) generating the evolution on the Fock space. In contrast to \cite{ESY1,ESY2,ESY3,ESY4}, the approach presented in this paper can be extended with no additional complication to Hamilton operators with external potential. This remark is important to describe experiments where the evolution of the condensate is observed after tuning the magnetic traps, rather than switching them off. 
\item[(iv)] The convergence (\ref{eq:mt}) and the fact that the limit is a rank-one projection immediately implies convergence of the higher order reduced density $\Gamma_{N,t}^{(k)}$, associated with the evolved state $\Psi_{N,t} = e^{-i\cH_N t} W(\sqrt{N} \varphi) T(k_0) \psi$. Similarly to (\ref{eq:GNt}), $\Gamma^{(k)}_{N,t}$ is defined as the non-negative trace class operators on $L^2 (\bR^{3k})$ with integral kernel
\[ \Gamma_{N,t}^{(k)} (x_1, \dots , x_k ; y_1, \dots, y_k) = \frac{\langle \Psi_{N,t}, a_{y_1}^* \dots a_{y_k}^* a_{x_k} \dots a_{x_1} \Psi_{N,t} \rangle}{\langle \Psi_{N,t} , \cN ( \cN-1) \dots (\cN-k+1) \Psi_{N,t} \rangle} \, .  \]
Following the arguments outlined in Section 2 of \cite{KP2009}, (\ref{eq:mt}) implies that, for every $k \in \bN$,
\[ \tr \, \left| \Gamma^{(k)}_{N,t} - |\varphi_t \rangle \langle \varphi_t|^{\otimes k} \right| \leq C \frac{k^{1/2}}{N^{1/4}} \exp\left(\frac{c_1}{2} \exp (c_2 |t|)\right).\done \]
To obtain bounds for the convergence of the $k$-particle reduced density with the same $N^{-1/2}$ rate as in (\ref{eq:mt}), following the same approach used below to study $\Gamma^{(1)}_{N,t}$ would require to control the growth of higher powers of the number of particle operator with respect to the fluctuation dynamics (\ref{eq:cU-0}). This may be doable, but the analysis becomes more involved. 
\item[(v)] Theorem \ref{thm:main} and the method used in its proof can also be applied to deduce the convergence towards the Gross-Pitaevskii dynamics for certain initial data with a fixed number of particles. In Appendix \ref{sec:Npart}, we consider initial $N$-particle states of the form 
$P_N W(\sqrt{N} \varphi) T(k_0) \psi$, for $\psi \in \cF$ satisfying (\ref{eq:ass-thm1}), assuming 
$\| P_N W(\sqrt{N} \varphi) T(k_0) \psi \| \gg N^{-1/2}$ for large $N$ (it is explained in Appendix \ref{sec:Npart} why this is a reasonable condition). Here $P_N$ denotes the orthogonal projection onto the $N$-particle sector of $\cF$. It remains to be understood which class of $N$-particle states can be written as $P_N W(\sqrt{N} \varphi) T(k_0) \psi$, for a $\psi \in \cF$ satisfying (\ref{eq:ass-thm1}). 
\end{itemize}

%It is also possible to extract convergence towards the limiting Gross-Pitaevskii dynamics for a certain %class of initial data having a fixed number of particles $N$. This is the content of our second %theorem.
%\begin{thm}\label{thm:Nstates}
%Let $\varphi \in H^4 (\bR^3)$, with $\| \varphi \|_2 =1$. Let $\cH_N$ be the Hamilton operator defined %in (\ref{eq:cHN}), with an non-negative and spherically symmetric potential $V \in L^1 (\bR^3)$ %decaying sufficiently fast at infinity (so that $|V(x)| \leq |x|^{-4-\eps}$, for some $\eps > 0$). Let \
%[ \psi_N = P_N W(\sqrt{N} \varphi) T(k_0) \psi / \| P_N W(\sqrt{N} \varphi) T(k_0) \psi \|,\]  for a $\psi \in %\cF$ satisfying (\ref{eq:ass-thm1}). 
%Here $P_N$ denotes the orthogonal projection onto the Fock space sector with exactly $N$ %particles. Hence, we can identify $\psi_N$ as a sequence in $L^2_s (\bR^{3N})$, and we can define %$\psi_{N,t} = e^{-iH_N t} \psi_N$, with the Hamilton operator (\ref{eq:ham-1}). Let $\gamma_{N,t}%^{(1)}$ denote the one-particle reduced density associated with $\psi_{N,t}$. Then
%\[  \tr \left| \gamma_{N,t}^{(1)} - |\varphi_t\rangle \langle \varphi_t| \right| \leq \frac{C e^{c_1 e^{c_2 |%t|}}}{N^{1/4}} \]
%for all $t \in \bR$ and $N \in \bN$, where $\varphi_t$ is the solution of the Gross-Pitaevskii equation %(\ref{eq:GP}), with initial data $\varphi_{t=0} = \varphi$. 
%\end{thm}

\section{Operators on the Fock space}
\label{sec:fock}

The bosonic Fock space over $L^2(\R^3)$ is the Hilbert space
\[
  \mathcal{F} = \bigoplus_{n \ge 0} L^2(\R^3)^{\otimes_s n} = \mathbb{C}
  \oplus \bigoplus_{n \ge 1} L^2_s(\R^{3n}),
\]
with the convention that $L^2(\R^3)^{\otimes_s 0} = \mathbb{C}$. Here
$L^2_s (\R^{3n})$ is the subspace of $L^2(\R^{3n})$ consisting of all functions
that are symmetric with respect to arbitrary permutations of the $n$
variables. Vectors in $\mathcal{F}$ are sequences $\psi = \{\psi^{(n)}\}_{n
\ge 0}$ of $n$-particle wave functions $\psi^{(n)} \in L^2_s(\R^{3n})$. The
inner product on $\mathcal{F}$ is defined as
\begin{align*}
  \langle \psi_1, \psi_2 \rangle & = \sum_{n \ge 0} \langle \psi_1^{(n)},
  \psi_2^{(n)} \rangle_{L^2(\R^{3n})} \\
  & = \overline{\psi_1^{(0)}} \psi_2^{(0)} + \sum_{n \ge 1} \int dx_1 \cdots
  dx_n \overline{\psi_1^{(n)}}(x_1, \dots, x_n) \psi_2^{(n)}(x_1, \dots, x_n).
\end{align*}

On $\cF$, we can describe states where the number of particles is not fixed.
The Fock space vector $\psi = \{ \psi^{(0)}, \psi^{(1)}, \dots  \}$
describes a coherent superposition of states with different number of
particles; the $n$-particle component is described by $\psi^{(n)}$, for any
$n \in \bN$ (the probability that the state has exactly $n$ particles is
given by $\| \psi^{(n)} \|^2$).  A state with exactly $N$ particles is
described on $\mathcal{F}$ by a sequence $\{\psi^{(n)}\}_{n \ge 0}$ where
$\psi^{(n)} = 0$ for all $n \neq N$ and $\psi^{(N)} = \psi_N$. The vector
$\{1, 0, 0, \dots \} \in \mathcal{F}$ is called the vacuum and is denoted by
$\Omega$. States with a fixed number of particles are eigenvectors of the
number of particles operator, defined by
\[ (\cN \psi)^{(n)} = n \psi^{(n)} \, . \]

To define an evolution on $\cF$, we introduce the Hamilton operator 
\[ (\cH_N \psi)^{(n)} = \cH_N^{(n)} \psi^{(n)} \]
with the $n$-th sector operator 
\[ \cH_N^{(n)} = \sum_{j=1}^n -\Delta_{x_j} + \sum_{i<j}^n N^2 V(N (x_i - x_j)). \]
Note that the subscript $N$ in the notation $\cH_N$ is not related with the number of particles (since this is not fixed on the Fock space), but only reflects the scaling in the interaction potential (of course, at the end, $N$ will also be related with the number of particles in the initial Fock state; otherwise, there would be no relation with the regime discussed in Section \ref{s:intro}). Observe that, by definition, the Hamiltonian $\cH_N$ commutes with $\cN$. As a consequence, the evolution generated by $\cH_N$ leaves each $n$-particle sector invariant. In particular, 
\[ e^{-i\cH_N t}  \left\{ 0, \dots , 0, \psi_N , 0 ,\dots \} = \{ 0, \dots, 0, e^{-iH_N t} \psi_N , 0 ,\dots \right\} \]
where $H_N$ is the $N$-particle Hamiltonian defined in (\ref{eq:ham-1}). In this sense, the $N$-body dynamics is embedded in the Fock space representation. 

It is very useful to introduce creation and annihilation operators on $\cF$. 
For $f \in L^2(\R^3)$, the creation operator $a^*(f)$ and the annihilation
operator $a(f)$ on $\mathcal{F}$ are defined as
\[
  \begin{split}
    (a^*(f) \psi)^{(n)}(x_1, \dots, x_n) & = \frac{1}{\sqrt{n}} \sum_{j=1}^n
    f(x_j) \psi^{(n-1)}(x_1, \dots, x_{j-1}, x_{j+1}, \dots, x_n), \\
    (a(f) \psi)^{(n)}(x_1, \dots, x_n) & = \sqrt{n+1} \int dx \,
    \overline{f(x)} \psi^{(n+1)}(x, x_1, \dots, x_n).
  \end{split}
\]
The operators $a^*(f)$ and $a(f)$ are unbounded, densely defined and closed.
Note that $a^*(f)$ is linear in $f$, while $a(f)$ is anti-linear. The creation
operator $a^*(f)$ is the adjoint of the annihilation operator $a(f)$, and
they satisfy the canonical commutation relations
\begin{equation}\label{eq:ccr}
  [a(f), a^*(g)] = \langle f, g \rangle_{L^2} \qquad \text{and} \qquad [a(f),
  a(g)] = [a^*(f), a^*(g)] = 0
\end{equation}
for $f,g \in L^2(\R^3)$. We also introduce the self-adjoint operator
\[
  \phi(f) = a^*(f) + a(f).
\]

Although creation and annihilation operators are not bounded, they are bounded with respect to the number of particles operator $\cN$ (actually, to its square root). This is the content of the next standard lemma (see \cite{RS} for a proof of this well-known result).
\begin{lem} \label{l:a}
  Let $f \in L^2(\R^3)$. Then, for any $\psi \in
  \mathcal{F}$,
  \begin{equation} \label{aNorm}
    \begin{aligned}
      \norm{a(f)\psi} & \leq \norm{f}_2 \norm{\Ncal^{1/2}\psi}, \\
      \norm{a^*(f)\psi} & \leq \norm{f}_2 \norm{(\Ncal+1)^{1/2}\psi}, \\
      \norm{\phi(f) \psi} & \leq 2 \norm{f}_2 \norm{(\N+1)^{1/2} \psi}.
    \end{aligned}
  \end{equation}
\end{lem}

We will make use of operator-valued distributions $a_x^*$ and $a_x$, with $x
\in \R^3$, defined so that
\[
  a^*(f) = \int \di x\, f(x) a^*_x \qquad \text{and} \qquad a(f) = \int \di
  x\, \cc{f(x)} a_x
\]
for $f \in L^2(\R^3)$. The canonical commutation
relations assume the form
\[
  [a_x, a_y^*] = \delta(x-y) \qquad \text{and} \qquad [a_x, a_y] = [a_x^*,
  a_y^*] = 0,
\]
where $\delta$ is the Dirac delta distribution.

In terms of the operator valued distributions $a_x^*$ and $a_x$, the number of particles operator $\cN$ 
and the Hamilton operator $\cH_N$ can be written as
\[ \cN = \int dx \, a_x^* a_x \]
and 
\begin{equation}\label{eq:ham-F} \cH_N = \int dx \, \nabla_x a_x^* \nabla_x a_x + \frac{1}{2} \int dx dy N^2 V(N (x-y)) a^*_x a^*_y a_y a_x. \end{equation}
The first term in the Hamiltonian is the kinetic energy; since it will play an important role in our analysis, we introduce the notation 
\[ \cK = \int dx \, \nabla_x a_x^* \nabla a_x. \]
Note that, like $\cH_N$, $\cK$ leaves every $n$-particle sector invariant, and
\[ (\cK \psi)^{(n)} = \sum_{j=1}^n -\Delta_{x_j} \psi^{(n)} \]
for any $\psi \in \cF$. 

%It is important to introduce the notion of reduced density matrices for vectors on the Fock space. To %this end, we define, for any $\psi \in \cF$, the $k$-particles reduced density as the non-negative %trace-class operator on $L^2 (\bR^{3k})$ with integral kernel 
%\begin{equation} \label{fock_density}
%\Gamma_{\psi}^{(k)}(x_1, \dots, x_k ; y_1, \dots , y_k) := \frac{\scal{\psi}{a^*_{x_1} \dots a^*_{x_k} %a_{y_1} \dots a_{y_k} \psi}}{\scal{\psi}{\Ncal(\Ncal-1)\cdots(\Ncal-k+1)\psi}}.\done
%\end{equation}
%It is simple to check that, for $\psi$ in the $N$-particle sector of $\fock$, this definition is equivalent to %definition \eqref{eq:kred}.

\subsection{Weyl operators and coherent states}

For $f \in L^2(\R^3)$, we define the Weyl operator 
\[
  W(f) = e^{a^*(f) - a(f)}\done
\]
acting on the Fock space $\cF$. In the following lemma we collect some well-known properties of Weyl operators.
\begin{lem}\label{l:W}
  Let $f, g \in L^2(\R^3)$.
  \begin{enumerate}
    \item \label{l:W1} Weyl operators satisfy the Weyl relations
      \[
        W(f) W(g) = W(g) W(f) e^{-2i \Im \langle f, g \rangle_{L^2}} = W(f+g)
        e^{-i \Im \langle f, g \rangle_{L^2}}.
      \]
    \item \label{l:W2} The operator $W(f)$ is unitary on $\mathcal{F}$ and
      \[
        W(f)^* = W(f)^{-1} = W(-f).
      \]
    \item \label{l:W3} We have
      \begin{equation}\label{eq:W3}
        W(f)^* a (g)  W(f) = a (g) + \langle g, f \rangle_{L^2}  \qquad \text{and} \qquad W(f)^* a^* (g)
        W(f) = a^* (g) + \langle f, g \rangle_{L^2}.
      \end{equation}
%      where $\langle . , . \rangle$ denotes the inner product on $L^2 (\bR^3)$. 
 \end{enumerate}
\end{lem}

Using Weyl operators, we construct coherent states on $\cF$. For $f \in L^2 (\bR^3)$, the coherent state with wave function $f$ is defined as $W(f) \Omega$, where $\Omega = \{ 1, 0, \dots \}$ is the vacuum vector in $\cF$, describing a state with no particles. Since $W(f)$ is unitary, coherent states are always normalized. {F}rom the canonical commutation relations (\ref{eq:ccr}), it follows that  
\[ \begin{split} 
W(f) \Omega &= e^{-\| f \|_2^2 /2} e^{a^* (f)} \Omega = e^{-\| f\|_2^2/2}
\sum_{n \geq 0} \frac{(a^* (f))^n}{n!} \Omega \\ &= e^{-\| f\|_2^2/2}
\left\{ 1 , f , \frac{f^{\otimes 2}}{\sqrt{2!}} , \dots , \frac{f^{\otimes
n}}{\sqrt{n!}} , \dots \right\}. \end{split} \]
In particular, coherent states do not have a fixed number of particles. Instead, they are given by linear combinations of states with all possible number of particles. {F}rom (\ref{eq:W3}), the expected 
number of particles in the state $W(f) \Omega$ is given by
\[ \langle W(f) \Omega, \cN \, W(f) \Omega \rangle = \int dx \, \langle \Omega, (a_x^* + \overline{f} (x)) (a_x + f(x)) \Omega \rangle = \| f \|_2^2. \]
More precisely, one can show that the number of particles in the coherent state $W(f) \Omega$ is a Poisson random variable with average and variance given by $\| f \|_2^2$. 

Coherent states have particularly nice algebraic properties, which also simplify the study of their time-evolution. These properties are a consequence of the fact that coherent states are eigenvectors of all annihilation operators; in fact, from (\ref{eq:W3}), we find 
\[  a(g) W(f) \Omega = W(f) ( a(g) + \langle g , f \rangle ) \Omega = \langle g ,f \rangle W(f) \Omega \]
for all $f,g \in L^2 (\bR^3)$.

\subsection{Bogoliubov transformations}
\label{sec:bt}

For $f,g \in L^2 (\bR^3)$, we introduce the notation
\begin{equation}\label{eq:Afg} A(f,g) = a(f) + a^* (\overline{g}). \end{equation}
We observe that 
\begin{equation}\label{eq:bog11} 
(A(f,g))^* = A(\overline{g}, \overline{f}) = A \left( \left( \begin{array}{ll} 0 & J \\ J & 0 \end{array} \right) (f,g) \right)\end{equation}
where $J:L^2 (\bR^3) \to L^2 (\bR^3)$ is the antilinear operator defined by $Jf = \overline{f}$. {F}rom the canonical commutation relations (\ref{eq:ccr}), we find that the operators $A(f,g)$ satisfy the commutation relations
\begin{equation}\label{eq:bog22} \left[ A(f_1, g_1) , A^* (f_2, g_2) \right] = \left\langle (f_1, g_1) , S (f_2, g_2) \right\rangle_{L^2 \oplus L^2} \qquad \text{with } S = \left( \begin{array}{ll} 1 & 0 \\ 0 & -1 \end{array} \right). \end{equation}

A bounded linear map $\Theta : L^2 (\bR^3) \oplus L^2 (\bR^3) \to L^2 (\bR^3) \oplus L^2 (\bR^3)$ satisfying
\begin{equation}\label{eq:bog1}
\Theta \left( \begin{array}{ll} 0 & J \\ J & 0 \end{array} \right) = \left( \begin{array}{ll} 0 & J \\ J & 0 \end{array} \right) \Theta \end{equation}
and
\begin{equation}\label{eq:bog2} S = \Theta^* \, S \, \Theta \end{equation}
is called a Bogoliubov transformation. Bogoliubov transformations are linear maps of the pairs $(f,g) \in L^2 (\bR^3) \oplus L^2 (\bR^3)$ with the property that the operators $b^* (f)$ and $b(f)$,\done{} defined similarly to (\ref{eq:Afg}) by the equation \[ b (f) + b^* (\overline{g}) = A (\Theta (f,g)) \] for any $f,g \in L^2 (\bR^3)$, are still creation and annihilation operators satisfying canonical commutation relations and being adjoint to each other. It is simple to check that a general Bogoliubov transformation can be written as  
\begin{equation}\label{eq:theta-UV} \Theta = \left( \begin{array}{ll} U &  \overline{V}  \\ V & \overline{U} \end{array} \right) \end{equation}
for bounded linear maps $U, V: L^2 (\bR^3) \to L^2 (\bR^3)$ satisfying $U^* U - V^* V = 1$ and $U^* \overline{V} = V^*\overline{U}$. Here we use the notation $\overline{U} = J U J$ for any bounded operator $U$ on $L^2 (\bR^3)$ (the integral kernel of $\overline{U} = J U J$ is given by $\overline{U(x,y)}$).

For a kernel $k \in L^2 (\bR^3 \times \bR^3)$ with $k (x,y) = k(y,x)$, we define now the operator 
\begin{equation}\label{eq:Tk-def} T (k) = \exp \left(\frac{1}{2} \int dx dy \, \left(k(x,y) a_x^* a_y^* - \overline{k} (x,y) a_x a_y \right) \right) \end{equation} 
acting on the Fock space $\cF$. 
\begin{lem} \label{l:bt}
Let $k \in L^2(\R^3 \times \R^3)$ be symmetric, in the sense that $k(x,y) = k(y,x)$.
\begin{itemize}
\item[(i)] The operator $T(k)$ is unitary on $\mathcal{F}$ and
  \[
    T(k)^* = T(k)^{-1} = T(-k).
  \]
\item[(ii)] For every $f,g \in L^2 (\bR^3)$, we have 
\begin{equation}\label{eq:TAT} T (k)^* A(f,g) T(k) = A (\Theta_k (f,g)) \end{equation}
where $\Theta_k : L^2 (\bR^3) \oplus L^2 (\bR^3) \to L^2 (\bR^3) \oplus L^2 (\bR^3)$ is the Bogoliubov transformation defined by the matrix 
\[ \Theta_k = \left(\begin{array}{ll} \text{ch} (k)  & \text{sh} (k)  \\ \text{sh} (\overline{k}) & \text{ch} (\overline{k}) \end{array} \right) \]
where $\text{ch} (k), \text{sh} (k) : L^2 (\bR^3) \to L^2 (\bR^3)$ are the bounded operators defined by 
\[ \begin{split} 
\text{ch} (k) & = \sum_{n\geq 0} \frac{1}{(2n)!} (k\overline{k})^n \\
\text{sh} (k)  & = \sum_{n \geq 0} \frac{1}{(2n+1)!} (k\overline{k})^n k \end{split} \]
where products of $k$ and $\overline{k}$ have to be understood in the sense of operators. 
\item[(iii)] We decompose 
\begin{equation}\label{eq:pr-def} \text{ch} (k) = 1 + p (k) , \quad \text{sh} (k) = k + r (k), \end{equation}
where $1$ denotes the identity operator on $L^2 (\bR^3)$. Then $p (k)$ and $r (k)$ (and therefore sh$(k)$) are Hilbert-Schmidt operators, with 
\begin{equation}
%\begin{split}
 \| p (k) \|_{2}  \le e^{\| k \|_{2}}, \quad 
\| r (k) \|_{2} \le e^{\| k \|_{2}},  \quad \| \text{sh} (k) \|_2 \leq  e^{\| k \|_{2}}.
%\end{split}
\end{equation}
(Here $\| p (k) \|_2$ denotes the $L^2 (\bR^3 \times \bR^3)$ norm of the kernel $p(k) (x,y)$, which agrees with the Hilbert-Schmidt norm of the operator $p(k)$).
\item[(iv)] Suppose now that $k \in L^2 (\bR^3 \times \bR^3)$ is s.t.\ $\nabla_1 k  \in L^2 (\bR^3 \times \bR^3)$. Then, by symmetry, also $\nabla_2 k \in L^2 (\bR^3 \times \bR^3)$ (we use here the notation $(\nabla_1 k) (x,y) = \nabla_x k (x,y)$ and $(\nabla_2 k)(x,y) = \nabla_y k (x,y)$; note that $\nabla_1 k$ and $\nabla_2 k$ are the integral kernels of the operator products $\nabla k$ and $-k \nabla$). Moreover 
\[  \begin{split}  \| \nabla_1 p (k) \|_{2} , \| \nabla_1 r (k) \|_2 & \le
e^{\| k \|_{2}} \| \nabla_1 (k \overline{k}) \|_{2}, \\ \| \nabla_2 p (k)  \|_{2}, \| \nabla_2 r (k) \|_2 & \le e^{\| k \|_{2}} \| \nabla_2 (\overline{k} k)   \|_2.  \end{split} \]
\item[(v)] If the kernel $k$ depends on a parameter $t$ (later, it will depend on time), and if derivatives w.r.t. $t$ are denoted by a dot, we have
\[   \| \dot{p} (k) \|_2 , \| \dot{r} (k)\|_2  \leq \| \dot{k} \|_2 \, e^{\| k \|_2} \]
and
\[ \begin{split} 
\| \nabla_1 \dot p (k) \|_2 , \| \nabla_1 \dot r (k) \|_2  &\leq C  e^{\| k
\|_2} \left( \| \dot k \|_2 \| \nabla_1 (k\overline{k}) \|_2 + \| \nabla_1
(\dot{k} \overline{k}) \|_2 + \| \nabla_1 (k \overline{\dot{k}}) \|_2
\right),\done \\
 \| \nabla_2 \dot p (k) \|_2 , \| \nabla_2 \dot r (k) \|_2  &\leq C  e^{\| k \|_2} \left( \| \dot k \|_2 \| \nabla_2 (k\overline{k}) \|_2 + \| \nabla_2 (\dot{k} \overline{k}) \|_2 + \| \nabla_2 (k \overline{\dot{k}}) \|_2 \right).\done 
 \end{split} \]
\end{itemize}
\end{lem}

\begin{proof}
(i) is clear. To prove (ii), we observe that, letting 
\[  B = \frac{1}{2} \int dx dy \, \left(k(x,y) a_x^* a_y^* - \overline{k} (x,y) a_x a_y \right) , \]
we have, for any $f,g \in L^2 (\bR^3)$,  
\[ \begin{split} e^{-B} A (f,g) e^{B} & = A (f,g) + \int_0^1 d \lambda_1  \frac{d}{d\lambda_1} e^{-\lambda_1 B} A (f,g) e^{\lambda_1 B}\\
& = A (f,g) - \int_0^1 d \lambda_1  \, e^{-\lambda_1 B} [B,A (f,g)]
e^{\lambda_1 B}.
 \end{split} \]
Iterating, we find
\begin{equation}\label{eq:baker} \begin{split} e^{-B} A (f,g) e^{B} = \; &A (f,g) + \sum_{j=1}^n \frac{(-1)^j}{j!} \,\text{ad}_B^j (A (f,g)) \\ &+(-1)^{n+1}  \int_0^1 d\lambda_1 \int_0^{\lambda_1} d\lambda_2 \dots \int_0^{\lambda_{n}} d\lambda_{n+1} \, e^{-\lambda_{n+1} B}  \text{ad}_B^{n+1} (A (f,g)) \, e^{\lambda_{n+1} B}  \end{split}\end{equation}
where $\text{ad}_B^1 (C) = [B,C]$ and $\text{ad}_B^{n+1}(C) = [B,\text{ad}_B^n(C)]$. A simple computation shows that 
\[  \text{ad}_B^1 (A (f,g)) = [ B, A(f,g)] = - A \left(\left(  \begin{array}{ll}0 &  k  \\ \overline{k} & 0 \end{array} \right) \left( \begin{array}{l} f \\ g \end{array} \right) \right)  \]
and therefore that 
\[ \text{ad}^j_B (A (f,g)) = (-1)^j A  \left(\left(  \begin{array}{ll}0 &  k  \\ \overline{k} & 0 \end{array} \right)^j \left( \begin{array}{l} f \\ g \end{array} \right) \right). \]
We have 
\[ \begin{split} \left(  \begin{array}{ll}0 &  k  \\ \overline{k} & 0 \end{array} \right)^{2m} &= \left(  \begin{array}{ll} (k \overline{k})^m &  0 \\  0  & (\overline{k} k)^m  \end{array} \right)  \quad \text{and } \quad
\left(  \begin{array}{ll} 0 &  k  \\ \overline{k} & 0 \end{array} \right)^{2m+1} = \left(  \begin{array}{ll} 0 & (k \overline{k})^m k \\  (\overline{k} k)^m \overline{k} & 0 \end{array} \right) 
\end{split} \]
for every $m \in \bN$. Inserting all this in (\ref{eq:baker}), we obtain (\ref{eq:TAT}), if we can show that the error converges to zero. We claim, more precisely, that the error term on the r.h.s.\ of (\ref{eq:baker}) vanishes, as $n \to \infty$, when applied on the domain $D(\cN^{1/2})$. 
To prove this claim, we start by observing that
\begin{equation}\label{eq:NTN} \| (\cN+1)^{1/2} e^{-\lambda B} (\cN+1)^{-1/2} \|% = \| (\cN+1)^{-1/2} e^{\lambda B} (\cN+1)^{1/2} \|
\leq e^{|\lambda| \| k \|_2} \end{equation}
for every $\lambda \in \bR$. Assuming for example, that $n$ is odd, (\ref{eq:NTN}) implies that 
\[ \begin{split} \Big\| \int_0^1 d\lambda_1 \int_0^{\lambda_1} &d\lambda_2 \dots \int_0^{\lambda_{n}} d\lambda_{n+1} \, e^{-\lambda_{n+1} B}  \text{ad}_B^{n+1} (A (f,g)) \, e^{\lambda_{n+1} B} (\cN +1)^{-1/2} \Big\| \\ \leq \; & \frac{e^{\| k \|_2}}{(n+1)!}  \| A \left( (k\overline{k})^{(n+1)/2} f , (\overline{k} k)^{(n+1)/2} g \right) (\cN+1)^{-1/2} \| \\
\leq \; & (\| f \|_2 + \| g \|_2) \, e^{\| k \|_2} \, \frac{\| (k\overline{k})^{(n+1)/2} \|_2}{(n+1)!} \leq C \frac{\| k \|_2^n}{(n+1)!}\done 
\end{split}\]
which vanishes as $n \to \infty$. The case $n$ even can be treated similarly. To prove (\ref{eq:NTN}), we observe that 
\[ \begin{split} \frac{d}{d\lambda} \| (\cN+1)^{1/2} & e^{-\lambda B} \psi \|^2  \\
 =\; & \langle e^{-\lambda B} \psi, \left[ B , \cN \right] \, e^{-\lambda B} \psi \rangle \\
 = \; & -\int dx dy \,  k(x,y) \langle e^{-\lambda B} \psi, a_x^* a_y^* \, e^{-\lambda B} \psi \rangle - \int dx dy \, \overline{k} (x,y) \langle e^{-\lambda B} \psi, a_x a_y \, e^{-\lambda B} \psi \rangle \\ 
\leq \; & 2\int dx \, \| a_x e^{-\lambda B} \psi \| \, \| a^* (k(x,.)) e^{-\lambda B} \psi \| \\ 
\leq \; &2 \| k \|_2 \,  \| (\cN+1)^{1/2} e^{-\lambda B} \psi \|^2.
\end{split} \]
Gronwall's Lemma implies (\ref{eq:NTN}).   

To prove (iii), we notice that
\[  \| p (k) \|_2 = \left\| \sum_{n\geq 1} \frac{(k\overline{k})^n}{(2n)!} \right\|_2 \leq \sum_{n\geq 1} \frac{ \| (k\overline{k})^n \|_2}{(2n)!}  \leq \sum_{n\geq 1} \frac{ \| k \|^{2n}_2}{(2n)!} \leq e^{\| k \|_2} 
\]
where we used that, by Cauchy-Schwarz, for any two kernels $K_1, K_2 \in L^2 (\bR^3 \times \bR^3)$ 
\begin{equation}\label{eq:CS-k}\begin{split} \| K_1 K_2  \|_2^2 &= \int dx dy \, \left| \int dz \, K_1 (x,z) K_2  (z,y) \right|^2 \\ & = \int dx dy dz_1 dz_2 K_1 (x,z_1) \overline{K}_1 (x,z_2) K_2 (z_1, y) \overline{K}_2 (z_2,y) \\ &\leq \int dx dy dz_1 dz_2 \, |K_1 (x,z_1)|^2 \, |K_2 (z_2, y)|^2 \\ &= \| K_1 \|^2_2 \, \| K_2 \|^2_2. \end{split} \end{equation}
The bounds for $r(k)$ and $\text{sh} (k)$ can be proven similarly. 

To show (iv) we write, using the fact that the series for $p (k)$, $r (k) $ and $\text{sh} (k)$ are absolutely convergent, 
\begin{align*}
    \nabla_1 p (k) & =  \nabla_1 (k \overline{k}) \left[
    \sum_{n=1}^\infty \frac{1}{(2n)!} (k \overline{k})^{n-1} \right] \, , \\
    \nabla_1 r (k) & = \nabla_1 (k \overline{k}) \left[ \sum_{n=1}^\infty \frac{1}{(2n+1)!} (k
    \overline{k})^{n-1} k \right] .
%    \\ \nabla_1 \text{sh} (k) & = \nabla_1 k + \nabla_1 k \left[ \sum_{n=1}^\infty \frac{1}{(2n+1)!}
 %   (\overline{k} k)^n \right].
  \end{align*}
Applying (\ref{eq:CS-k}), we find the desired bounds. The bounds for the derivative $\nabla_2$ can be obtained similarly. 

Finally, to show (v), we remark that%, writing $p (k) = \sum_{n \geq 0} (k \overline{k})^n / (2n)!$, 
\[ \| \dot p (k) \|_2 \leq \sum_{n \geq 1} \frac{1}{(2n)!} 2n \| \dot k \|_2 \| k \|_2^{2n-1} \leq \| \dot k \|_2 \, e^{ \|k \|_2}.  \]
The bound for $\dot r (k)$ can be proven analogously. {F}rom the product rule, we also find that
\[ \begin{split} 
\| \nabla_1 \dot p (k) \|_2 & \leq \sum_{n=2}^\infty \frac{n-1}{(2n)!}   \, \| k  \|_2^{2(n-2)} 
\| \partial_t (k \overline{k}) \|_2  \, \| \nabla_1 (k \overline{k}) \|_2  + 
\sum_{n=1}^\infty \frac{1}{(2n)!} \, \| k \|_2^{2(n-1)}  \| \nabla_1 \partial_t (k \overline{k}) \|_2 \\
& \leq e^{\| k \|_2} \left( \| \dot k \|_2 \,  \| \nabla_1 (k \overline{k}) \|_2 + \| \nabla_1 (k \dot{\overline{k}}) \|_2  + \| \nabla_1 (\dot k\overline{k}) \|_2 \right). \end{split} \]
The other bounds are shown similarly.
\end{proof}

\section{Construction of the fluctuation dynamics}

In this section, we will construct an approximation for the full many-body evolution of an initial data of the form $W(\sqrt{N} \varphi) T(k_0) \psi$, as considered in Theorem \ref{thm:main}. 
Our approximation will consist of two parts. First of all, the evolution of an approximately coherent state will be approximated by a coherent state with an evolved one-particle wave function (later, we will take care of the correlation structure). 

For a given $\varphi \in H^1 (\bR^3)$, we define $\varphi^{(N)}_t$ as the solution of the modified time-dependent Gross-Pitaevskii equation
\begin{equation}\label{eq:mod-GP} i\partial_t \varphi_t^{(N)} = - \Delta \varphi_t^{(N)} + \left( N^3 f(N.) V(N.) * |\varphi^{(N)}_t|^2 \right) \varphi_t^{(N)} \end{equation}
where $f$ denotes the solution of the zero-energy scattering equation (\ref{eq:0en-0}). 
For technical reason, which will become clear later on, it is more convenient for us to work with the solution of the modified Gross-Pitaevskii equation (\ref{eq:mod-GP}), rather than directly with the solution of (\ref{eq:GP}). Since $N^3 f(Nx) V(Nx) \to 8 \pi a_0 \delta (x)$, the solution $\varphi_t^{(N)}$ converges towards the solution of (\ref{eq:GP}), as $N \to \infty$. This is proven, together with other important properties of the solutions of (\ref{eq:GP}) and (\ref{eq:mod-GP}), in the next proposition. 
\begin{proposition} \label{t:pdes}
Let $V \in L^1 \cap L^3 (\bR^3, (1+|x|^6) dx)$ be non-negative and spherically symmetric. Let $f$ denote the solution of the zero-energy scattering equation (\ref{eq:0en-0}), with boundary condition $f(x) \to 1$ as $|x| \to \infty$. Then, by Lemma \ref{lm:w} below, $0 \leq f \leq 1$ and therefore $Vf \geq 0$ with $Vf \in L^1 \cap L^3 (\bR^3, (1+|x|^6) dx)$. Let $\varphi \in H^1 (\bR^3)$, with 
$\| \varphi \|_2 = 1$.
\begin{itemize}
\item[(i)] Well-posedness. There exist unique global solutions $\varphi_. , \varphi_.^{(N)} \in C(\bR ; H^1 (\bR^3))$ of the Gross-Pitaevskii equation (\ref{eq:GP}) and, respectively, of the modified Gross-Pitaevskii equation (\ref{eq:mod-GP}), with initial data $\varphi$. These solutions are s.t. $\| \varphi_t \|_2 = \| \varphi_t^{(N)} \|_2 = 1$ for all $t \in \bR$. Moreover, there exists a constant $C > 0$ with 
\[  \| \varphi_t \|_{H^1} , \| \varphi_t^{(N)} \|_{H^1} \leq C \]
for all $t \in \bR$. 
\item[(ii)] Propagation of higher regularity. If we make the additional assumption that $\varphi \in H^n (\bR^3)$, for some integer $n \geq 2$, then $\varphi_t , \varphi_t^{(N)} \in H^n (\bR^3)$ for every $t \in \bR$. Moreover there exist constants $C>0$ depending on $\| \varphi \|_{H^n}$ and on $n$, and $K >0$, depending only on $\| \varphi \|_{H^1}$ and $n$, with
\begin{equation}\label{eq:hireg} \| \varphi_t \|_{H^n} , \| \varphi_t^{(N)} \|_{H^n} \leq C e^{K |t|} \end{equation}
for all $t \in \bR$. 
\item[(iii)] Regularity of time derivatives. Suppose $\varphi \in H^4 (\bR^3)$. Then there exists a constant $C>0$, depending on $\| \varphi \|_{H^4}$,  and $K > 0$, depending only on $\| \varphi \|_{H^1}$, such that
\[   \|\dot{\varphi}_t^{(N)} \|_{H^2} , \| \ddot{\varphi}_t^{(N)} \|_2 \leq C e^{K |t|} \]
for all $t \in \bR$.
\item[(iv)] Comparison of dynamics. Suppose now $\varphi \in H^2 (\bR^3)$. Then there exist constants $C,c_1,c_2 > 0$, depending on $\| \varphi \|_{H^2}$ ($c_2$ actually depends only on $\| \varphi \|_{H^1}$) such that 
\[ \| \varphi_t^{(N)} - \varphi_t \|_2 \leq \frac{C \exp (c_1 \exp (c_2 |t|))}{N} \,  \]
for all $t \in \bR$.
\end{itemize}
\end{proposition} 
The proof of Proposition \ref{t:pdes} can be found in Appendix \ref{s:pde}. 

\medskip

Using the solution $\varphi_t^{(N)}$ of (\ref{eq:mod-GP}), we are going to approximate the coherent part of the evolution. As explained in the introduction, however, this approximation is not good enough. The many-body evolution develops a singular correlation structure, which is completely absent in the evolved coherent state. As a consequence, fluctuations around the coherent approximation are too strong to be controlled. To solve this problem, we have to produce a better approximation of the many-body evolution, in particular an approximation which takes into account the short-scale correlation structure. To reach this goal, we are going to multiply the Weyl operator $W(\sqrt{N} \varphi_t^{(N)})$, which generates the coherent approximation to the many-body dynamics, by another unitary operator $T(k)$, having the form (\ref{eq:Tk-def}), obtained by taking the exponential of a quadratic expression in creation and annihilation operators. The kernel $k \in L^2 (\bR^3 \times \bR^3)$ has to be chosen so that $T(k)$ creates the correct correlations among the particles. Since correlations are, in good approximation, two-body effects, we can describe them through the solution $f$ of the zero-energy scattering equation (\ref{eq:0en-0}). We write 
%In this section, we will choose a kernel $k \in L^2 (\bR^3 \times \bR^3)$ and we will use it to %construct the unitary operator $T(k)$, as defined in (\ref{eq:Tk-def}). This operator should generate, %when applied to coherent states, the short-scale correlation structure developed by the many-body %evolution because of the singular interaction. It will thus allow us to give a better approximation of %the full dynamics; as a consequence, we will have a better chance to control the growth of the %fluctuations around it. 
%
%The correlation structure is generated by the two-body interactions. In good approximation, we can %model it using the solution of the zero-energy scattering equation. In the following, we will assume %that $V \geq 0$, and $V \in L^\infty (\bR^3) \cap L^1 (\bR^3)$. We define then $f$ to be the radial %solution of 
%\begin{equation}\label{eq:0enV}
%  \left( -\Delta + \frac{1}{2} V \right) f = 0
%\end{equation}
%with boundary condition $\lim_{|x|\to\infty} f(x) = 1$. We will write
\begin{equation}\label{eq:wdef}
  f (x) = 1 - w (x)
\end{equation}
with $\lim_{|x|\to\infty} w(x) = 0$. The scattering length of $V$ is defined
as
\[
  8 \pi a_0 = \int dx \, V(x) f(x).
\]
Equivalently, $a_0$ is given by 
\[
  a_0 = \lim_{|x| \to \infty} w(x)|x|.
\]
Note that, if $V$ has compact support inside $\{ x \in \bR^3 : |x| < R \}$, then $a_0 \leq R$ and $w(x) = a_0/|x|$ for $|x| > R$. In general, under our assumptions on $V$, one can prove the following properties of the function $w$.
\begin{lem}\label{lm:w}
Let $V \in L^1 \cap L^3(\bR^3, (1+|x|^6)dx)$ be spherically symmetric, with
$V \geq 0$. Denote by $f$ the solution of the zero-energy scattering
equation (\ref{eq:0en-0}) and let $w = 1 - f$. Then 
\[ 0 \leq w(x) \leq 1 \quad \text{for all } x \in \bR^3. \]
Moreover, there is a constant $C>0$ such that
\begin{equation}\label{eq:bdw} w(x) \leq \frac{C}{|x|+1} \qquad \text{and } \quad  |\nabla w (x)| \leq \frac{C}{|x|^2 + 1}. \end{equation}
\end{lem}
\begin{proof}
Standard arguments show that $0 \leq f (x) \leq 1$ holds for every $x \in \bR^3$ ($f(x) \leq 1$ follows from $V \geq 0$, because of the monotonic dependence of $f$ on the potential; see \cite[Appendix C]{LSSY}). This implies that $0 \leq w (x) \leq 1$ for all $x \in \bR^3$. From the zero energy scattering equation, we have $-\Delta w = Vf /2$. This implies that 
\[ w(x) = C \int dy \, \frac{1}{|x-y|} V (y) f (y) \qquad \text{and } \quad  \nabla w (x) = C \int dy \frac{x-y}{|x-y|^3} \, V (y) f (y)  \]
for an appropriate constant $C \in \bR$. 
%It follows that
%\[ \nabla w (x) = C \int dy \frac{x-y}{|x-y|^3} \, V (y) f (y). \]
Using $|x| \le  |x-y| + |y|$, the fact that $f \le 1$, and
the Hardy-Littlewood-Sobolev inequality, we find
\begin{align*}
  |(1+|x|)w(x)| & \le C \int dy \left( \frac{1}{|x-y|} + 1 +
\frac{|y|}{|x-y|} \right) V(y) f(y) \\& \le C (\| V \|_{3/2} + \| V \|_{L^1} + \| |y| V(y) \|_{3/2} )
\end{align*}
and, analogously, 
\begin{align*}
\left| (1+|x|^2) w (x) \right| &\leq C \int dy \left( \frac{1}{|x-y|^2} + 1 + \frac{|y|^2}{|x-y|^2} \right) V(y) f (y) 
\\ &\leq C \left( \| V \|_3 + \| V \|_1 + \| |y|^2 V \|_3 \right).
\end{align*}
The right hand side of the last two equations is bounded under the assumption $V \in L^1 \cap L^3 ((1+|x|^6) dx)$. 
\end{proof}

The zero-energy scattering equation for the rescaled potential $N^2 V (Nx)$ is then solved by $f(Nx)$. We define $w(Nx) = 1- f(Nx)$. Clearly 
\[ \lim_{|x| \to \infty} w(Nx) |x| = \frac{a_0}{N} ,\]
showing that the scattering length of $N^2 V(Nx)$ is $a_0/N$. Equivalently, this follows from $\int \di x\, N^2V(Nx)f(Nx) = 8\pi a_0/N$.\done

It follows immediately from Lemma \ref{lm:w} that $0 < w(Nx) < c$ for some $c <1$ and for all $x \in \bR^3$, and that there exists $C$ with 
\begin{equation}\label{eq:wN-bd} w(Nx) \leq \frac{C}{N|x| + 1} \qquad \text{and } \quad |\nabla_x w(Nx)| \leq C \frac{N}{N^2 |x|^2 + 1}. \end{equation}

%The function $f (Nx)$ describes the correlation structure developed by the many-body evolution, %varying on the microscopic scale $1/N$. 
%Variations on the macroscopic scale of order one are described (at least, this is what we want to %show) by the solution $\ph_t$ of the nonlinear Gross-Pitaevskii equation (\ref{eq:gp}). In order for the %unitary operator $T(k)$, as defined in (\ref{eq:Tk-def}), to produce the correct correlation structure, %we define the kernel
%\[ k (x,y) = - N w_N (x) \ph_t (x) \ph_t (y) \]

We will use the solution $f(Nx)$ of the scaled zero-energy scattering
equation to approximate the correlations among the particles, arising on the
microscopic scale. It is however important to keep in mind that these
correlations are also modulated on the macroscopic scale. The macroscopic
variation is described, or at least, this is what we expect, by the solution of the modified Gross-Pitaevskii equation (\ref{eq:mod-GP}). We define therefore the kernel
\begin{equation}\label{eq:kt} 
k_t (x,y)  = - N w (N (x-y)) \ph (x) \ph (y) \end{equation}
and the corresponding unitary operator
 \[ T (k_t) = \text{exp} \left( \frac{1}{2} \int dx dy \, \left( k_t (x,y) a_x^* a_y^* - \overline{k}_t (x,y) a_x a_y \right)\right). \]
In the next lemma, we collect several bounds for the kernel $k_t$ which will be useful in the following. 
\begin{lem} \label{l:kernels}
Let $\varphi_t^{(N)} \in H^1(\R^3)$ be the solution of (\ref{eq:mod-GP}) with initial data $\varphi \in H^1 (\bR^3)$. Let $w(Nx) = 1 - f(Nx)$, where $f$ solves the zero-energy scattering equation (\ref{eq:0en-0}). Let the kernel $k_t$ be defined as in (\ref{eq:kt}). \begin{enumerate}
\item \label{k} There exists a constant $C$, depending only on $\| \varphi^{(N)}_t \|_{H^1}$ such that
\[ \begin{split}  \| k_t \|_{2} &\le C , \\ \| \nabla_1 k_t \|_2 , \| \nabla_2 k_t
        \|_{2} &\le C \sqrt{N}, \\ \| \nabla_1 (k_t \overline{k}_t) \|_2 ,  \| \nabla_2 (k_t
        \overline{k}_t) \|_{2} &\le C .   \end{split}    \]
Defining $p(k_t)$ and $r (k_t)$ as in (\ref{eq:pr-def}), so that $\text{ch} (k_t) = 1 + p (k_t)$ and $\text{sh} (k_t) = k_t + r(k_t)$, it follows from Lemma \ref{l:bt}, part (iii) and (iv), that
      \[ \begin{split} \| p (k_t) \|_2 , \| r (k_t) \|_2 , \| \text{sh} (k_t) \|_2 &\leq C, \\
       \| \nabla_1 p(k_t) \|_2 , \| \nabla_2 p (k_t) \|_2  & \leq C, \\
        \| \nabla_1 r (k_t) \|_2 , \| \nabla_2 r (k_t) \|_2 & \leq C. \end{split} \] 
    \item \label{kr} For almost all $x,y \in \R^3$, we have the pointwise bounds 
      \[ \begin{split} 
        |k_t (x,y)| & \leq \min \left(N |\ph(x)| |\ph(y)| , \frac{1}{|x-y|} |\ph (x)| |\ph(y)| \right), \\
        |r (k_t) (x,y)| &\le C |\ph(x)| |\ph(y)|, \\ |p(k_t) (x,y)| &\leq C |\ph(x)| |\ph(y)| \, . \end{split} 
      \] 
%where the constant $C$ depends only on $\| \varphi_t^{(N)} \|_{H^1}$. 
    \item \label{sup} Suppose further that $\varphi \in H^2(\R^3)$. Then
      \[
         \sup_{x \in \R^3} \, \norm{k_t (.,x)}_{2}, \sup_{x \in \R^3} \, \norm{p (k_t) (.,x)}_{2},  \sup_{x \in \R^3} \, \norm{r (k_t) (.,x)}_{2} , \sup_{x \in \R^3}  \norm{\text{sh} (k_t) (.,x)}_{2} \leq C \norm{\ph}_{H^2}.
      \]
  \end{enumerate}
\end{lem}

We will also need bounds on the time derivative of the kernels $k_t$, $p(k_t)$, $r(k_t)$. These are collected in the following lemma.
\begin{lemma}\label{lm:dotk} 
Let $\varphi \in H^4(\R^3)$, and $\ph \in H^4 (\bR^3)$ be the solution of (\ref{eq:mod-GP}), with initial data $\varphi$. Let $w(Nx) = 1 - f(Nx)$, where $f$ is the solution of the zero-energy scattering equation (\ref{eq:0en-0}). Let the kernel $k_t$ be defined as in (\ref{eq:kt}), so that
\begin{equation}\label{eq:dtk} \dot{k}_t (x,y) = - N w (N (x-y)) \left( \dot{\varphi}_t^{(N)} (x) \varphi_t^{(N)} (y) + \varphi_t^{(N)} (x) \dot{\varphi}_t^{(N)} (y) \right). \end{equation}
Then there are constants $C,K >0$, where $C$ depends on the $\| \varphi \|_{H^4}$ and $K$ only on $\| \varphi \|_{H^1}$ such that the following bounds hold: 
\begin{itemize}
\item[(i)]   
 \[  \| \dot{k}_t \|_2 ,  \| \ddot k_t \|_2 , \| \dot{p} (k_t) \|_2 , \| \dot{r} (k_t) \|_2  \leq C e^{K|t|}, \]
\item[(ii)] 
\[ \| \nabla_1 \dot p (k) \|_2 ,  \| \nabla_2 \dot p (k) \|_2 , \| \nabla_1 \dot r (k) \|_2,  \| \nabla_2 \dot r (k) \|_2  \leq C  e^{K|t|}, \]
\item[(iii)] 
\[
\sup_x \, \norm{\dot k_t (., x)}_{2} ,  \sup_x \, \norm{\dot p (k_t) (., x)}_{2}, \sup_x \norm{\dot r (k_t) (., x)}_{2}, \sup_x \, \| \dot{\text{sh}} (k_t) (., x) \|_2  \leq C e^{K|t|}.
\] 
 \end{itemize}
 \end{lemma}
The proof of the last two lemmas can also be found in Appendix \ref{sec:kernels}. 

\medskip

As explained in the introduction, we are going to approximate the many-body evolution
\[ e^{-i \cH_N t} W(\sqrt{N} \varphi ) T(k_0) \psi \]
of an initial state which is almost coherent, but with the correct short-scale structure, by the Fock state $W(\sqrt{N} \varphi_t^{(N)}) T(k_t) \psi$, which is again almost coherent and has again the correct microscopic correlations. This leads us to the fluctuation dynamics, defined as the two-parameter group of unitary transformations 
\begin{equation}\label{eq:cU} \cU (t;s) = T^* (k_t) W^* (\sqrt{N} \varphi_t^{(N)}) e^{-i\cH_N (t-s)} W(\sqrt{N} \varphi_s^{(N)}) T(k_s) \end{equation}
where $\cU (s;s) = 1$ for all $s \in \bR$. 

The fluctuation dynamics satisfies the Schr\"odinger-type equation
\[ i\partial_t \cU (t;s) = \cL_N (t) \cU (t;s) \]
with the time-dependent generator
\[ \begin{split} \cL_N (t) = \; &T^* (k_t) \left[ i \partial_t W^* (\sqrt{N} \varphi_t^{(N)}) \right] W(\sqrt{N} \varphi_t^{(N)}) T (k_t) \\ &+ T^* (k_t) W^* (\sqrt{N} \varphi_t^{(N)}) \cH_N W (\sqrt{N} \varphi_t^{(N)}) T (k_t) + \left[ i \partial_t T^* (k_t) \right] T(k_t).  \end{split} \]

The next theorem, whose proof is deferred to Section \ref{sec:gen-fd}, is the main technical ingredient of this paper. It contains important estimates for the generator $\cL_N (t)$, which will be used in the next section to control the growth of the expectation of the number of particles operator with respect to the fluctuation dynamics $\cU (t;s)$. 
\begin{thm}\label{thm:L}
Define the time-dependent constant (of order $N$)
\begin{equation}\label{eq:CNt} \begin{split} 
C_N (t) = \; &\cL^{(0)}_{0,N} (t) + \int dx dy \, |\nabla_x \text{sh} (k_t) (y,x)|^2  \\
 &+  \int dx dy (N^3 V(N.) * |\varphi_t^{(N)}|^2) (x) \, |\text{sh} (k_t) (y,x)|^2 \\
 &+ \int dx dy dz \, N^3 V(N (x-y)) \varphi_t^{(N)} (x) \overline{\varphi}_t^{(N)} (y) \text{sh} (\overline{k}_t) (z,x)  \, \text{sh} (k_t) (z,y)  \\
 &+ \text{Re } \int dx dy dz \, N^3 V(N (x-y)) \varphi_t^{(N)} (x) \varphi_t^{(N)} (y) \text{sh} (\overline{k}_t) (z,x) \text{ch} (k_t) (z,y)  \\
&+ \int dx dy N^2 V(N (x-y)) \\ &\hspace{1.5cm}  \times \left[ \left|\int dz \, \text{sh} (\overline{k}_t) (z,x) \, \text{ch} (z,y) \right|^2  + \left| \int dz \, \text{sh} (\overline{k}_t) (z,x) \, \text{sh} (z,y) \right|^2 \right. \\ &\hspace{2cm}  \left. 
+  \int dz_1 dz_2 \, \text{sh} (\overline{k}_t) (z_1,x) \, \text{sh} (k_t) (z_1,y)  \text{sh} (\overline{k}_t) (z_2,x) \, \text{sh} (k_t) (z_2,y) \right] 
\end{split} \end{equation}
and let
\begin{equation}\label{eq:wtcL} \wt{\cL}_N (t) = \cL_N (t) - C_N (t). \end{equation}
Then we have, for some $K > 0$ depending only on $\norm{\varphi}_{H^1}$,
\begin{equation}\label{eq:thmL-1} \wt{\cL}_N (t) \geq \frac{1}{2} \cH_N - C \, \frac{\cN^2}{N}  - C \ech{\| \varphi_t^{(N)} \|_{H^2}^2}{\ekt}  \left(\cN + 1\right) \end{equation}
and
\begin{equation}\label{eq:thmL-1b} \wt{\cL}_N (t) \leq \frac{3}{2} \cH_N + C \, \frac{\cN^2}{N}  +  C \ech{\| \varphi_t^{(N)} \|_{H^2}^2}{\ekt}  \left(\cN + 1\right). \end{equation}
Moreover, 
\begin{equation}\label{eq:thmL-2} \pm \left[ \cN , \wt{\cL}_N (t) \right]  \leq \cH_N + C \frac{\cN^2}{N} + C \ech{\| \varphi_t^{(N)} \|_{H^2}^2}{\ekt} (\cN+1) \end{equation}
and 
\begin{equation}\label{eq:thmL-3}
\begin{split} 
\pm \dot{\wt{\cL}}_N (t)  \leq \; &\cH_N + \ech{C \| \varphi_t^{(N)} \|_{H^4}^2 \, \frac{\cN^2}{N}  \\ &+ C \left( \| \varphi_t^{(N)} \|_{H^4}  \| \varphi_t^{(N)} \|_{H^2} + \| \varphi_t^{(N)} \|_{H^2}^3 \right) \left(\cN+1 \right)}{C \ekt \left(\frac{\Ncal^2}{N} + \Ncal + 1 \right)}.  
\end{split}
\end{equation}
\end{thm}

\section{Growth of fluctuations}
\label{s:growthoffluct}
The goal of this section is to prove a bound, uniform in $N$, for the growth of the expectation of the number of particles operator with respect to the fluctuation dynamics. The properties of the generator $\cL_N (t)$ of the fluctuation dynamics, as established in Theorem \ref{thm:L}, play here a crucial role.

\begin{thm}\label{thm:N}
Suppose $\psi \in \cF$ (possibly depending on $N$) with $\| \psi \|_\cF  =1$ is such that \begin{equation}\label{eq:ass-grw} \begin{split} \left\langle \psi , \left( \frac{\cN^2}{N} + \cN + \cH_N \right) \psi \right\rangle &\leq C  \end{split} \end{equation}
for a constant $C>0$. Let $\varphi \in H^4 (\bR^3)$, and let $\varphi_t^{(N)}$ be the solution of the modified Gross-Pitaevskii equation (\ref{eq:mod-GP}) with initial data $\varphi$. Let $\cU (t;s)$ be the fluctuation dynamics defined in (\ref{eq:cU}). Then there exist constants $C,c_1,c_2 > 0$ such that
\[ \langle \psi , \cU^* (t;0) \cN \cU (t;0) \psi \rangle \leq C \exp (c_1 \exp (c_2 |t|)). \] 
\end{thm}

The strategy to prove Theorem \ref{thm:N} consists in applying Gronwall's
inequality. The derivative of the expectation of $\cN$ is given by the
expectation of the commutator $i[\cN , \cL_N (t)]$, where $\cL_N(t)$ is the
generator \eqref{eq:cLN} of the fluctuation dynamics. By (\ref{eq:thmL-2}),
this commutator  is bounded in terms of the energy, of $(\cN+ 1)$, and of
$\cN^2/N$ (the difference between $\cL_N (t)$ and the generator $\wt{\cL}_N
(t)$ appearing in (\ref{eq:thmL-2}) is a constant and hence does not
contribute to the commutator). The growth of the energy is controlled with
the help of (\ref{eq:thmL-3}). What remains to be done in order to apply
Gronwall's inequality is to bound the term $\cN^2 / N$. In the next
proposition, we show that the expectation of $\cN^2/N$ at time $t$ can be
controlled by its expectation at time $t=0$ (a harmless constant, by the
assumption (\ref{eq:ass-grw})) and by the expectation of $(\cN+1)$ (which
fits well in the scheme of Gronwall's inequality). 

\begin{proposition} \label{prop:apri}
Let the fluctuation dynamics $\cU (t;s)$ be defined as in (\ref{eq:cU}). Then there exists a constant $C > 0$ such that
 \[ \cU^* (t;0) \N^2 \, \cU (t;0) \le C \big( N  \, \cU^* (t;0) \N \, U (t;0) + N (\N+1) + (\N+1)^2 \big).
 \]
\end{proposition}

The next lemma is useful in the proof of Proposition \ref{prop:apri}.
\begin{lemma} \label{lm:TNT}
Let $k_t \in L^2(\R^3 \times \R^3)$ be as defined in (\ref{eq:kt}). Then there exists a constant $C$, depending only on $\| k_t \|_2$, such that
% Remark: If \norm{k_t} is bounded uniformly in t, so is C. 
\begin{align}
    T^* (k_t)  \, \N \, T (k_t) & \le C (\N+1), \label{eq:TNT} \\
%    T^* \K T & \le C (\K + \| \nabla_2 k \|_{L^2}^2), \label{TKT} \tag{ii} \\
    T^* (k_t) \, \N^2 \, T (k_t) & \le C (\N+1)^2 \label{eq:TN2T} % \\
%    \phi(f) \K \phi(f) & \le C_1 \big( \| f \|_{L^2}^2 \K (\N+1) + \| \nabla f \|_{L^2}^2 (\N + 1) \big). \label{fKf} \tag{iv}
  \end{align}
for all $t \in \bR$.
\end{lemma}

%\begin{rem}
%One can similarly prove that $T^* \K T \le C (\K + \| \nabla_2 k \|_{L^2}^2)$, where for $k(x,y) = -N %w_N(x-y) \varphi(x)\varphi(y)$ we have $\| \nabla_2 k \|_{L^2}^2$ of order $N$. Compared to %Proposition \ref{p:TNT} \eqref{TNT} this shows that the two-particle correlations implemented by $T$ %are small compared to the number of particles, but have a large influence on the kinetic energy. This %is due to the fact that short-scale fluctuations make for a large derivative of the wave-function. 
%\end{rem}

%\begin{cor}
%\label{cor:N2}
%If $\Psi \in \fock$ with a constant $C_1$ such that $\scal{\Psi}{\left( \frac{1}{N}\Ncal^2+\Ncal \right)%Psi} \leq C_1 \scal{\Psi}{\Psi}$, then there exists a constant $C_2$ such that
%\bd
%\frac{1}{N}\scal{U_t T^\ast_0 \Psi}{\Ncal^2 U_t T^\ast_0 \Psi} \leq C_2 \scal{U_t T^\ast_0 \Psi}
%{\left(\Ncal+1\right)U_t T^\ast_0 \Psi}.
%\ed
%\end{cor}
%\begin{proof}[Proof of Corollary \ref{cor:N2}]
%This is evident from Lemma \ref{l:ap} and Proposition \ref{p:TNT}.
%\end{proof}

%\begin{proof}[Proof of Proposition \ref{p:f}]
%  We give only an outline of the proof. Recall that $\phi(\varphi) =
%  a^*(\varphi) + a(\varphi)$. Then parts \eqref{f1}, \eqref{f2} and \eqref{f3}
%  follow easily by a brief calculation using parts \ref{l:W1} and \ref{l:W3}
%  of Lemma \ref{l:W}. Similarly, part \eqref{f4} follows from Lemma
%  \ref{l:bt}.
%\end{proof}

\begin{proof}
We use the decomposition $\text{ch} (k_t) = 1 + p(k_t)$ and the shorthand notation $c_x(z) = \text{ch} (k_t) (z,x)$, $p_x(z) = p(k_t) (z,x)$ and $s_x(z) = \text{sh} (k_t) (z,x)$. We have
 \[\begin{split} 
  \langle \psi, T^* (k_t)  \N T (k_t) \psi \rangle = \; & \int dx \, \langle  \psi, (a^* (c_x) + a (s_x)) (a (c_x) + a^* (s_x)) \psi \rangle \\ = \; &\int dx \, \| (a_x + a(p_x) + a^*(s_x)) \psi
  \|^2 \\  \leq \; &C  \int dx \, \| a_x \psi \|^2
    + \int dx \, \| a(p_x) \psi \|^2 + \int dx \, \| a^*(s_x) \psi \|^2 \\
    \leq \; &C ( 1 + \| p (k_t) \|_2^2 + \| \text{sh} (k_t) \|_2^2 ) \| (\N + 1)^{1/2}  \psi \|,
      \end{split}\]
and (\ref{eq:TNT}) follows by Lemma \ref{l:bt} (since $\| p (k_t) \|_2, \| \text{sh} (k_t) \|_2 \leq e^{\| k_t \|_2})$. To prove (\ref{eq:TN2T}), we observe that 
\begin{align*}
 & \langle \psi, T^* (k_t) \N^2 T (k_t) \psi \rangle = \int dxdy \, \langle \psi, T^* (k_t)  a_x^* a_x a_y^* a_y T (k_t) \psi \rangle \\
    & = \int dx \, \langle \psi, T^* (k_t)  a_x^* \N a_x T (k_t) \psi \rangle + \langle
    \psi, T^* (k_t)  \N T (k_t) \psi \rangle \\
    & = \int dx \, \langle (a(c_x) + a^*(s_x)) \psi, T^* (k)  \N T (k) (a (c_x)  + a^*(s_x)) \psi \rangle + \langle \psi, T^* (k_t) \N T (k_t) \psi \rangle.
  \end{align*}
  Then, applying \eqref{eq:TNT},  we obtain
  \begin{align*}
    & \langle \psi, T^* (k_t) \N^2 T (k_t) \psi \rangle \\
    & \leq C \int dx \, \| (\N+1)^{1/2} (a (c_x) +
    a^*(s_x)) \psi \|^2 + C \langle \psi, (\N+1) \psi \rangle \\
    & \leq C \int dx \, (\| a_x \N^{1/2} \psi \|^2 + \| a(p_x)
    \N^{1/2} \psi \|^2 + \| a^*(s_x) (\N+2)^{1/2} \psi \|^2 ) + C \langle
    \psi, (\N+1) \psi \rangle \\
    & \leq C(1 + \| p (k_t) \|_{2}^2 + \| \text{sh} (k_t)  \|_{2}^2) \langle \psi, (\N+1)^2
    \psi \rangle.
  \end{align*}
The bounds from Lemma \ref{l:bt} imply \eqref{eq:TN2T}.
\end{proof}

\begin{proof}[Proof of Proposition \ref{prop:apri}]
{F}rom Lemma \ref{lm:TNT}, we find
\begin{equation}\label{ep4-0}   \langle \psi , \cU^* (t;0) \cN^2 \cU (t;0) \psi \rangle \leq C  \langle \psi , \cU^* (t;0) T^* (k_t) \cN^2 T (k_t) \cU (t;0) \psi \rangle. \end{equation}
We now show how to bound the r.h.s.\ of the last equation. Using the definition of the fluctuation dynamics $\cU (t;0) = T^* (k_t) W^* (\sqrt{N} \varphi_t^{(N)}) e^{-i\cH_N t} W^* (\sqrt{N} \varphi) T (k_0)$, we find 
\begin{equation}\label{ep4}
\begin{split}
\langle \psi, &\cU^* (t;0) T^* (k_t) \N^2 T (k_t) \cU (t;0) \psi \rangle 
     \\  = \; & \langle \cN \, T (k_t) \cU (t;0) \psi,  W^* (\sqrt{N} \varphi_t^{(N)}) \, (\cN - \sqrt{N} \phi (\varphi_t^{(N)}) + N) e^{-i\cH_N t}  W (\sqrt{N} \varphi) T (k_0) \psi \rangle \\
      = \; & \langle \cN \, T (k_t) \cU (t;0) \psi,  W^* (\sqrt{N} \varphi_t^{(N)}) \,  \cN e^{-i \cH_N t}
      W (\sqrt{N} \varphi) T (k_0) \psi \rangle \\ & - \sqrt{N} \, \langle \cN \, T (k_t) \cU (t;0) \psi,  W^* (\sqrt{N} \varphi_t^{(N)} )\, \phi (\varphi_t^{(N)}) e^{-i \cH_N t}  W (\sqrt{N} \varphi) T (k_0) \psi \rangle \\
      &+  N  \langle \psi,\cU^* (t;0) T^* (k_t) \cN T (k_t) \cU (t;0) \psi \rangle 
      \end{split} \end{equation}
where we used the notation $\phi (f) = a( f) + a^* (f)$, and the property (\ref{eq:W3}) to show that 
\[ W^* (\sqrt{N} \varphi_t^{(N)}) \cN W (\sqrt{N} \varphi_t^{(N)}) = \cN - \sqrt{N} \phi (\varphi_t^{(N)})  +N.  \]      
In the first term on the r.h.s.\ of (\ref{ep4}), we use now the fact that $\cN$ commutes with $\cH_N$. In the second term, on the other hand, we move the factor $\phi (\varphi_t^{(N)})$ back to the left of the Weyl operator $W (\sqrt{N} \varphi_t^{(N)})$, using that \[ W^* (\sqrt{N} \varphi_t^{(N)}) \phi (\varphi_t^{(N)}) = \left(\phi (\varphi_t^{(N)}) + 2 \sqrt{N} \right) \, W^* (\sqrt{N} \varphi_t^{(N)})\,. \] We conclude that
\begin{equation}\begin{split}
\langle \psi , &\cU^* (t;0) T^* (k_t) \cN^2 T(k_t) \cU (t;0) \psi \rangle \\ = \; & \langle \cN \, T (k_t) \cU (t;0) \psi,  W^* (\sqrt{N} \varphi_t^{(N)}) \,  e^{-i \cH_N t} W (\sqrt{N} \varphi) (\cN + \sqrt{N} \phi (\varphi) + N) \, T (k_0) \psi \rangle \\ &-  \sqrt{N} \, \langle \cN \, T (k_t) \cU (t;0) \psi,  \left(\phi (\varphi_t^{(N)}) + 2\sqrt{N} \right)  W^* (\sqrt{N} \varphi_t^{(N)}) \, e^{-i \cH_N t}  W (\sqrt{N} \varphi) T (k_0) \psi \rangle \\
 &+ N  \langle \psi, \cU^* (t;0)  T^* (k_t) \cN T (k_t) \cU (t;0) \psi \rangle \\
 = \; &  \langle \cN \, T (k_t) \cU (t;0) \psi,  W^* (\sqrt{N} \varphi_t^{(N)}) \,  e^{-i \cH_N t} W (\sqrt{N} \varphi)  \cN \, T (k_0) \psi \rangle \\ &+\sqrt{N} \langle \cN \, T (k_t) \cU (t;0) \psi,  W^* (\sqrt{N} \varphi_t^{(N)}) \,  e^{-i \cH_N t} W (\sqrt{N} \varphi)  \phi (\varphi) \, T (k_0) \psi \rangle \\
 &-\sqrt{N} \, \langle \cN \, T (k_t) \cU (t;0) \psi,  \phi (\varphi_t^{(N)})  W^* (\sqrt{N} \varphi_t^{(N)}) \, e^{-i \cH_N t}  W (\sqrt{N} \varphi) T (k_0) \psi \rangle.
\end{split} \end{equation}
By Cauchy-Schwarz, we obtain 
\[ \begin{split} 
\langle \psi , \cU^* (t;0) &T^* (k_t) \cN^2  T(k_t) \cU (t;0) \psi \rangle \\ \leq \; & \| \cN \, T (k_t) \cU (t;0) \psi \|\\ & \hspace{1cm} \times  \left(  \| \cN \, T (k_0) \psi \rangle \| + \sqrt{N} \| \phi (\varphi) \, T (k_0) \psi \| + \sqrt{N} \|  \phi (\varphi_t^{(N)}) T (k_t) \cU (t;0)  \psi \| \right) \\ \leq \; & \frac{1}{2} \| \cN \, T (k_t) \cU (t;0) \psi \|^2 \\ &+ C \left( \|
\cN \, T (k_0) \psi \|^2 + N \| \phi (\varphi) \, T (k_0) \psi \|^2 + N \|  \phi (\varphi_t^{(N)}) T (k_t) \cU (t;0)  \psi \|^2 \right). 
\end{split} \]
Subtracting the first term appearing on the r.h.s., and using the bound $\| \phi (f) \psi \| \leq \| f \|_2 \| (\cN+1)^{1/2} \psi \|$, we find that
\[ \begin{split} \langle \psi , & \cU^* (t;0) T^* (k_t) \cN^2  T(k_t) \cU (t;0) \psi \rangle \\ & \leq C \left( \|
\cN \, T(k_0)  \psi \|^2 + N \| (\cN+1)^{1/2} T(k_0) \psi \|^2 + N \| (\cN+1)^{1/2} T(k_t) \cU (t;0)  \psi \|^2 \right).  \end{split} \]
By Lemma \ref{lm:TNT} we conclude that 
\[ \begin{split} \langle \psi , \cU^* (t;0) T^* (k_t)  \cN^2 &T (k_t) \cU (t;0) \psi \rangle \\ \leq \; &C N  \| \cN^{1/2}  \cU (t;0) \psi \|^2 + C \| \cN \psi \|^2 + CN \| (\cN+1)^{1/2} \psi \|^2 . \end{split} \]
With (\ref{ep4-0}), this concludes the proof of the proposition.
\end{proof}

We are now ready to show Theorem \ref{thm:N}.

\begin{proof}[Proof of Theorem \ref{thm:N}]
Let $C_N (t)$ be defined as in (\ref{eq:CNt}) and define
\[ \wt{\cU} (t;s) = e^{i \int_s^t C_N (\tau) d\tau} \cU (t;s). \]
Then $\wt{\cU} (t;s)$ satisfies the Schr\"odinger type equation
\[ i\partial_t \wt{\cU} (t;s) = \wt{\cL}_N (t) \wt{\cU} (t;s), \quad \text{with } \wt{\cU} (s;s) = 1 \]
for all $s \in \bR$, and with generator $\wt{\cL}_N (t) = \cL_N (t) - C_N (t)$, as defined in (\ref{eq:wtcL}). On the other hand, since the two evolutions only differ by a phase, we have
\[ \langle \psi , \cU^* (t;0) \cN \cU (t;0) \psi \rangle = \langle \psi , \wt{\cU}^* (t;0) \cN \wt{\cU} (t;0) \psi \rangle. \]
We now use the properties of $\wt{\cL}_N(t)$, as established in (\ref{eq:thmL-1}), (\ref{eq:thmL-1b}), (\ref{eq:thmL-2}) and (\ref{eq:thmL-3}).\ech{To this end, we observe that, by the assumption $\varphi \in H^4 (\bR^3)$ and by Proposition \ref{t:pdes}, there exist constants $C,K > 0$ such that \begin{equation}\label{eq:varphi-bds}  \| \varphi_t^{(N)} \|_{H^4}^2 , \left(\| \varphi_t^{(N)} \|_{H^4} \| \varphi_t^{(N)} \|_{H^2} + \| \varphi_t^{(N)} \|_{H^2}^3 \right) \leq C e^{K |t|} \end{equation}
uniformly in $N$, and for all $t \in \bR$.}{} Eq.\ (\ref{eq:thmL-1}) implies \ech{therefore}{} 
\begin{equation}\label{eq:en-bd} \cH_N \leq 2 \wt{\cL}_N (t) + C \frac{\cN^2}{N} + C e^{K |t|} (\cN + 1). \end{equation}
{F}rom Proposition \ref{prop:apri}, we conclude that there exist a constant $C_1$ (depending on $\langle \psi , (\cN^2/ N + \cN + 1) \psi \rangle$\ech{and on the constant $C$ in (\ref{eq:varphi-bds})}{}), such that
\begin{equation}\label{eq:grw-1} \begin{split} 
0 \leq \Big\langle \psi , \wt{\cU}^* (t;0) &\cH_N \wt{\cU} (t;0) \psi \Big\rangle \leq \left\langle \psi , \wt{\cU}^* (t;0) \left( 2 \wt{\cL}_N (t) + C_1 e^{K|t|} (\cN + 1) \right) \wt{\cU} (t;0) \psi \right\rangle. \end{split} \end{equation}
{F}rom (\ref{eq:thmL-2}), combined with (\ref{eq:en-bd}) and Proposition \ref{prop:apri}, there exists moreover a constant $C_2 > 0$ (depending on $\langle \psi , (\cN^2/ N + \cN + 1) \psi \rangle$\ech{ and on the constant $C$ in (\ref{eq:varphi-bds})}{}) such that
\[  
\frac{d}{dt} \left\langle \psi, \wt{\cU}^* (t;0) \cN \wt{\cU} (t;0) \psi \right\rangle  \leq  \left\langle \psi , \wt{\cU}^* (t;0)  \left(2 \wt{\cL}_N (t) +C_2 e^{K|t|} (\cN + 1) \right)  \wt{\cU} (t;0) \psi \right\rangle.
 \]
We now estimate the growth of the expectation\done{} of the generator $\wt{\cL}_N (t)$. Using (\ref{eq:thmL-3}) together with (\ref{eq:en-bd}) and Proposition \ref{prop:apri}, we conclude that there exists a constant $C_3 >0$ (again, depending on $\langle \psi , (\cN^2/ N + \cN + 1) \psi \rangle$\ech{ and on the constant $C$ in (\ref{eq:varphi-bds})}{}), with 
\[ \begin{split}
\frac{d}{dt} \left\langle \psi, \wt{\cU}^* (t;0) \wt{\cL}_N (t) \wt{\cU} (t;0) \psi \right\rangle = \; & \left\langle \psi, \wt{\cU}^* (t;0) \dot{\wt{\cL}}_N (t) \, \wt{\cU} (t;0) \psi \right\rangle \\ \leq \; & \left\langle \psi , \wt{\cU}^* (t;0) \left( 2 \wt{\cL}_N (t) + C_3 e^{K|t|} (\cN + 1) \right) \wt{\cU} (t;0) \psi \right\rangle. 
\end{split} \]
We now fix $D := \max (C_1 + 1, C_2 , C_3 , K)$. Then, we have
\[ \begin{split} 
\frac{d}{dt} \Big\langle \psi , &\wt{\cU}^* (t;0) \left( \wt{\cL}_N (t) + D e^{K \lvert t\rvert} (\cN + 1) \right) \wt{\cU} (t;0) \psi \Big\rangle \\  \leq \; & \left\langle \psi, \wt{\cU}^* (t;0) \, \left[ (2 + 2D e^{K|t|}) \wt{\cL}_N (t) + (D e^{K |t|} +D K e^{K |t|} + D^2 e^{2K |t|}) (\cN + 1) \right] \wt{\cU} (t;0) \psi \right\rangle\\  \leq \; &(2 + D e^{K|t|}) \, \left\langle \psi, \wt{\cU}^* (t;0) \, \left( \wt{\cL}_N (t) + D e^{K|t|} (\cN + 1) \right) \wt{\cU} (t;0) \psi \right \rangle \\ \leq \; & 2D e^{K |t|} \left\langle \psi, \wt{\cU}^* (t;0) \, \left( \wt{\cL}_N (t) + D e^{K|t|} (\cN + 1) \right) \wt{\cU} (t;0) \psi \right\rangle. \end{split} \]
By Gronwall, we conclude that
\[ \begin{split}  \Big\langle \psi , \wt{\cU}^* &(t;0) \left( \wt{\cL}_N (t) + 
D e^{K \lvert t\rvert} (\cN + 1) \right) \wt{\cU} (t;0) \psi \Big\rangle \\ \leq \; & C e^{2 \frac{D}{K} e^{K |t|}} \, \left\langle \psi , \left( \wt{\cL}_N (0) + (\cN + 1)\right) \psi \right\rangle \\ \leq \; &  C e^{2 \frac{D}{K} e^{K |t|}} \left\langle \psi , \left( \frac{3}{2} \cH_N+ C \frac{\cN^2}{N} + C (\cN+ 1) \right) \psi \right\rangle \end{split} \]
where in the last inequality, we used the upper bound (\ref{eq:thmL-1b}). {F}rom the assumption (\ref{eq:ass-grw}), we obtain 
\[  \Big\langle \psi , \wt{\cU}^* (t;0) \left( \wt{\cL}_N (t) + 
D e^{K \lvert t\rvert} (\cN + 1) \right) \wt{\cU} (t;0) \psi \Big\rangle
\leq C \exp (c_1 \exp (c_2 |t|)). \]
The claim now follows from (\ref{eq:en-bd}), because $D \geq C_1 +1$.
\end{proof}

\section{Proof of the main theorem}

Using the bounds established in Theorem \ref{thm:N}, we proceed now to prove our main result.
%, Theorem \ref{thm:main} concerning the evolution of Fock space states of the form $W(\sqrt{N} 
%\varphi) T(k_0) \psi$ and Theorem \ref{thm:Nstates} concerning the evolution of $N$ particle states %of the form $P_N W(\sqrt{N} \varphi) T(k_0) \psi$. 

\begin{proof}[Proof of Theorem \ref{thm:main}]
Let $\Psi_{N,t} = e^{-i\cH_N t} W(\sqrt{N} \varphi) T(k_0) \psi$. The one-particle reduced density of $\Psi_{N,t}$ has the integral kernel
\begin{equation}\label{eq:fin-G1} \Gamma_{N,t}^{(1)} (x;y) = \frac{1}{\langle \Psi_{N,t} , \cN \Psi_{N,t} \rangle} \langle \Psi_{N,t} , a_y^* a_x \Psi_{N,t} \rangle. \end{equation}
We start by computing the denominator. Since $\Hcal_N$ commutes with the number of particles operator, we find
\[ \begin{split} 
\langle \Psi_{N,t} , \cN \Psi_{N,t} \rangle = \; &\langle \psi,T^* (k_0) W^*
(\sqrt{N} \varphi) \cN W (\sqrt{N} \varphi)  T(k_0) \psi\rangle  \\ = \;
&\langle \psi, T^* (k_0) \left( \cN - \sqrt{N} \phi (\varphi) + N \right)
T(k_0) \psi \rangle \\ = \; &N + \left\langle \psi,  T^* (k_0) \left(\cN -
\sqrt{N} \phi(\varphi) \right) T(k_0) \psi \right\rangle. \end{split} \]
By Lemma \ref{lm:TNT}
\[  \langle \psi , T^* (k_0) \cN T(k_0) \psi \rangle \leq C \langle \psi , (\cN+1) \, \psi \rangle \leq C \]
and 
\[ \left| \langle \psi , T^*(k_0) \phi (\varphi) T(k_0) \psi \rangle \right| \leq C \langle \psi, T^* (k_0) (\cN+1)^{1/2} T(k_0) \psi \rangle \leq C \langle \psi , (\cN+1) \psi \rangle \leq C. \]
Hence, there exists $C>0$ with
\begin{equation}\label{eq:fin-G2} \left| \langle \Psi_{N,t} , \cN \Psi_{N,t} \rangle - N \right| \leq C N^{1/2}. \end{equation}
On the other hand, with $\varphi^{(N)}_t$ denoting the solution of the modified Gross-Pitaevskii equation (\ref{eq:mod-GP}), the numerator of (\ref{eq:fin-G1}) can be written as  
\[ \begin{split} 
\langle \Psi_{N,t}  &, a_y^* a_x \Psi_{N,t} \rangle \\ = \;  & \langle \psi, T(k_0) W (\sqrt{N} \varphi) e^{it \cH_N} a_y^* a_x e^{-i\cH_N t} W(\sqrt{N} \varphi) T(k_0) \psi \rangle \\  = \; & N \, \varphi_t (x) \overline{\varphi}_t (y) \\ &+  \sqrt{N} \, \varphi_t (x) \, \langle \psi, T^* (k_0) W^* (\sqrt{N} \varphi) e^{it \cH_N} \left(a_y^* - \sqrt{N} \overline{\varphi}^{(N)}_t (y)\right) e^{-i\cH_N t} W(\sqrt{N} \varphi) T(k_0) \psi \rangle \\  &+ \sqrt{N} \, \overline{\varphi}_t (y) \,  \langle \psi, T^* (k_0) W^* (\sqrt{N} \varphi) e^{it \cH_N} \left(a_x - \sqrt{N} \varphi^{(N)}_t (x) \right) e^{-i\cH_N t} W(\sqrt{N} \varphi) T(k_0) \psi \rangle \\ &+ \Big\langle \psi , T^* (k_0) W^* (\sqrt{N} \varphi) e^{it \cH_N} \left(a_y^* -\sqrt{N} \overline{\varphi}_t^{(N)} (y)\right) \\ &\hspace{5cm} \times \left(a_x - \sqrt{N} \varphi^{(N)}_t (x) \right) e^{-i\cH_N t} W(\sqrt{N} \varphi) T(k_0) \psi \Big\rangle. \end{split}\]
Recognising that
\[ \begin{split} (a_y^* - \sqrt{N} \overline{\varphi}^{(N)}_t (y)) &= W (\sqrt{N} \varphi^{(N)}_t) a_y^* W^* (\sqrt{N} \varphi_t^{(N)}) \\ (a_x - \sqrt{N} \varphi^{(N)}_t (x)) &= W (\sqrt{N} \varphi^{(N)}_t) a_x W^* (\sqrt{N} \varphi_t^{(N)}) \end{split} \]
we obtain
\[ \begin{split} \langle \Psi_{N,t} &, a_y^* a_x \Psi_{N,t} \rangle \\ = \; &N \varphi_t^{(N)} (x) \overline{\varphi}_t^{(N)} (y) \\ &+ \sqrt{N} \, \varphi^{(N)}_t (x) \, \langle \psi, T^* (k_0) W^* (\sqrt{N} \varphi) e^{it \cH_N} W (\sqrt{N} \varphi^{(N)}_t) \\ &\hspace{4cm} \times a_y^* W^* (\sqrt{N} \varphi^{(N)}_t) e^{-i\cH_N t} W(\sqrt{N} \varphi) T(k_0) \psi \rangle \\  &+ \sqrt{N} \, \overline{\varphi}^{(N)}_t (y)  \langle \psi, T^* (k_0) W^* (\sqrt{N} \varphi) e^{it \cH_N} W (\sqrt{N} \varphi^{(N)}_t) \\ &\hspace{4cm} \times a_x W^* (\sqrt{N} \varphi^{(N)}_t) e^{-i\cH_N t} W(\sqrt{N} \varphi) T(k_0) \psi \rangle \\ &+ \langle \psi , T^* (k_0) W^* (\sqrt{N} \varphi) e^{it \cH_N} W (\sqrt{N} \varphi^{(N)}_t) a_y^* a_x W^* (\sqrt{N} \varphi^{(N)}_t)  e^{-i\cH_N t} W(\sqrt{N} \varphi) T(k_0) \psi \rangle. \end{split}\]
Combining the last equation with (\ref{eq:fin-G2}) and inserting in (\ref{eq:fin-G1}), we have that
\[ \begin{split} 
\Big| \Gamma^{(1)}_{N,t} &(x;y) - \overline{\varphi}_t^{(N)} (y) \varphi_t^{(N)} (x) \Big| \\ \leq \; & \frac{C}{\sqrt{N}} |\varphi_t^{(N)} (x)| |\varphi_t^{(N)} (y)| \\ &+ \frac{1}{\sqrt{N}} |\varphi^{(N)}_t (y)| \, \| a_x \, W^* (\sqrt{N} \varphi^{(N)}_t)  e^{-i\cH_N t} W(\sqrt{N} \varphi) T(k_0) \psi  \| \\ &+ \frac{1}{\sqrt{N}} |\varphi^{(N)}_t (x)| \, \| a_y \, W^* (\sqrt{N} \varphi^{(N)}_t)  e^{-i\cH_N t} W(\sqrt{N} \varphi) T(k_0) \psi  \| \\ & + \frac{1}{N} 
 \| a_y \, W^* (\sqrt{N} \varphi^{(N)}_t)  e^{-i\cH_N t} W(\sqrt{N} \varphi) T(k_0) \psi  \|   \\ &\hspace{4cm} \times \| a_x \, W^* (\sqrt{N} \varphi^{(N)}_t)  e^{-i\cH_N t} W(\sqrt{N} \varphi) T(k_0) \psi  \|. \end{split} \] 
Taking the square and integrating over $x,y$, we find
\[ \left\| \Gamma^{(1)}_{N,t} - |\varphi^{(N)}_t \rangle \langle \varphi_t^{(N)}| \right\|_{\text{HS}} \leq \frac{C}{\sqrt{N}} \| (\cN + 1)^{1/2} \,  W^* (\sqrt{N} \varphi^{(N)}_t) e^{-i\cH_N t} W(\sqrt{N} \varphi) T(k_0) \psi \|^2. \]
Lemma \ref{lm:TNT} implies that
\[ \begin{split} \left\| \Gamma^{(1)}_{N,t} - |\varphi^{(N)}_t \rangle \langle \varphi_t^{(N)}| \right\|_{\text{HS}}  \leq \; &\frac{C}{\sqrt{N}} \,  \| (\cN + 1)^{1/2} \, T^* (k_t) W^* (\sqrt{N} \varphi^{(N)}_t) e^{-i\cH_N t} W(\sqrt{N} \varphi) T(k_0) \psi \|^2 \\ =\; &\frac{C}{\sqrt{N}} \| (\cN + 1)^{1/2} \, \cU (t;0) \psi \|^2  
\end{split}\]
with the fluctuation dynamics $\cU (t;s)$ defined in (\ref{eq:cU}).
{F}rom Theorem \ref{thm:N} we conclude that
\[ \left\| \Gamma^{(1)}_{N,t} - |\varphi^{(N)}_t \rangle \langle \varphi_t^{(N)}| \right\|_{\text{HS}} \leq \frac{C\exp (c_1 \exp (c_2 |t|))}{\sqrt{N}}. \]
Since $|\varphi^{(N)}_t\rangle \langle \varphi^{(N)}_t|$ is a rank-one projection, and $\Gamma_{N,t}^{(1)} \geq 0$, it follows that the difference $\Gamma_{N,t}^{(1)} - |\varphi^{(N)}_t \rangle \langle \varphi_t|$ can have only  one negative eigenvalue. Since the trace of $\Gamma_{N,t}^{(1)} - |\varphi_t \rangle \langle \varphi^{(N)}_t|$ vanishes, it must have one negative eigenvalue, with absolute value equal to the sum of all positive eigenvalues. As a consequence, the trace norm of the difference is controlled by the operator norm (given by the absolute value of the negative eigenvalue) and therefore also by the Hilbert-Schmidt norm. This shows that
\[  \tr \; \left| \Gamma^{(1)}_{N,t} - |\varphi^{(N)}_t \rangle \langle \varphi_t^{(N)}| \right| \leq \; \frac{C\exp (c_1 \exp (c_2 |t|))}{\sqrt{N}}. \]
Theorem \ref{thm:main} now follows because, if $\varphi_t$ denotes the solution of the Gross-Pitaevskii equation (\ref{eq:GP}), Proposition \ref{t:pdes} implies that
\[ \tr \; \left| |\varphi_t \rangle \langle \varphi_t| -  |\varphi^{(N)}_t \rangle \langle \varphi^{(N)}_t| \right| \leq 2 \| \varphi_t - \varphi_t^{(N)} \|_2 \leq \frac{C e^{c_1 |t|}}{N} \, .\qedhere \]
\end{proof}

\section{Key bounds on the generator of the fluctuation dynamics}
\label{sec:gen-fd}

In this section, we prove Theorem \ref{thm:L},  concerning the generator $\cL_N (t)$ of the fluctuation dynamics
\[  \cU (t;s) = T^* (k_t) W^* (\sqrt{N} \varphi_t^{(N)}) e^{-i\cH_N (t-s)} W(\sqrt{N} \varphi_s^{(N)}) T(k_s) \]
as defined in (\ref{eq:cU}). We write
\begin{equation}\label{eq:cLN} \begin{split}  \cL_N (t) = \; &T^* (k_t) \left[ i \partial_t W^* (\sqrt{N} \varphi_t^{(N)}) \right] W(\sqrt{N} \varphi_t^{(N)}) T (k_t) \\ &+ T^* (k_t) W^* (\sqrt{N} \varphi_t^{(N)}) \cH_N W (\sqrt{N} \varphi_t^{(N)}) T (k_t) + \left[ i \partial_t T^* (k_t) \right] T(k_t) \\
= \; &T^* (k_t) \cL^{(0)}_N (t) T (k_t) + \left[ i \partial_t T^* (k_t) \right] T(k_t) \end{split} \end{equation}
with 
\[ \cL_N^{(0)} (t) =  \left[ i \partial_t W^* (\sqrt{N} \varphi_t^{(N)}) \right] W(\sqrt{N} \varphi_t^{(N)}) + 
W^* (\sqrt{N} \varphi_t^{(N)}) \cH_N W (\sqrt{N} \varphi_t^{(N)}). \]
A simple computation shows that
%\begin{multline*}
\begin{displaymath}
 \left[ i \partial_t W^* (\sqrt{N} \varphi_t^{(N)}) \right] W(\sqrt{N}
 \varphi_t^{(N)})% \\
= - a(\sqrt{N} i \partial_t \varphi_t^{(N)}) - a^*
 (\sqrt{N} i \partial_t \varphi_t^{(N)}) - N \scal{\ph}{i\partial_t \ph}.\done
% + N \Im \scal{\varphi_t^{(N)}
% }{\partial_t \varphi_t^{(N)}}
\end{displaymath}
% \end{multline*}
 On the other hand, using (\ref{eq:W3}), we find
 \[ \begin{split} W^* (\sqrt{N} &\varphi_t^{(N)}) \cH_N W (\sqrt{N} \varphi_t^{(N)})  \\
 %= \; & \int dx \, \nabla_x (a_x^* + \sqrt{N} \overline{\varphi}_t^{(N)} (x) ) \nabla_x (a_x + \sqrt{N} 
 %\varphi_t^{(N)} (x)) \\ &+ \int dx dy \, N^2 V (N(x-y)) \,  (a_x^* + \sqrt{N} \overline{\varphi}_t^{(N)} (x) ) %(a_y^* + \sqrt{N} \overline{\varphi}_t^{(N)} (y) ) \\ & \hspace{6cm} \times 
%(a_y + \sqrt{N} \varphi_t^{(N)} (y)) (a_x + \sqrt{N} \varphi_t^{(N)} (x)) \\
= \; & N \left[ \| \nabla \varphi_t^{(N)} \|_2^2 + \frac{1}{2} \int dx \, (N^3 V (N .) * |\varphi_t^{(N)}|^2) (x) |\varphi_t^{(N)} (x)|^2 \right] \\
&+ \sqrt{N} \left[  a^*\left( (N^3 V (N.) *| \varphi_t^{(N)}|^2) \varphi_t^{(N)} \right) + a\left( (N^3 V (N.) *| \varphi_t^{(N)}|^2) \varphi_t^{(N)} \right) \right] \\
&+ \left[ \int dx \, \nabla_x a_x^* \nabla_x a_x + \int dx \, (N^3 V (N.) * |\varphi_t^{(N)}|^2) (x) a_x^* a_x \right. \\ 
&\hspace{.5cm}+ \int dx dy N^3 V(N (x-y)) \varphi_t^{(N)} (x) \overline{\varphi}_t^{(N)} (y) a_x^* a_y \\
&\hspace{.5cm}+ \left.  \frac{1}{2} \int dx dy N^3 V(N (x-y)) \left( \varphi_t^{(N)} (x) \varphi_t^{(N)} (y) a_x^* a_y^* + 
\overline{\varphi}_t^{(N)} (x) \overline{\varphi}_t^{(N)} (y) a_x a_y \right) \right] \\
&+ \frac{1}{\sqrt{N}} \int dx dy \, N^3 V(N(x-y)) a_x^* \left( \varphi_t^{(N)} (y) a_y^* + \overline{\varphi}_t^{(N)} (y) a_y \right) a_x \\
&+\frac{1}{2N} \int dx dy N^3 V(N (x-y)) a_x^* a_y^* a_y a_x. \end{split} \]
Combining the last two equations and using (\ref{eq:mod-GP}), we conclude that
\[ \begin{split} \cL_N^{(0)} (t) = \; & %N \, \text{Im } \left\langle
%\varphi_t^{(N)} , \left(N^3 (f(N.) - 1/2) V(N.) * |\varphi_t^{(N)}|^2
%\right] \varphi_t^{(N)} \right\rangle \\
%N \left[ \| \nabla \varphi_t^{(N)} \|_2^2 + \frac{1}{2} \int dx \, (N^3 V (N .) * |\varphi_t^{(N)}|^2) (x) |\varphi_t^{(N)} (x)|^2 \right] \\
N \int \di x\left( N^3 V(N.)\Big(\frac{1}{2}-f(N.)\Big)\ast \lvert \ph\rvert^2 \right)(x) \lvert \ph(x)\rvert^2 \done\\
&+ \sqrt{N}  \left[  a^*\left( (N^3
w(N.) V (N.) *| \varphi_t^{(N)}|^2) \varphi_t^{(N)} \right)
% \right. \\ & \quad \qquad \left. 
+ a\left( (N^3 w (N.)V (N.) *| \varphi_t^{(N)}|^2) \varphi_t^{(N)} \right) \right] \\
&+ \left[ \int dx \, \nabla_x a_x^* \nabla_x a_x + \int dx \, (N^3 V (N.) * |\varphi_t^{(N)}|^2) (x) a_x^* a_x \right. \\ 
&\hspace{.5cm}+ \int dx dy N^3 V(N (x-y)) \varphi_t^{(N)} (x) \overline{\varphi}_t^{(N)} (y) a_x^* a_y \\
&\hspace{.5cm}+ \left.  \frac{1}{2} \int dx dy N^3 V(N (x-y)) \left( \varphi_t^{(N)} (x) \varphi_t^{(N)} (y) a_x^* a_y^* + 
\overline{\varphi}_t^{(N)} (x) \overline{\varphi}_t^{(N)} (y) a_x a_y \right) \right] \\
&+ \frac{1}{\sqrt{N}}  \int dx dy \, N^3 V(N(x-y)) a_x^* \left( \varphi_t^{(N)} (y) a_y^* + \overline{\varphi}_t^{(N)} (y) a_y \right) a_x \\
&+\frac{1}{2N} \int dx dy N^3 V(N (x-y)) a_x^* a_y^* a_y a_x \\
=: \; & \cL_{0,N}^{(0)} (t) + \cL_{1,N}^{(0)} (t) + \cL_{2,N}^{(0)} (t) +  
\cL_{3,N}^{(0)} (t) + \cL_{4,N}^{(0)} 
\end{split} \]
where $\cL^{(0)}_{j,N} (t)$, for $j=0, \dots , 4$, is the part of $\cL_N^{(0)} (t)$ containing $j$ creation and annihilation operators. Recall here that $w (x) = 1 - f (x)$, as defined in (\ref{eq:wdef}). 

{F}rom (\ref{eq:cLN}), we find that the generator of the fluctuation dynamics is given by
\begin{equation}\label{eq:dec-cLN} \cL_N (t) =  \cL_{0,N}^{(0)} (t) +
\sum_{j=1}^4 T^* (k_t) \cL_{j,N}^{(0)} (t) T (k_t) +\left[ i \partial_t T^*
(k_t) \right] T(k_t). \end{equation}

In the next subsections, we study separately the different terms on the r.h.s.\ of (\ref{eq:dec-cLN}). The final goal of this analysis, a proof of Theorem \ref{thm:L}, will be reached in Subsection \ref{sub:thmL}.

\medskip

{\it Notation.} In the rest of this section, we will use the shorthand notation \begin{equation}\label{eq:def-csprx} 
c_x (y) = \text{ch} (k_t) (y,x), \;  s_x (y) = \text{sh} (k_t) (y,x), \; p_x (y) = p (k_t) (y,x), \; 
r_x (y) = r (k_t) (y,x). \end{equation}
Moreover, $\| p \|_2 , \| r \|_2, \|\text{sh} \|_2$ will denote the $L^2$-norms of the kernels $p (k_t) (x,y)$, $r (k_t) (x,y)$, and $\text{sh} (k_t) (x,y)$ over $\bR^3 \times \bR^3$ (in other words, they denote the Hilbert-Schmidt norms of the corresponding operators). The norms $\| p_x \|_2, \| r_x \|_2, \| \text{sh}_x \|_2$, on the other hand, indicate norms over $\bR^3$. Finally, the notation $\langle . , . \rangle$ will denote the $L^2$ inner product. We will abbreviate $T(k_t)$ by $T$.

\subsection{Analysis of $T^* \cL^{(0)}_{1,N} (t) T$}

Conjugating the linear term $\cL^{(0)}_{N,1} (t)$ produces again linear terms. {F}rom Lemma \ref{l:bt}, we obtain
\begin{equation}\label{eq:TL1T} \begin{split} T^* & \cL^{(0)}_{1,N} (t) T \\ = \; & \sqrt{N} \int dx dy \, N^3 V (N (x-y)) w (N (x-y)) |\varphi_t^{(N)} (y)|^2 \left( \varphi^{(N)}_t (x) T^* a^*_x T + \overline{\varphi}_t^{(N)} (x) T^* a_x T \right) \\ 
= \; &\sqrt{N} \int dx dy \, N^3 V (N (x-y)) w (N (x-y)) |\varphi_t^{(N)} (y)|^2 \, \varphi^{(N)}_t (x)  \, ( a^*(c_x) + a(s_x)) \\ &+ \sqrt{N} \int dx dy \, N^3 V (N (x-y)) w (N (x-y)) |\varphi_t^{(N)} (y)|^2 \, \overline{\varphi}_t (x)  \, ( a (c_x) + a^* (s_x)). \end{split} \end{equation}
These terms are potentially dangerous because they are large (of order $\sqrt{N}$) and do not commute with the number of particles. We will see however that they cancel with contributions arising from the cubic part $\cL_{3,N}^{(0)} (t)$.

\subsection{Analysis of $T^* \cL_{2,N}^{(0)} (t) T$}

We write $\cL_{2,N}^{(0)} (t) = \cK + \wt{\cL}_{2,N}^{(0)} (t)$, with \[ \cK = \int dx \, \nabla_x a^*_x \nabla_x a_x \] being the kinetic energy operator, and we consider separately the effects of $\cK$ 
and of the other quadratic terms collected in $\wt{\cL}_{2,N}^{(0)} (t)$. 

\subsubsection{Properties of $T^* \cK T$}

We have
\[ \begin{split} T^* \cK T = \; & \int dx \, \nabla_x (a^* (c_x) + a(s_x)) \, \nabla_x (a(c_x) + a^* (s_x)) \\
= \; & \int dx \nabla_x a^* (c_x) \nabla_x a(c_x) + \int dx \nabla_x a^* (c_x) \nabla_x a^* (s_x) \\ & + \int dx \nabla_x a(s_x) \nabla_x a(c_x) + \int dx \nabla a^* (s_x) \nabla_x a (s_x) + \int dx  \| \nabla_x s_x \|_2^2.
\end{split} \]
Following (\ref{eq:pr-def}), we decompose  $c_x (y) = \delta (x-y) + p_x (y)$ and $s_x (y) = k (x,y) + r_x (y)$. Hence
\begin{equation}\label{eq:TKT} 
\begin{split} 
T^* \cK T = \; & \cK + \int dx dy \, | \nabla_x \text{sh}_{k_t} (y,x)|^2 \\  & + \int dx \nabla_x a^*_x a (\nabla_x p_x) + \int dx a^* (\nabla_x p_x) \nabla_x a_x + \int dx a^* (\nabla_x p_x) a (\nabla_x p_x)  \\ &+ \int dx \nabla_x a_x^* a^* (\nabla_x k_x)  + \int dx a^* (\nabla_x p_x) a^* (\nabla_x k_x) +  \int dx \nabla_x a_x^* a^* (\nabla_x r_x)  \\ & + \int dx a^* (\nabla_x p_x) a^* (\nabla_x r_x) + \int dx \, a(\nabla_x k_x) \nabla_x a_x + \int dx a(\nabla_x r_x) \nabla_x a_x  \\ & + \int dx a(\nabla_x k_x) a (\nabla_x p_x) + \int dx a(\nabla_x r_x) a(\nabla_x p_x) + \int dx \, a^* (\nabla_x k_x) a (\nabla_x k_x)  \\ & + \int dx \, a^* (\nabla_x r_x) a (\nabla_x k_x) + \int dx \, a^* (\nabla_x k_x) a(\nabla_x r_x) + \int dx \, a^* (\nabla_x r_x) a(\nabla_x r_x). 
\end{split} 
\end{equation}

The properties of $T^* \cK T$ are summarized in the next proposition. 
\begin{proposition}\label{prop:TKT}
We have
\begin{equation}\label{eq:pr-K}
\begin{split}
T^* \cK T = \; &  \int dx dy | \nabla_x \text{sh}_{k_t} (y,x) |^2 + \cK  \\
&+ N^3 \int dx dy (\Delta w) (N (x-y)) \, \left(\varphi_t^{(N)} (x) \varphi_t^{(N)} (y)  a_x^* a^*_y + \overline{\varphi}_t^{(N)} (x) \overline{\varphi}_t^{(N)} (y)  a_x a_y  \right) \\
&+ \cE_K (t) 
\end{split} \end{equation}
where the error $\cE_K (t)$ is an operator such that\done{} for every $\delta > 0$ there exists a constant $C_\delta > 0$ with 
\begin{equation}\label{eq:TKT-cl} 
\begin{split} 
\pm \, \cE_K (t) &\leq \delta \cK + C_\delta \ech{\| \varphi_t^{(N)}
\|_{H^2}}{\ekt} \, (\cN + 1), \\
\pm \, \left[ \cN , \cE_K (t) \right] &\leq \delta \cK + C_\delta \ech{\, \|
\varphi_t^{(N)} \|_{H^2}}{\ekt} \, (\cN + 1), \\
\pm \dot{\cE}_K (t) &\leq \delta \cK + C_\delta \ech{\| \varphi_t^{(N)} \|^2_{H^3}}{\ekt} \,  (\cN+1). 
\end{split}
\end{equation}
\end{proposition}

To prove Proposition \ref{prop:TKT}, we will use the next lemma. 
\begin{lemma}\label{lm:K}
Let $j_1, j_2 \in L^2 (\bR^3 \times \bR^3)$. Let $j_{i,x} (z) := j_i (z,x)$ for $i =1,2$. Then we have
\begin{equation}\label{eq:lmK-1}
\left| \int dx \langle \psi , a^\sharp (j_{1,x}) a^\sharp (j_{2,x}) \psi \rangle \right| \leq C \| j_1 \|_2 \| j_2 \|_2  \, \| (\cN+1)^{1/2} \psi \|^2. 
\end{equation}
Here and in the following $a^\sharp$ can be either the annihilation operator $a$ or the creation operator $a^*$. Moreover, for every $\delta > 0$, there exists $C_\delta > 0$ such that 
\begin{equation}\label{eq:lmK-2}
\begin{split}
\left| \int dx \langle \psi , \nabla_x a_x^* a^\sharp (j_{1,x}) \psi \rangle \right| &\leq\delta \cK +  C_\delta \| j_1 \|^2_2  \| (\cN+1)^{1/2} \psi \|^2  \quad \text{and, by conjugation} \\
\left| \int dx \langle \psi , a^\sharp (j_{1,x}) \nabla_x a_x  \psi \rangle \right| &\leq \delta \cK +  C_\delta \| j_1 \|_2^2 \, \| (\cN+1)^{1/2} \psi \|^2. \end{split} \end{equation}
Terms where the argument of a creation and/or annihilation operator is the kernel $\nabla_x k_x$ (whose $L^2$-norm diverges as $N \to \infty$) can be handled with the following bounds. For every $\delta > 0$ there exists $C_\delta >0$ s.t. 
\begin{equation}\label{eq:lmK-3}
\begin{split}
\left| \int dx \langle \psi , a^* (\nabla_x k_x) a^\sharp (j_{1,x}) \psi \rangle \right| & \leq \delta \cK + C_\delta  (1+\| j_1 \|^2_2) \| (\cN+1)^{1/2} \psi \|^2 \quad \text{and, by conjugation} \\
\left| \int dx \langle \psi , a^\sharp (j_{1,x}) a (\nabla_x k_x)  \psi \rangle \right| & \leq \delta \cK + C_\delta (1+\| j_1 \|^2_2) \| (\cN+1)^{1/2} \psi \|^2.
\end{split}
\end{equation}
Moreover, we have
\begin{equation}\label{eq:lmK-4}
\left| \int dx \langle \psi , a^* (\nabla_x k_x) a (\nabla_x k_x) \psi
\rangle \right| \leq C \| \cN^{1/2} \psi \|^2.
\end{equation}
To control the time derivative of $\cE_K (t)$, we will also use the following bounds. For every $\delta > 0$ there exists $C_\delta > 0$ such that 
\begin{equation}\label{eq:lmK-5} \left| \int dx \langle \psi , a^* (\nabla_x \dot{k}_x) a^\sharp (j_{1,x}) \psi \rangle \right| \leq \delta \cK + C_\delta \ech{\, \| \varphi_t^{(N)} \|_{H^3}^2}{\ekt}(1+\| j_1 \|^2_2) \, \| (\cN+1)^{1/2} \psi \|^2 .
\end{equation}
Moreover,
\begin{equation}\label{eq:lmK-6} 
\left| \int dx \langle \psi , a^* (\nabla_x \dot{k}_x) a (\nabla_x k_x) \psi \rangle \right| \leq C \ech{\| \varphi_t^{(N)} \|_{H^3}^2}{\ekt}  \| \cN^{1/2} \psi \|^2. \end{equation}
\end{lemma}
\begin{proof}
To prove (\ref{eq:lmK-1}), we compute
\[ \begin{split}  \left| \int dx \langle \psi , a^\sharp (j_{1,x}) a^\sharp (j_{2,x}) \psi \rangle \right|  \leq  \; &\int dx \| a^\sharp (j_{1,x}) \psi \| \, \| a^\sharp (j_{2,x}) \psi \| \\ \leq \; &\left(\int dx \| j_{1,x} \|_2 \| j_{2,x} \|_2\right) \| (\cN+1)^{1/2} \psi \|^2 \\ \leq \; &\| j_1 \|_2 \| j_2 \|_2  \, \| (\cN+1)^{1/2} \psi \|^2. \end{split} \]
Eq.\ (\ref{eq:lmK-2}), on the other hand, follows by
\begin{equation} \begin{split}
\left| \int dx \langle \psi, \nabla_x a_x^* a^\sharp (j_{1,x}) \psi \rangle \right| \leq \; & \int dx \| \nabla_x a_x \psi \| \, \| a^\sharp (j_{1,x}) \psi \| \\ \leq \; &\delta \langle \psi, \cK \psi \rangle + C_\delta  \int dx \| j_{1,x} \|_2^2 \, \| (\cN+1)^{1/2} \psi \|^2 \\ \leq \; & \delta \langle \psi , \cK \psi \rangle + C_\delta \| j_1 \|^2_2 \, \| (\cN+1)^{1/2} \psi \|^2. \end{split} \end{equation}
To show (\ref{eq:lmK-3}), we need to integrate by parts. We write
\[ \int dx a^* (\nabla_x k_x) a^\sharp (j_{1,x})  = \int dx dy \, \nabla_x k (y,x) \, a^*_y  a^\sharp (j_{1,x}) \]
and we observe that
\[ \nabla_x k(y,x) = - \nabla_y k(y,x) - N w (N (y-x)) \left( \nabla
\varphi^{(N)}_t (x) \varphi^{(N)}_t (y) + \varphi_t^{(N)} (x) \nabla
\varphi_t^{(N)} (y) \right). \]
Hence
\[\begin{split} 
 \int dx a^* (\nabla_x k_x) a^\sharp (j_{1,x}) = \; & \int dx dy k (x,y) \, \nabla_y a^*_y  a^\sharp (j_{1,x}) 
\\ &- \int dx dy N w (N (x-y)) \nabla \varphi_t^{(N)} (x) \varphi_t^{(N)} (y) \, a^*_y  \, a^\sharp(j_{1,x})
\\& - \int dx dy N w (N (x-y)) \nabla \varphi_t^{(N)} (y) \varphi_t^{(N)} (x) \, a^*_y  \, a^\sharp(j_{1,x}).\end{split} \]
This implies, using Lemma \ref{lm:w} to bound $Nw(N(x-y))$,
\[ \begin{split}  \Big|  \int &dx \langle \psi,   a^* (\nabla_x k_x) a^\sharp (j_{1,x}) \psi \rangle \Big| \\ \leq \; & \int dx dy \, |k(x,y)| \| \nabla_y a_y \psi \| \, \| a^\sharp (j_{1,x}) \psi \| \\ &+ C \int dx dy \,  \frac{1}{|x-y|}  \left( |\nabla \varphi_t^{(N)} (x)| |\varphi^{(N)}_t (y)| +|\nabla \varphi_t^{(N)} (y)| |\varphi^{(N)}_t (x)| \right) \| a_y \psi \| \, \| a^\sharp (j_{1,x}) \psi \|  \\ 
\leq \; &\delta \int dx dy |\varphi_t^{(N)} (x)|^2 \| \nabla_y a_y \psi \|^2 + C_\delta \int dx dy \frac{1}{|x-y|^2} |\varphi^{(N)}_t (y)|^2 \|j_{1,x}\|_2^2 \| |(\cN+1)^{1/2} \psi \|^2 \\ &
+C \int dx dy \frac{1}{|x-y|^2} |\varphi^{(N)}_t (y)|^2 \|j_{1,x}\|_2^2 \| |(\cN+1)^{1/2} \psi \|^2 + 
C \int dx dy | \nabla \varphi_t^{(N)} (x)|^2 \| a_y \psi \|^2 \\
&+C\int dx dy \frac{1}{|x-y|^2} |\varphi_t^{(N)} (x)|^2 \| a_y \psi \|^2 
+ C \int dx dy |\nabla \varphi_t^{(N)} (y)|^2  \|j_{1,x} \|_2^2 \|
(\cN+1)^{1/2} \psi \|^2. \end{split} \]
Using Hardy's inequality, we conclude that for every $\delta >0$ there exists $C_\delta > 0$ (depending on $ \| j_1 \|_2$, $\delta$, $\| \varphi_t^{(N)} \|_{H^1}$) such that 
\[ \left| \langle \psi , \int dx\, a^* (\nabla_x k_x) a^\sharp (j_{1,x}) \psi \rangle \right| \leq \delta \langle \psi , \cK \psi \rangle + C_\delta (1+  \| j_1 \|_2^2 ) \, \langle \psi , (\cN + 1) \psi \rangle. \]
To show (\ref{eq:lmK-4}), we write
\[  \begin{split} \int dx  \, a^* (\nabla_x k_x) a (\nabla_x k_x) & = \int dx dy_1 dy_2 \, \nabla_x k (y_1,x) \nabla_x \overline{k} (y_2,x) \, a_{y_1}^* a_{y_2} \\ & =: \int dy_1 dy_2 \; g (y_1, y_2) a_{y_1}^* a_{y_2}. \end{split} \]
We have\done
\[ \begin{split} \Big| \int dy_1 dy_2\, g (y_1 , y_2) \langle &\psi , a_{y_1}^* a_{y_2} \psi \rangle \Big| \\ \leq \; & \int dy_1 dy_2 |g (y_1 , y_2)| \| a_{y_1} \psi \| \, \|  a_{y_2} \psi \| \\ \leq  \; & \left(\int dy_1 dy_2 |g (y_1 , y_2)|^2 \right)^{1/2}  \left( \int dy_1 dy_2 \, \| a_{y_1} \psi \|^2 \, \|  a_{y_2} \psi \|^2 \right)^{1/2} \\ = \; & \| g \|_2 \, \| \cN^{1/2} \psi \|^2. \end{split} \] 
Moreover, from the definition (\ref{eq:kt}) of the kernel $k_t$ and from the bounds of Lemma \ref{lm:w}, we find\done 
\[ \begin{split} \| g \|_2^2 \leq \; &C \int dy_1 dy_2 dx_1 dx_2  \frac{|\varphi_t^{(N)} (x_1)|^2 |\varphi_t^{(N)} (x_2)|^2 |\varphi_t^{(N)} (y_1)|^2 |\varphi_t^{(N)} (y_2)|^2}{|x_1 - y_1|^2 |x_1 - y_2|^2 |x_2-y_1|^2 |x_2 - y_2|^2}  \\ &+C \int dy_1 dy_2 dx_1 dx_2 \frac{|\nabla \varphi_t^{(N)} (x_1)|^2 |\nabla \varphi_t^{(N)} (x_2)|^2 |\varphi_t^{(N)} (y_1)|^2 |\varphi_t^{(N)} (y_2)|^2}{|x_1 - y_1| |x_1 - y_2| |x_2-y_1| |x_2 - y_2|} \\  \leq \;& C \int dy_1 dy_2 dx_1 dx_2  \frac{|\varphi_t^{(N)} (x_1)|^6 |\varphi_t^{(N)} (x_2)|^2}{|x_1 - y_1|^2 |x_1 - y_2|^2 |x_2-y_1|^2 |x_2 - y_2|^2} \\ &+C 
\int dy_1 dy_2 dx_1 dx_2  \frac{ |\nabla\varphi_t^{(N)} (x_1)|^2 |\nabla \varphi_t^{(N)} (x_2)|^2 
|\varphi_t^{(N)} (y_1)|^2 |\varphi_t^{(N)} (y_2)|^2}{|x_1 - y_1|^2 |x_2 - y_2|^2}  \\ \leq \; &C \int dx_1 dx_2 \frac{1}{|x_1 - x_2|^2} |\varphi_t^{(N)} (x_1)|^6 |\varphi_t^{(N)} (x_2)|^2 \\ &+ C \left( \sup_x \int dy \frac{1}{|x-y|^2} |\varphi_t^{(N)} (y)|^2 \right)^2 \| \varphi_t^{(N)} \|_{H^1}^4 \\ \leq \; & C 
\end{split} \]
for a constant $C$ depending only on the $H^1$-norm of $\varphi_t^{(N)}$. The last two bounds prove (\ref{eq:lmK-4}).

The inequalities (\ref{eq:lmK-5}), (\ref{eq:lmK-6}) can be proven similarly to (\ref{eq:lmK-3}) and (\ref{eq:lmK-4}); this time, however, the bounds will contain the norm $\| \dot{\varphi}_t^{(N)} \|_{H^1}$, which is bounded by $C e^{K |t|}$, as proven in Proposition \ref{t:pdes}.
\end{proof}

\begin{proof}[Proof of Proposition \ref{prop:TKT}] 
We prove the first bound in (\ref{eq:TKT-cl}). To this end, we observe that Lemma \ref{lm:K} can be used to bound all factors on the r.h.s.\ of (\ref{eq:TKT})  (using the uniform estimates for $\| \nabla_1 p (k_t) \|_{L^2 (\bR^3 \times \bR^3)}, \| \nabla_1 r_{k_t} \|_{L^2 (\bR^3 \times \bR^3)}$ from Lemma \ref{l:kernels}), with two exceptions, given by the term
\begin{equation}\label{eq:exc-K} \int dx \nabla_x a^*_x a^* (\nabla_x k_x) \end{equation}
and its hermitian conjugate. To control (\ref{eq:exc-K}), we use that, from (\ref{eq:kt}),
\[ \nabla_x k_t (y,x) = - N^2 \nabla w (N (x-y)) \varphi_t^{(N)} (x) \varphi_t^{(N)} (y) - N w (N (x-y)) \nabla\varphi_t^{(N)} (x) \varphi_t^{(N)} (y).  \]
Hence
\begin{equation}\label{eq:exc2-K} \begin{split} 
\int dx \, \nabla_x a_x^* a^* (\nabla_x k_x)  = &- N^2 \int dx dy \nabla w (N (x-y)) \varphi_t^{(N)} (x) \varphi_t^{(N)} (y) \nabla_x a^*_x a_y^* \\ &- N \int dx dy \, w(N (x-y)) \nabla \varphi_t^{(N)} (x) \varphi_t^{(N)} (y) \nabla_x a_x^* a_y^*.   \end{split}\end{equation}
The last term can be written as
\begin{equation} \int dx  \nabla_x a_x^* a^* (j_x) \end{equation}
with $j(y,x) = -N w(N (x-y)) \nabla \varphi_t^{(N)} (x) \varphi_t^{(N)} (y)$. Combining Lemma \ref{lm:w} and Proposition~\ref{t:pdes}, we find that $j \in L^2 (\bR^3 \times \bR^3)$, with uniformly bounded norm. Hence, Lemma \ref{lm:K} implies that, for every $\delta > 0$, there exists $C_\delta >0$ with
\begin{equation}\label{eq:exc3-K} \left| \int dx \, \langle \psi , \nabla_x a_x^* a^* (j_x) \psi \rangle \right| \leq \delta \| \cK^{1/2} \psi \|^2 + C_\delta \| (\cN + 1)^{1/2} \psi \|^2. \end{equation}
The first term on the r.h.s.\ of (\ref{eq:exc2-K}), on the other hand, can be written as
\[ \begin{split}
- N^2 \int dx dy \nabla w (N (x-y)) \varphi_t^{(N)} &(x) \varphi_t^{(N)} (y) \nabla_x a^*_x a_y^* \\ &=
 N^3 \int dx dy (\Delta w) (N (x-y)) \varphi_t^{(N)} (x) \varphi_t^{(N)} (y)
 a^*_x a_y^* \\ &\quad + N^2 \int dx dy \nabla w (N (x-y)) \nabla \varphi^{(N)}_t (x) \varphi_t^{(N)} (y) a_x^* a_y^*. \end{split} \]
The first contribution on the r.h.s.\ of the last equation is large and appears explicitly on the r.h.s.\ of (\ref{eq:pr-K}) (it will cancel later, when combined with other terms arising from $\wt{\cL}^{(0)}_{2,N}$ and $\cL^{(0)}_{4,N}$). The second term, on the other hand, is an error; integrating by parts, it can be expressed as
\[ \begin{split} N^2 \int dx dy \nabla &w (N (x-y)) \nabla \varphi^{(N)}_t (x) \varphi_t^{(N)} (y) a_x^* a_y^* \\ = \; &- N \int dx dy  w (N (x-y)) \nabla \varphi^{(N)}_t (x) \nabla \varphi_t^{(N)} (y) a_x^* a_y^* \\ &- N \int dx dy w (N (x-y)) \nabla \varphi^{(N)}_t (x)  \varphi_t^{(N)} (y) a_x^* \nabla_y a_y^*  \\ = \; &- \int dx \nabla \varphi_t^{(N)} (x) a_x^* a^* (N w(N(x-.) \nabla \varphi_t^{(N)}) + \int dy  \nabla_y a_y^* a^* (j_y) \end{split} \]
with $j(x,y) = -N w(N(x-y)) \nabla\ph(x)\ph(y)$. The second term is bounded as in (\ref{eq:exc3-K}). The first term, on the other hand, is estimated by
\begin{equation}\label{eq:lmK-sp} \begin{split}\Big|   \int dx \nabla \varphi_t^{(N)} (x)  \langle a_x \psi, &a^* (N w(N(x-.)) \nabla \varphi_t^{(N)}) \psi \rangle \Big| \\ &\leq C \sup_x \| N w(N(x-.)) \nabla \varphi_t^{(N)} \|_2 \| (\cN+1)^{1/2} \psi \|^2 \\ &\leq C \| \varphi_t^{(N)} \|_{H^2} \| (\cN+1)^{1/2} \psi \|^2   \end{split} \end{equation} and \eqref{eq:hireg}.

Also the second bound in (\ref{eq:TKT-cl}) follows from Lemma \ref{lm:K}. In fact, when one takes the commutator of $\cN$ with the terms on the r.h.s.\ of (\ref{eq:TKT}) one either finds zero (for all terms with one creation and one annihilation operators, which therefore preserve the number of particles), or one finds again the same terms (up to a possible change of sign). This follows because, by the canonical commutation relations
\[ [ \cN , a (f) ] = - a (f) \qquad \text{and } \quad [\cN , a^* (f)] = a^* (f) \]
for every $f \in L^2 (\bR^3)$. Finally, the third bound in (\ref{eq:TKT-cl}) is a consequence of Lemma \ref{lm:K} as well. In fact, the time derivative $\dot{\cE}_K (t)$ is a sum of terms very similar to the terms appearing on the r.h.s.\ of (\ref{eq:TKT}), with the difference that one of the appearing kernels contains a time-derivative. Combining the estimates from Lemma \ref{lm:K} (including, in this case, also (\ref{eq:lmK-5}), (\ref{eq:lmK-6})) with the bounds for $\| \nabla_1 \dot{p}_{k_t} \|_2$, $\| \nabla_1 \dot{r}_{k_t} \|_2$ from Lemma \ref{lm:dotk} and with the bound for $\| \dot{\varphi}_t^{(N)} \|_{H^1}$ from  Proposition \ref{t:pdes} (needed to control terms similar to (\ref{eq:exc2-K}), (\ref{eq:lmK-sp}), with 
a factor of $\varphi^{(N)}_t$ replaced by $\dot{\varphi}^{(N)}_t$), we obtain the last inequality in (\ref{eq:TKT-cl}). 
%Note that, on the l.h.s.\ of (\ref{eq:lmK-sp}), it is always possible to arrange the terms so that the time %derivative hits the first factor $\nabla \varphi_t^{(N)} (x)$, outside the argument of $a^*$; for this %reason, the resulting contribution is bounded by $\| \nabla \dot{\varphi}_t^{(N)} \|_2 \| \varphi_t^{(N)} \|%_{H^2}$, which is controlled by $\| \varphi_t^{(N)}\|_{H^3}^2$ via the modified Gross-Pitaevskii %equation \eqref{eq:mod-GP}. The proof is completed noting $\norm{\ph}_{H^3}^2 \leq C \ekt$ by 
%\eqref{eq:hireg} for some $C,K$. 
\end{proof}

\subsubsection{Properties of $T^* \wt{\cL}_{2,N}^{(0)} (t) T$}

We consider now the other quadratic terms, collected in $\wt{\cL}_{2,N} (t)$, defined by $\cL^{(0)}_{2,N} (t) = \cK + \wt{\cL}_{2,N}^{(0)} (t)$. 
We have
\[ \begin{split} 
T^* \wt{\cL}_{2,N}^{(0)} (t) T = & \; \int dx \left( N^3 V (N .) * |\varphi_t^{(N)}|^2 \right) (x) (a^* (c_x) + a(s_x)) (a(c_x) + a^* (s_x)) \\ & + \int dx dy \, N^3 V(N(x-y)) \varphi_t^{(N)} (x) \overline{\varphi}_t^{(N)} (y) (a^* (c_x) + a(s_x)) (a (c_y) + a^* (s_y)) \\ &+ \frac{1}{2} \int dx dy \, N^3 V(N (x-y)) \\
&\hspace{2cm} \times \left[  \varphi_t^{(N)} (x) \varphi_t^{(N)} (y) (a^* (c_x) + a(s_x)) (a^* (c_y) + a (s_y)) + \text{h.c.}  \right]. 
\end{split} \]
Expanding the products, and bringing all terms to normal order, we find
\begin{equation}\label{eq:T2T} \begin{split} 
T^*  \wt{\cL}_{2,N}^{(0)}& (t) T \\ =\; & \int dx \left(N^3 V (N .) * |\varphi_t^{(N)}|^2 \right) (x) \\ &\hspace{1cm}\times \left[  a^* (c_x) a (c_x) + a^* (s_x) a (s_x) + a^* (c_x) a^* (s_x) + a(s_x) a(c_x) + \langle s_x , s_x \rangle \right] \\
&+ \int dx dy N^3 V (N (x-y)) \varphi_t^{(N)} (x) \overline{\varphi}_t^{(N)} (y) \\ &\hspace{1cm}\times  \left[ a^* (c_x) a(c_y) + a^* (s_y) a(s_x) + a^* (c_x) a^* (s_y) + a(s_x) a(c_y) + \langle s_x , s_y \rangle \right] \\
&+ \frac{1}{2} \int dx dy N^3 V(N (x-y)) \varphi_t^{(N)} (x) \varphi_t^{(N)} (y) \\ &\hspace{1cm}\times \left[ a^* (c_x) a^* (c_y) + a^* (c_x) a(s_y) + a^* (c_y) a (s_x) + a(s_x) a(s_y) + \langle s_x , c_y \rangle \right] \\
&+ \frac{1}{2} \int dx dy N^3 V(N (x-y)) \overline{\varphi}_t^{(N)} (x) \overline{\varphi}_t^{(N)} (y) \\ &\hspace{1cm}\times \left[ a (c_x) a (c_y) + a^* (s_y) a(c_x) +  a^* (s_x) a(c_y) + a^* (s_x) a^* (s_y) + \langle c_y , s_x \rangle \right]. 
\end{split} \end{equation}

The properties of $T^* \wt{\cL}_{2,N}^{(0)} (t) T$ are summarized in the following proposition.
\begin{proposition}\label{prop:TwtL2T}
We have 
\begin{equation} \begin{split} 
T^* \wt{\cL}_{2,N}^{(0)} &(t) T \\ =\; & \int dx (N^3 V(N.) * |\varphi_t^{(N)}|^2) (x) \langle s_x , s_x \rangle \\ &+ \int dx dy N^3 V(N (x-y)) \varphi_t^{(N)} (x) \overline{\varphi}_t^{(N)} (y) \langle s_x , s_y \rangle \\ &+ \text{Re } \int dx dy N^3 V(N (x-y)) \varphi_t^{(N)} (x) \varphi_t^{(N)} (y) \langle s_x ,c_y \rangle \\ &+ \frac{1}{2} \int dx dy N^3 V(N (x-y)) \left[ \varphi_t^{(N)} (x) \varphi_t^{(N)} (y)\, a_x^* a_y^*  + \overline{\varphi}_t^{(N)} (x) \overline{\varphi}_t^{(N)} (y) \, a_x a_y \right]  \\ &+ \cE_2 (t) \end{split} \end{equation}
where the error $\cE_2 (t)$ is such that, for appropriate constants $C,K > 0$,
\begin{equation} \label{eq:T2T-cl}
\begin{split}
\pm  \cE_2 (t) &\leq C \ech{\| \varphi_t^{(N)} \|_{H^2}^2}{\ekt} \,  (\cN + 1), \\  \pm \left[ \cN , \cE_2 (t) \right]  &\leq C \ech{\| \varphi_t^{(N)} \|_{H^2}^2}{\ekt} \, (\cN + 1), \\ \pm \dot{\cE}_2 (t) &\leq C \, \ech{\left(\| \varphi_t^{(N)} \|_{H^4} \| \varphi_t^{(N)} \|_{H^2} +   \| \varphi_t^{(N)} \|_{H^2}^3 \right)}{\ekt} \, (\cN+1) .
\end{split}
\end{equation} 
\end{proposition}

To show Proposition \ref{prop:TwtL2T}, we will make use of the next lemma.
\begin{lemma}\label{lm:2b}
Let $j_1, j_2 \in L^2 (\bR^3 \times \bR^3)$. Let $j_{i,x} (z) := j_i (z,x)$ for $i=1,2$. Then there exists a constant $C$ such that 
\begin{equation}\label{eq:lm2b-1}  
\begin{split}
\int dx dy N^3 V(N (x-y)) |\varphi_t^{(N)} (x)| |\varphi_t^{(N)} (y)| & \| a^\sharp (j_{1,x}) \psi \| \, \| a^\sharp (j_{2,y}) \psi \| \\  &\leq C \| \varphi_t^{(N)} \|_{H^2}^2  \, \| j_1 \|_2 \, \| j_2 \|_2 \, \| (\cN+1)^{1/2} \psi \|^2. \end{split} \end{equation}
Moreover,  
\begin{equation}\label{eq:lm2b-2} 
\begin{split}
\int dx dy N^3 V(N (x-y)) |\varphi_t^{(N)} (x)| |\varphi_t^{(N)} (y)| &\| a^\sharp (j_{1,x}) \psi \|  \| a_y \psi \| \\&\leq C\| \varphi_t^{(N)} \|_{H^2}^2  \, \| j_1 \|_2 \| (\cN+1)^{1/2} \psi \|^2 \end{split} \end{equation}
and 
\begin{equation}\label{eq:lm2b-3}
\int dx dy N^3 V(N (x-y))   |\varphi_t^{(N)} (x)| |\varphi_t^{(N)} (y)|  \, \| a_x \psi \| \, \|  a_y  \psi \|  
\leq C \| \varphi_t^{(N)}\|_{H^2}^2  \,\| (\cN+1)^{1/2} \psi \|^2 .
\end{equation}
%These bounds imply that
%\begin{equation}\label{eq:lm2b-3}  
%\begin{split}
%\int dx dy N^3 V(N (x-y)) |\varphi_t^{(N)} &(x)| |\varphi_t^{(N)} (y)| \| a^\sharp (j_{1,x}) \psi \| \, \| a (c_y) %\psi \|  \\ &\leq C\| \varphi_t^{(N)} \|_{H^2}^2 \| j_1 \|_2 (1+\| p \|_2) \,  \| (\cN+1)^{1/2} \psi \|^2 \end{split}%\end{equation}
The bounds remain true if both creation and/or annihilation operators act on the same variable, in the sense that
\begin{equation}\label{eq:lm2b-4} \begin{split}
 \int dx dy N^3 V(N (x-y))&|\varphi_t^{(N)} (x)|^2 \| a^\sharp (j_{1,y})  \psi \| \, \|  a^\sharp (j_{2,y}) \psi \|  
 \\  &\hspace{3cm} \leq C\| \varphi_t^{(N)} \|_{H^2}^2 \|j_1 \|_2 \| j_2 \|_2 \, \| (\cN+1)^{1/2} \psi \|^2, \\
 \int dx dy N^3 V(N (x-y))& |\varphi_t^{(N)} (x)|^2  \| a^\sharp (j_{1,y}) \psi \| \, \| a_y \psi \| \leq C\| \varphi_t^{(N)} \|_{H^2}^2 \| j_1 \|_2 \| (\cN+1)^{1/2} \psi \|^2, \\
\int dx dy N^3 V(N (x-y)) &|\varphi_t^{(N)} (x)|^2 \| a_y \psi \|^2 \leq C \| \varphi_t^{(N)} \|_{H^2}^2 \| \cN^{1/2} \psi \|^2.
%\int dx dy N^3 V(N (x-y)) |\varphi_t^{(N)} (x)|^2 &\| a^\sharp (j_{1,y}) \psi \| \, \| a (c_y)  \psi  \|\\  & \leq C %\| \varphi_t^{(N)} \|_{H^2}^2  \| j_1 \|_2 (1+ \| p \|) \| (\cN+1)^{1/2} \psi \|^2
\end{split} \end{equation}
%We will also make use of the following estimates:
%\begin{equation}\label{eq:lm2b-5} \begin{split} 
%\int dx dy N^3 V(N (x-y))   |\varphi_t^{(N)} (x)| |\varphi_t^{(N)} (y)|  \, \| a (c_x) \psi \| \, \|  a (c_y)  \psi \| %&\leq C  \| (\cN+1)^{1/2} \psi \|^2 \\
%\int dx dy N^3 V(N (x-y)) |\varphi_t^{(N)} (x)|^2 \| a (c_y) \psi \| \, \|  a (c_y) \psi \| &\leq C \| (\cN
%+1)^{1/2} \psi \|^2 \end{split} \end{equation}
\end{lemma}

\begin{proof}
To prove (\ref{eq:lm2b-1}), we notice that for any $\alpha > 0$
\[ \begin{split} 
\int dx dy &N^3 V(N (x-y)) |\varphi_t^{(N)} (x)| |\varphi_t^{(N)} (y)|  \| a^\sharp (j_{1,x}) \psi \| \, \| a^\sharp (j_{2,y}) \psi \| \\ \leq \; & \alpha \int dx dy N^3 V(N (x-y)) |\varphi^{(N)}_t (x)|^2 \| j_{2,y} \|_2^2 \| (\cN+1)^{1/2} \psi \|^2 \\ &+ \alpha^{-1}
 \int dx dy N^3 V(N (x-y)) |\varphi^{(N)}_t (y)|^2 \| j_{1,x} \|_2^2 \| (\cN+1)^{1/2} \psi \|^2 \\
 \leq \; & \| N^3 V(N.) * |\varphi_t^{(N)}|^2 \|_\infty  \left( \alpha \| j_{1} \|^2_2 + \alpha^{-1}
 \| j_2 \|_{2}^2 \right) \, \| (\cN+1)^{1/2} \psi \|^2.
 \end{split} \]
 Using $\| N^3 V(N.) * |\varphi_t^{(N)}|^2 \|_\infty \leq C \| \varphi_t^{(N)} \|_{H^2}^2$, and optimizing over $\alpha >0$, we obtain (\ref{eq:lm2b-1}). Analogously, (\ref{eq:lm2b-2}) follows from
\[ \begin{split} 
\int dx dy N^3 V(N (x-y)) &|\varphi^{(N)}_t (x)| |\varphi^{(N)}_t (y)| \| a^\sharp (j_{1,x}) \psi \| \| a_y \psi \| \\ 
\leq \; & \alpha \int dx dy N^3 V(N (x-y)) |\varphi^{(N)}_t (x)|^2 \| a_y \psi \|^2  \\ &+ \alpha^{-1}
 \int dx dy N^3 V(N (x-y)) |\varphi^{(N)}_t (y)|^2 \| j_{1,x} \|_2^2 \| (\cN+1)^{1/2} \psi \|^2 \\
 \leq \; & \| N^3 V(N.) * |\varphi_t^{(N)}|^2 \|_\infty (\alpha + \alpha^{-1} \| j_1\|_{2}^2) \| (\cN+1)^{1/2} \psi \|^2
 \end{split} \]
for any $\alpha > 0$. Optimizing over $\alpha$ gives (\ref{eq:lm2b-2}). Eq.\ (\ref{eq:lm2b-3}) follows from Cauchy-Schwarz. The bounds in (\ref{eq:lm2b-4}) can be shown similarly.
%Eq.\ (\ref{eq:lm2b-3}) can be proven by writing $a(c_y) = a_y + a (p_y)$, using the notation $p_y (z) = %p_{k_t} (y,z)$, with $p_{k_t}$ as defined in (\ref{eq:def-pk}), and applying (\ref{eq:lm2b-1}) and 
%(\ref{eq:lm2b-2}). The bounds in (\ref{eq:lm2b-4}) are proven similarly; this time, however, one can %immediately integrate over $x$ after bound the result with the uniform estimate $\| N^3 V(N.) * |
%\varphi_t^{(N)}|^2 \|_\infty < C$. Finally, to prove (\ref{eq:lm2b-5}), we write again $a (c_y) = a_y + a %(p_y)$. Using (\ref{eq:lm2b-1}) and (\ref{eq:lm2b-2}), we obtain
%\[  \begin{split} 
%\int dx dy &N^3 V(N (x-y))   |\varphi_t^{(N)} (x)| |\varphi_t^{(N)} (y)|  \, \| a (c_x) \psi \| \, \|  a (c_y)  \psi \|
%\\ \leq \; & \int dx dy N^3 V(N (x-y)) |\varphi_t^{(N)} (x)| |\varphi_t^{(N)} (y)| \| a_x \psi \| \| a_y \psi \| + C %\| (\cN+1)^{1/2} \psi \|^2 \\ \leq \; & C \int dx dy N^3 V(N (x-y)) |\varphi_t^{(N)} (x)|^2 \| a_y \psi \|^2 + C %\| (\cN+1)^{1/2} \psi \|^2  \\ \leq \; & C \| N^3 V(N.) * |\varphi_t^{(N)}|^2 \|_\infty \| \cN^{1/2} \psi \|^2 + C \| %(\cN+1)^{1/2} \psi \|^2 \end{split} \] 
%which imply the first estimate in (\ref{eq:lm2b-5}). The second estimate follows similarly.
\end{proof} 

\begin{proof}[Proof of Proposition \ref{prop:TwtL2T}]
To prove the first bound in (\ref{eq:T2T-cl}), we notice that the quadratic terms on the r.h.s.\ of (\ref{eq:T2T}) can be controlled with Lemma \ref{lm:2b}, decomposing, if needed, $a (c_x) = a_x + a(p_x)$ and then applying (\ref{eq:lm2b-1}), (\ref{eq:lm2b-2}), or (\ref{eq:lm2b-3}). There are two exceptions, given by the terms proportional to $a^* (c_x) a^* (c_y)$ and its hermitian conjugate, proportional to $a (c_x) a(c_y)$. For these two terms the bounds from Lemma \ref{lm:2b} do not apply. Instead, using $a^* (c_x) = a^*_x + a^* (p_x)$, we write 
\begin{equation}\label{eq:2-cancel} \begin{split} 
\int dx dy N^3 & V(N (x-y)) \varphi_t^{(N)} (x) \varphi_t^{(N)} (y) a^* (c_x) a^* (c_y) \\ = \; &   
\int dx dy N^3 V(N (x-y)) \varphi_t^{(N)} (x) \varphi_t^{(N)} (y) a_x^* a_y^* \\ &+ 
\int dx dy N^3 V(N (x-y)) \varphi_t^{(N)} (x) \varphi_t^{(N)} (y) a^* (p_x) a_y^* \\ &+ 
\int dx dy N^3 V(N (x-y)) \varphi_t^{(N)} (x) \varphi_t^{(N)} (y) a^* (c_x) a^* (p_y).
 \end{split} \end{equation}
The contribution of the last two terms can be bounded by Lemma \ref{lm:2b}, because one of the arguments of the creation operators is square integrable. In fact
\[ \begin{split} \Big| \int dx dy N^3 V(N (x-y)) &\varphi_t^{(N)} (x) \varphi_t^{(N)} (y)  \langle \psi , a^* (p_x) a_y^* \psi \rangle \Big| \\ \leq \; & \int dx dy N^3 V(N (x-y)) |\varphi_t^{(N)} (x)| |\varphi_t^{(N)} (y)|   \| a_y \psi \| \, \| a^* (p_x) \psi \|   \\ \leq \; & C \| \varphi_t^{(N)} \|_{H^2}^2 \, \| (\cN+1)^{1/2} \psi \|^2 \leq C \ekt \| (\cN+1)^{1/2} \psi \|^2
\end{split} \]
by (\ref{eq:lm2b-2}) and \eqref{eq:hireg}, and similarly for the last term on the r.h.s.\ of (\ref{eq:2-cancel}).
% while
%\[ \begin{split} \Big| \int dx dy N^3 V(N (x-y)) &\varphi_t^{(N)} (x) \varphi_t^{(N)} (y)  \langle \psi , a^* 
%(c_x) a^* (p_y) \psi \rangle \Big| \\ \leq \; & \int dx dy N^3 V(N (x-y)) |\varphi_t^{(N)} (x)| |\varphi_t^{(N)} %(y)|   \| a (c_x) \psi \| \, \| a^* (p_y) \psi \|   \\ \leq \; & C \| (\cN+1)^{1/2} \psi \|^2 
%\end{split} \]
%by (\ref{eq:lm2b-3}). 
The hermitian conjugate of (\ref{eq:2-cancel}), proportional to $a (c_x) a (c_y)$, can be handled identically. 

The second bound in (\ref{eq:T2T-cl}) follows similarly, using the fact that the commutator of $\cN$ with the terms on the r.h.s.\ of (\ref{eq:T2T}) leaves their form unchanged (apart from the constant terms and the quadratic terms with one creation and one annihilation operators, whose contribution to the commutator $[\cN , \cE_2 (t)]$ vanishes).  

Also the third bound in (\ref{eq:T2T-cl}) can be proven analogously, using
the bounds for $\| \dot{\text{sh}}_{k_t} \|_2$ and $\| \dot{p}_{k_t} \|_2$, as proven in Lemma \ref{lm:dotk} . When the time derivative hits the factor $\varphi_t^{(N)} (x)$ or $\varphi_t^{(N)} (y)$, it generates a contribution which is bounded by $\| \dot\varphi_t^{(N)} \|_{H^2} \leq C \| \varphi_t^{(N)} \|_{H^4}^4 \leq C e^{K|t|}$ (for some $K$ depending only on $\norm{\varphi}_{H^1}$; here we used Proposition \ref{t:pdes}).
\end{proof}

\subsection{Analysis of $T^* \cL^{(0)}_{3,N} (t) T$}

We consider now the contributions arising from the cubic terms in $\cL^{(0)}_{3,N} (t)$.
We have
\[ \begin{split}
T^* & \cL_{3,N}^{(0)} (t) T \\ = \; & \frac{1}{\sqrt{N}} \int dx dy \, N^3 V(N (x-y)) \\ &\hspace{.5cm} \times \left[ \varphi_t^{(N)} (y) \, ( a^* (c_x) + a (s_x)) (a^* (c_y) + a (s_y)) (a (c_x) + a^* (s_x))  + \text{h.c.} \right] \\ = \; & \frac{1}{\sqrt{N}} \int dx dy N^3 V(N(x-y)) \varphi_t^{(N)} (y) \\ & \hspace{.5cm} 
\times \left[ \, a^* (c_x) a^* (c_y) a^* (s_x) + a^* (c_x) a^* (c_y) a(c_x) + a^* (c_x) a (s_y) a^* (s_x) +a^* (c_x) a (s_y) a (c_x)  \right. \\ & \hspace{1cm} \left. +  a (s_x) a^* (c_y) a^* (s_x) + a (s_x) a^* (c_y) a(c_x) + a (s_x) a (s_y) a^* (s_x) +a (s_x) a (s_y) a (c_x) \right]  \\ &+\text{h.c.} 
\end{split} \]
Writing the terms in normal order, we find
\begin{equation}\label{eq:cub-no} \begin{split} 
T^* &\cL_{3,N}^{(0)} (t) T  \\ = \; &\frac{1}{\sqrt{N}}  \int dx dy N^3 V(N(x-y)) \varphi_t^{(N)} (y) \\ & \hspace{.5cm}  \times \left[ a^* (c_x) a^* (c_y) a^* (s_x) + a^* (c_x) a^* (c_y) a(c_x) + a^* (c_x)  a^* (s_x) a (s_y) +a^* (c_x) a (s_y) a (c_x)  \right. \\ & \hspace{1cm} \left. +   a^* (c_y) a^* (s_x) a (s_x) + a^* (c_y) a (s_x)  a(c_x) + a^* (s_x) a (s_x) a (s_y)  +a (s_x) a (s_y) a (c_x) \right]  \\
&+ \frac{1}{\sqrt{N}} \int dx dy N^3 V (N (x-y)) \varphi_t^{(N)} (y) \\ &\hspace{.5cm} \times  \left[  \langle s_y , s_x \rangle (a^* (c_x) + a (s_x)) + \langle s_x, c_y \rangle \left( a (c_x) + a^* (s_x) \right) + \langle s_x, s_x \rangle \left(a^* (c_y) + a (s_y)\right) \right] 
\\ &+\text{h.c.}
\end{split} \end{equation}

The properties of $T^* \cL_{3,N}^{(0)} (t) T$ are summarized in the following proposition. 
\begin{proposition}\label{prop:TL3T}
We have
\begin{equation}\label{eq:pr-TL3T} \begin{split} T^* &\cL^{(0)}_{3,N} T \\ =  \; & \frac{1}{\sqrt{N}} 
\int dx dy N^3 V (N (x-y))  \\ & \hspace{.5cm} \times \left[ \varphi_t^{(N)} (y) \overline{k}_t (x,y)  \left( a (c_x) + a^* (s_x) \right)   + \overline{\varphi}_t^{(N)} (y) k_t (x,y)  \left( a^* (c_x) + a (s_x) \right)  \right] \\ &+ \cE_3 (t) \\
= \; &- \sqrt{N} \int dx dy N^3 V(N (x-y)) w(N (x-y)) |\varphi_t^{(N)} (y)|^2 \varphi_t^{(N)} (x) (a (c_x) + a^* (s_x)) + \text{h.c.} \\ &+ \cE_3 (t)
\end{split} \end{equation}
where we used the definition (\ref{eq:kt}) of the kernel $k_t$ and where the error term $\cE_3 (t)$ is such that for every $\delta > 0$ there exists a constant $C_\delta > 0$ with
\begin{equation}\label{eq:TL3T-cl}
\begin{split} \pm \cE_3 (t)  &\leq \delta \int dx dy N^2 V (N (x-y)) a_x^*
a_y^* a_x a_y + \delta \,\frac{\cN^2}{N} + C_\delta \ech{\| \varphi_t^{(N)}
\|_{H^2}^2}{\ekt} \left(\cN + 1\right), \\
\pm \left[ \cN , \cE_3 (t) \right] &\leq \delta \int dx dy N^2 V (N (x-y))
a_x^* a_y^* a_x a_y + \delta \,  \frac{\cN^2}{N} + C_\delta  \ech{\|
\varphi_t^{(N)} \|_{H^2}^2}{\ekt}  \left(\cN + 1\right), \\
\pm \dot{\cE}_3 (t) & \leq \delta \int dx dy N^2 V (N (x-y)) a_x^* a_y^* a_x a_y + \delta \, \frac{\cN^2}{N} 
\ech{\\ & \hspace{.5cm} + C_\delta \left( \| \varphi_t^{(N)} \|_{H^4}^2 \| \varphi_t^{(N)} \|_{H^2}^2 + \| \varphi_t^{(N)} \|_{H^2}^4 \right)}{+ C_\delta \ekt} \,  \left( \cN + 1\right).
\end{split} \end{equation}
\end{proposition}

Notice here that the first term on the r.h.s.\ of (\ref{eq:pr-TL3T}) cancels exactly with the contribution (\ref{eq:TL1T}); we will make use of this crucial observation in the proof of Theorem \ref{thm:L} below.

\begin{proof}
To bound the cubic terms on the r.h.s.\ of (\ref{eq:cub-no}), we systematically apply Cauchy-Schwarz.
This way, we control cubic terms by quartic and quadratic contributions, which are then estimated making use of Lemma \ref{lm:2b} (the quadratic part) and Lemmas \ref{lm:4} and \ref{lm:4b} (the quartic part). For example,
\[ \begin{split} \Big| \frac{1}{\sqrt{N}}\int dx dy &N^3 V(N(x-y)) \varphi_t^{(N)} (y) \langle  \psi , a^* (c_x) a^* (c_y) a^* (s_x) \psi \rangle \Big| \\ \leq \; &  \frac{1}{\sqrt{N}} \int dx dy N^3 V (N (x-y)) |\varphi_t^{(N)} (y)|  \| a (c_x) a (c_y) \psi \| \, \| a^* (s_x) \psi \| \\ \leq \; & \frac{\delta}{N}  \int dx dy \, N^3 V (N (x-y)) \| a (c_x) a(c_y) \psi \|^2 \\ &+ C_\delta\int dx dy N^3 V (N (x-y)) |\varphi_t^{(N)} (y)|^2 \| a^* (s_x)  \psi \|^2 \\ \leq \; & \frac{C \delta}{N} \int dx dy N^3 V(N (x-y)) \| a_y a_x \psi \|^2 + C_\delta  \| (\cN + 1)^{1/2}  \psi \|^2  \end{split} \]
where, in the last line, we used (\ref{eq:lm4b-1}) (from Lemma \ref{lm:4b}) and (\ref{eq:lm2b-4}) (from Lemma \ref{lm:2b}). All other cubic terms can be bounded similarly. We always separate the three creation and/or annihilation operators putting a small weight $\delta$ in front of the quartic term and
in such a way that, in the resulting quartic contribution, two operators depend on the $x$ and two on the $y$ variable. The corresponding quadratic term depends on $x$ and can always be bounded by (\ref{eq:lm2b-4}). It should be noted that the quartic contribution has either the form $\| a (c_x) a (c_y) \psi \|^2$ or $\| a (c_x) a^\sharp (j_y) \psi \|^2$ or $\| a^\sharp (j_{1,x}) a^\sharp (j_{2,y}) \psi \|^2$, with square-integrable arguments $j_1, j_2$ (here $a^\sharp$ is either $a$ or $a^*$). These terms can always be controlled using Lemma \ref{lm:4} or Lemma \ref{lm:4b}. As for the linear contributions on the r.h.s.\ of (\ref{eq:cub-no}), the first and third can simply be bounded by $N^{-1/2} \, (\cN+1)^{1/2}$, since 
\[ |\langle s_y , s_x \rangle| \leq C |\varphi^{(N)}_t (x)| |\varphi^{(N)}_t (y)|. \]
To bound the second linear term, we write 
\[ \begin{split} 
\frac{1}{\sqrt{N}} \int dx dy & N^3 V(N (x-y)) \varphi_t^{(N)} (y) \langle s_x, c_y \rangle (a(c_x) + a^* (s_x))  \\ = \; & \frac{1}{\sqrt{N}} \int dx dy  N^3 V(N (x-y)) \varphi_t^{(N)} (y) \overline{k}_t (x,y) \, (a(c_x) + a^* (s_x)) + \wt{\cE} (t) \end{split} \]
where $\pm \wt{\cE} (t) \leq N^{-1/2} (\cN+1)^{1/2}$ because, using Lemma
\ref{l:kernels}, \[ |\langle s_x , c_y \rangle - \overline{k}_t (x,y)| \leq
C |\varphi_t^{(N)} (x)| |\varphi_t^{(N)} (y)|.\] {F}rom  \eqref{eq:hireg},
this concludes the proof of the first estimate in (\ref{eq:TL3T-cl}). The
other two estimates are proven analogously, using the fact that the
commutators of $\cN$ with the terms on the r.h.s.\ of (\ref{eq:cub-no}) have
the same form as the terms on the r.h.s.\ of (\ref{eq:cub-no}) (with
possibly just a different sign), and using the bounds for $\|
\dot{\text{sh}}_{k_t} \|_2$ and $\| \dot{p}_{k_t} \|_2$ from Lemma
\ref{lm:dotk}, and the bound for $\| \dot{\varphi}_t^{(N)} \|_{H^2}$ from
Proposition \ref{t:pdes}.
\end{proof}

\subsection{Analysis of $T^* \cL^{(0)}_{4,N} T$}

We consider next the contributions arising from the quartic part $\cL_{4,N}^{(0)}$ of $\cL_N^{(0)} (t)$. We have
\[\begin{split} T^* \cL_{4,N}^{(0)} (t) T = &\frac{1}{2N} \int dx dy \, N^3
V(N (x-y)) \\ &\hspace{1cm} \times (a^* (c_x) + a (s_x)) (a^* (c_y) + a
(s_y)) (a (c_y) + a^* (s_y)) (a(c_x) + a^* (s_x)). \end{split} \] 
Expanding the products, we find 
\[ \begin{split} 
2 T^* & \cL_{4,N}^{(0)} (t) T \\ =  \; & \int dx dy \, N^2 V (N (x-y)) \Big[ a^* (c_x) a^* (c_y) a^* (s_y) a^* (s_x) + a^* (c_x) a^* (c_y) a^* (s_y) a (c_x) \\ &\hspace{.5cm} + a^* (c_x) a^* (c_y) a(c_y) a^* (s_x) + a^* (c_x) a^* (c_y) a (c_y) a (c_x)  + a^* (c_x) a (s_y) a^* (s_y) a^* (s_x)  \\ &\hspace{.5cm}  + a^* (c_x) a (s_y) a^* (s_y) a (c_x) + a^* (c_x) a (s_y) a(c_y) a^* (s_x) + a^* (c_x) a (s_y) a (c_y) a (c_x)  \\ &\hspace{.5cm}  + a (s_x) a^* (c_y) a^* (s_y) a^* (s_x) + a (s_x) a^* (c_y) a^* (s_y) a (c_x) + a (s_x) a^* (c_y) a(c_y) a^* (s_x)  \\ &\hspace{.5cm} + a (s_x) a^* (c_y) a (c_y) a (c_x)  + a (s_x) a (s_y) a^* (s_y) a^* (s_x) + a (s_x) a (s_y) a^* (s_y) a (c_x)  \\ &\hspace{.5cm}  + a (s_x) a (s_y) a(c_y) a^* (s_x) + a (s_x) a (s_y) a (c_y) a (c_x) \Big]. 
\end{split} \]
Writing all terms in normal order, we obtain
\begin{equation}\label{eq:TL4T} \begin{split} 
2T^* & \cL_{4,N}^{(0)} (t) T \\ =  \; & \int dx dy \, N^2 V (N (x-y)) \Big[ a^* (c_x) a^* (c_y) a^* (s_y) a^* (s_x) + a^* (c_x) a^* (c_y) a^* (s_y) a (c_x)  \\ &\hspace{.5cm} + a^* (c_x) a^* (c_y) a^* (s_x) a(c_y)  + a^* (c_x) a^* (c_y) a (c_y) a (c_x) + a^* (c_x)  a^* (s_y) a^* (s_x) a (s_y)  \\ &\hspace{.5cm}  
+ a^* (c_x) a^* (s_y) a (s_y) a (c_x) + a^* (c_x)  a^* (s_x)  a (s_y) a(c_y) + a^* (c_x) a (s_y) a (c_y) a (c_x)  \\ &\hspace{.5cm}  + a^* (c_y) a^* (s_y) a^* (s_x)  a (s_x) + a^* (c_y) a^* (s_y)  a (s_x) a (c_x) 
+  a^* (c_y)  a^* (s_x) a (s_x) a(c_y)  \\ &\hspace{.5cm}  + a^* (c_y)  a (s_x)  a (c_y) a (c_x) + a^* (s_y)   a^* (s_x) a (s_x) a (s_y) + a^* (s_y)  a (s_x) a (s_y) a (c_x)  \\ &\hspace{.5cm}  + a^* (s_x) a (s_x) a (s_y) a(c_y)  + a (s_x) a (s_y) a (c_y) a (c_x) \Big] \\
&+ \int dx dy \, N^2 V (N (x-y))  \Big[ \langle c_y, s_x \rangle a^* (c_x) a^* (c_y) + \langle s_x, c_y \rangle a (c_y)  a(c_x)   \\ &\hspace{.5cm} + 2 \langle s_y, s_y \rangle a^* (c_x)  a^* (s_x) + 2 \langle s_y, s_y \rangle a (s_x) a (c_x) 
+ 2 \langle s_y, s_x \rangle a^* (c_x)  a^* (s_y)   \\ &\hspace{.5cm}
+ 2\langle s_x , s_y \rangle a (s_y) a (c_x) + 2 \langle s_y, s_y \rangle   a^* (c_x) a (c_x) 
+ 2 \langle s_y , s_x \rangle a^* (c_x) a(c_y) 
  \\ &\hspace{.5cm} + 2 \langle s_x , s_y \rangle  a^* (s_x)  a (s_y) 
+ 2 \langle s_y , s_y \rangle  a^* (s_x)  a (s_x) +  \langle c_y , s_x \rangle  a^* (c_x)   a (s_y)   \\ &\hspace{.5cm} + \langle s_x, c_y \rangle a^* (s_y) a (c_x) 
+  \langle s_x , c_y \rangle a^* (s_y) a^* (s_x)   + \langle c_y , s_x \rangle a (s_x) a (s_y)  
  \\ &\hspace{.5cm} + \langle s_x , c_y \rangle a^* (s_x) a(c_y)  + \langle c_y, s_x \rangle  a^* (c_y) a (s_x) \Big] \\
&+ \int dx dy \, N^2 V (N (x-y)) \left[ |\langle s_x , c_y \rangle|^2  +  |\langle s_x , s_y \rangle|^2   
+ \langle s_y , s_y \rangle \langle s_x ,s_x \rangle \right]. 
\end{split} \end{equation}

The properties of $T^* \cL^{(0)}_{4,N} T$ are summarized in the next proposition. 
\begin{proposition}\label{prop:TL4T}
We have 
\begin{equation}\begin{split} 2 T^* \cL_{4,N}^{(0)} (t) T =\; & \int dx dy N^2 V (N (x-y)) \left[ |\langle s_x , c_y \rangle|^2  +  |\langle s_x , s_y \rangle|^2   + \langle s_y , s_y \rangle \langle s_x ,s_x \rangle \right] \\  &+ \int dx dy N^2 V (N (x-y)) a_x^* a_y^* a_y a_x \\ &+ \int dx dy \, N^2 V (N (x-y)) \left( k (x,y) a^*_x a^*_y + \overline{k} (x,y) a_x a_y \right) + \cE_4 (t) \end{split} \end{equation}
where the error $\cE_4 (t)$ is such that, for every $\delta > 0$, there exists a constant $C_\delta > 0$ with  
\begin{equation}\label{eq:TL4T-cl} 
\begin{split}
\pm \, \cE_4 (t) &\leq \delta \int dx dy N^2 V(N (x-y)) a_x^* a_y^* a_y a_x
+ C_\delta  \frac{\cN^2}{N}  + C_\delta \ech{\| \varphi_t^{(N)}
\|_{H^2}^2}{\ekt} \left(\cN + 1 \right), \\
\pm \, \left[ \cN , \cE_4 (t) \right]  &\leq  \delta \int dx dy N^2 V(N
(x-y)) a_x^* a_y^* a_y a_x + C_\delta \frac{\cN^2}{N} + C_\delta \ech{\|
\varphi_t^{(N)} \|_{H^2}^2}{\ekt} \left(\cN + 1 \right), \\
\pm \dot{\cE}_4 (t) &\leq 
\delta \int dx dy N^2 V(N (x-y)) a_x^* a_y^* a_y a_x + \ech{C_\delta \| \varphi_t^{(N)} \|_{H^4}^2  \frac{\cN^2}{N} \\ &\hspace{.5cm} + C_\delta \left(\| \varphi_t^{(N)} \|_{H^4} \| \varphi_t^{(N)} \|_{H^2} + \| \varphi_t^{(N)} \|_{H^2}^3 \right) \, \left(\cN + 1 \right)}{C_\delta \ekt \left(\frac{\Ncal^2}{N} +\Ncal + 1 \right)}.
\end{split}
\end{equation}
\end{proposition}

To prove Proposition \ref{prop:TL4T}, we will make use of the following two lemmas.
\begin{lemma}\label{lm:4}
Suppose $j_1, j_2 \in L^2 (\bR^3 \times \bR^3)$ are kernels with the property that 
\[ M_i  := \max \left( \sup_x \int dy \, |j_i (x,y)|^2 , \sup_y \int dx \, |j_i (x,y)|^2 \right) < \infty \]
for $i =1 ,2$. Let $j_{i,x} (z) := j_i (z,x)$ and recall the definition $c_x
(z) = \text{ch}_{k_t} (z,x)$ from (\ref{eq:def-csprx}). Then there exists a
constant $C$ depending only on $M_1, M_2$ and on the $L^2$-norms $\| j_1
\|_2$, $\| j_2 \|_2$ such that
\begin{equation}\label{eq:lm4-1}
\int dx dy N^3 V (N (x-y)) \| a^\sharp (j_{1,x}) a_y \psi \|^2 \leq C M_1 \| (\cN+1) \psi \|^2
\end{equation}
and 
\begin{equation}\label{eq:lm4-2}
\int dx dy N^3 V(N (x-y)) \| a^\sharp (j_{1,x}) a^\sharp (j_{2,y}) \psi \|^2 \leq C \min \left(M_1 \| j_2 \|^2_2 , M_2 \| j_1 \|_2^2 \right) \, \| (\cN+1) \psi \|^2 .
\end{equation}
As a consequence
\begin{equation}\label{eq:lm4-3}
\int dx dy N^3 V (N (x-y)) \| a^\sharp (j_{1,x}) a (c_y) \psi \|^2 \leq C M_1 \, \| (\cN+1) \psi \|^2. \end{equation}
The inequalities remain true (and are easier to prove) if both operators act on the same variable. In other words
\begin{equation}\label{eq:lm4-4}
\begin{split}
\int dx dy N^3 V (N (x-y)) \| a^\sharp (j_{1,x}) a_x \psi \|^2 &\leq C M_1
\, \| (\cN+1) \psi \|^2, \\
\int dx dy N^3 V(N (x-y)) \| a^\sharp (j_{1,x}) a^\sharp (j_{2,x}) \psi \|^2
&\leq C \min (M_1 \|j_2 \|_2^2 , M_2 \|j_1\|_2^2 ) \, \| (\cN+1) \psi \|^2, \\
\int dx dy N^3 V (N (x-y)) \| a^\sharp (j_{1,x}) a (c_y) \psi \|^2 &\leq C M_1  \| (\cN+1) \psi \|^2.
\end{split}
\end{equation}
Here $a^\sharp$ is either the annihilation operator $a$ or the creation operator $a^*$. 
\end{lemma}
\begin{proof}
To prove (\ref{eq:lm4-1}), we observe that
\[ \begin{split} \int dx dy N^3 V (N (x-y)) & \| a^\sharp (j_{1,x}) a_y \psi \|^2 \\ \leq \; & \int dx dy N^3 V (N (x-y)) \| j_{1,x} \|_2^2 \| (\cN+1)^{1/2} a_y \psi \|^2 \\ \leq \; & M_1 \int dx dy N^3 V (N (x-y)) \| a_y \cN^{1/2} \psi \|^2 \\ = \; &C M_1 \| \cN \psi \|^2. \end{split} \]
As for (\ref{eq:lm4-2}), we notice that (considering for example the case $a^\sharp (j_{2,y}) = a^* (j_{2,y})$) 
\[ \begin{split} 
\int dx dy N^3 V(N (x-y)) & \| a^\sharp (j_{1,x}) a^* (j_{2,y}) \psi \|^2 \\ \leq \; & \int dx dy N^3 V (N (x-y)) \| j_{1,x} \|_2 \| (\cN+1)^{1/2} a^* (j_{2,y}) \psi \|^2 \\ \leq \; & \int dx dy N^3 V (N (x-y)) \| j_{1,x} \|^2_2 \| a^* (j_{2,y}) (\cN+2)^{1/2} \psi \|^2  \\ \leq \; & \int dx dy N^3 V (N (x-y)) \| j_{1,x} \|^2_2  \| j_{2,y} \|_2^2 \| (\cN+1)^{1/2} (\cN+2)^{1/2} \psi \|^2 \\ \leq \; & C M_1 \| j_2 \|^2_{2} \| (\cN+1) \psi \|^2.  \end{split} \]
Eq.\ (\ref{eq:lm4-3}) follows from the first two, by writing $a (c_y) = a_y + a (p_y)$ (recall here that we are using the notation $p_y (z) = p (k_t) (z,y)$ with the kernel $p (k_t) \in L^2 (\bR^3 \times \bR^3)$ defined in Lemma~\ref{l:kernels}). Eq.\ (\ref{eq:lm4-4}) follows similarly; in this case, however, one can immediately integrate over the variable $y$, simplifying the proof. 
\end{proof}

Terms of the form (\ref{eq:lm4-3}), but with $j_{1,x}$ replaced by $c_x$ (which is not in $L^2$) are treated differently.
\begin{lemma}\label{lm:4b}
Recall the definition $c_x (z) = \text{ch}_{k_t} (z,x)$ from (\ref{eq:def-csprx}). Then there exists a constant $C>0$ with
\begin{equation}\label{eq:lm4b-1}
\begin{split}
\int dx dy N^3 V(N(x-y)) \| a (c_x) a(c_y) \psi \|^2 \leq \; &C \int dx dy N^3 V (N(x-y)) \| a_x a_y \psi \|^2 \\ &+ C\, \| (\cN+1) \psi \|^2. \end{split} \end{equation}
More precisely, we have 
\begin{equation}\label{eq:lm4b-2} \int dx dy N^3 V(N(x-y)) \| a (c_x) a (c_y) \psi \|^2 = \int dx dy N^3 V(N (x-y)) \| a_x a_y \psi \|^2 + \wt{\cE} (t) \end{equation}
where the error $\wt{\cE} (t)$ is such that, for every $\delta > 0$, there exists a constant $C_\delta$ 
with
\[ \pm \wt{\cE} (t) \leq \delta \int dx dy N^3 V (N (x-y)) \| a_x a_y \psi \|^2 + C_\delta \| (\cN+1) \psi \|^2. \]
\end{lemma}
\begin{proof}
We write $a (c_x) = a_x + a (p_x)$, using the notation $p_x (z) = p (k_t) (z,x)$ introduced in (\ref{eq:def-csprx}). We have
\[ \| a (c_x) a (c_y) \psi \| \leq \| a_x a_y \psi \| + \| a_x a (p_y) \psi \| + \| a (p_x) a (c_y) \psi \|. \] 
Therefore, using (\ref{eq:lm4-1}) and (\ref{eq:lm4-3}), we immediately find (using Lemma \ref{l:kernels} to bound $\| p \|_2$ and $\sup_x \| p_x \|_2$)
\[ \int dx dy N^3 V(N(x-y)) \| a(c_x) a(c_y) \psi \|^2 \leq C \int dx dy N^3 V (N (x-y)) \| a_x a_y \psi \|^2 + C \| (\cN+1) \psi \|^2. \]
To prove (\ref{eq:lm4b-2}), we notice that
\[ \begin{split} 
\int dx dy N^3 &V (N(x-y)) \| a (c_x) a(c_y) \psi \|^2 \\ = \; & \int dx dy N^3 V (N(x-y)) \langle \psi , a^* (c_x) a^* (c_y) a(c_y) a(c_x) \psi \rangle \\
= \; & \int dx dy N^3 V (N(x-y)) \langle \psi , a^*_x a^*_y a_y a_x \psi \rangle \\
&+\int dx dy N^3 V (N(x-y))  \Big\langle \psi , \big[ a^* (p_x) a^*_y a_y a_x + a^* (c_x) a^* (p_y) a_y a_x \\ &\hspace{4cm} + a^* (c_x) a^* (c_y) a (p_y) a_x + a^* (c_x) a^* (c_y) a (c_y) a(p_x) \big] \psi \Big\rangle \\
=: \; & \int dx dy N^3 V (N (x-y)) \| a_x a_y \psi \|^2 + \wt{\cE} (t) \end{split}\]
where
\[ \begin{split}
|\wt{\cE} (t)| \leq \; & \int dx dy N^3 V(N(x-y)) \Big[ \| a (p_x) a_y \psi \| \| a_y a_x \psi \| + \| a(c_x) a(p_y) \psi \| \| a_y a_x \psi \| \\ &\hspace{3cm} 
+ \| a (c_x) a (c_y) \psi \| \| a(p_y) a_x \psi \| + \| a (c_x) a(c_y) \psi \| \| a(c_y) a(p_x) \psi \| \Big]  \\ \leq\; & \delta \int dx dy N^3 V(N (x-y)) \left[ \| a_x a_y \psi \|^2 + \| a (c_x) a(c_y) \psi \|^2 \right] \\ &+ C_\delta \int dx dy N^3 V(N (x-y)) \left[ \| a(p_x) a_y \psi \|^2 + \| a(c_x) a(p_y) \psi \|^2 \right] 
\\ \leq \; & \delta \int dx dy N^3 V(N (x-y)) \| a_x a_y \psi \|^2 +
C_\delta \| (\cN+1) \psi \|^2. \end{split} \]
Here, in the last inequality, we used (\ref{eq:lm4-1}), (\ref{eq:lm4-3}) from Lemma \ref{lm:4} and 
(\ref{eq:lm4b-1}).
\end{proof}

\begin{proof}[Proof of Proposition \ref{prop:TL4T}] 
To prove the first bound in (\ref{eq:TL4T-cl}) we observe that all quartic terms on the r.h.s.\ of (\ref{eq:TL4T}) can be bounded using Lemmas \ref{lm:4} and \ref{lm:4b}. For example, the contribution arising from the 
first term on the r.h.s.\ of (\ref{eq:TL4T}) is bounded by
\[ \begin{split} 
\Big| \int dx dy \,  N^2 V (N (x-y)) \langle &\psi ,  a^* (c_x) a^* (c_y) a^* (s_y) a^* (s_x) \psi \rangle \Big| \\\leq \; & \int dx dy N^2 V(N (x-y)) \| a (c_x) a (c_y) \psi \| \| a^* (s_y) a^* (s_x) \psi \| \\
\leq \; & \delta \int dx dy N^2 V (N (x-y)) \| a (c_x) a(c_y) \psi \|^2 \\ &+C_\delta \int dx dy N^2 V(N (x-y)) \| a^* (s_y) a^* (s_x) \psi \|^2 \\
\leq \; & C \delta \int dx dy N^2 V (N (x-y)) \| a_x a_y \psi \|^2 + \frac{C_\delta}{N} \|(\cN+1) \psi \|^2 \end{split} \]
where, in the last inequality, we used (\ref{eq:lm4b-1})  and
(\ref{eq:lm4-2}). All the other quartic terms on the r.h.s.\ of
(\ref{eq:TL4T}), with the exception of the fourth term (the one containing
only $c_x$ or $c_y$ as arguments of the creation and annihilation
operators), can be bounded similarly; the key observation here is that all
these terms have at least one creation or annihilation operator with square
integrable argument (this allow us to apply Lemma \ref{lm:4}). Moreover, in
all these terms, the quartic expression does not contain the annihilation
operators $a(c_x)$ and $a(c_y)$ in the two factors on the left, nor the
creation operators $a^* (c_x)$ and $a^*  (c_y)$ in the two factors on the
right (in Lemma~\ref{lm:4}, in particular in (\ref{eq:lm4-3}) it is of
course important that the factor $a(c_y)$ in the norm appears as an
annihilation and not as a creation operator). To bound the fourth term on the r.h.s.\ of (\ref{eq:TL4T}), where all the arguments of the creation and annihilation operators are not integrable, we cannot apply Lemma \ref{lm:4}. Instead, we use (\ref{eq:lm4b-2}) from Lemma \ref{lm:4b}. We obtain
\[ \begin{split}  \int dx dy N^2 V(N(x-y)) \langle \psi, a^* (c_x) &a^* (c_y) a(c_y) a(c_x) \psi \rangle \\ &=  \int dx dy N^2 V(N (x-y)) \langle \psi , a_x^* a_y^* a_y a_x \psi \rangle + \wt{\cE} (t) \end{split} \]
where the error $\wt{\cE} (t)$ is such that, for every $\delta > 0$, there exists $C_\delta >0$ with
\[ |\wt{\cE} (t)| \leq  \delta \int dx dy N^2 V(N (x-y)) \| a_x a_y \psi \|^2 + \frac{C_\delta}{N} \| (\cN+1) \psi \|^2. \]

The quadratic terms on the r.h.s.\ of (\ref{eq:TL4T}) can be bounded using Lemma \ref{lm:2b}. To this end, we observe that
\[ |\langle s_x, s_y \rangle| , |\langle s_x, c_y \rangle| \leq C N |\varphi^{(N)}_t (x)| |\varphi_t^{(N)} (y)|. \]
It is therefore easy to check that all the quadratic terms, with the exception of the first two (the quadratic terms appearing on the eighth line of (\ref{eq:TL4T})), have a form suitable to apply one of the bounds in Lemma \ref{lm:2b}. More precisely, we apply (\ref{eq:lm2b-1}), if the arguments of the two creation and/or annihilation operators are either $s_x$ or $s_y$. If, on the other hand, one of the two arguments is $c_x$ or $c_y$ and the other one is $s_x$ or $s_y$, we write $a^\sharp (c_x) = a_x + a^\sharp (p_x)$ and then we apply (\ref{eq:lm2b-1}) (to bound the contribution proportional to $a^\sharp (p_x)$) and (\ref{eq:lm2b-2}) (to bound the contribution proportional to $a_x^\sharp$). Finally, if both arguments are either $c_x$ or $c_y$ (and we have  exactly one creation and one annihilation operators), we write $a^\sharp (c_x) = a^\sharp_x + a^\sharp (p_x)$ and we apply (\ref{eq:lm2b-1}), (\ref{eq:lm2b-2}) and (\ref{eq:lm2b-3}). To control the two remaining quadratic contributions, we observe that, writing $a^* (c_x)= a^*_x + a^* (p_x)$,  
\begin{equation} \begin{split}\label{eq:4-2}  \int dx dy N^2 V(N (x-y)) & \langle c_y , s_x \rangle \langle \psi , a^* (c_x) a^* (c_y) \psi \rangle \\ =\; & \int dx dy N^2 V(N (x-y)) \langle c_y , s_x \rangle \langle \psi , a^*_x a^*_y \psi \rangle
\\ &+  \int dx dy N^2 V(N (x-y)) \langle c_y , s_x \rangle \langle \psi , a^* (p_x) a^*_y \psi\rangle
\\ &+  \int dx dy N^2 V(N (x-y)) \langle c_y , s_x \rangle \langle \psi , a^* (c_x) a^*(p_y) \psi \rangle. \end{split} \end{equation}
Since $|\langle c_y , s_x \rangle| \leq C N |\varphi^{(N)}_t (x)| |\varphi^{(N)}_t (y)|$, the last two terms can be bounded (in absolute value) using (\ref{eq:lm2b-2}) and (\ref{eq:lm2b-3}), respectively. We find
\[ \begin{split}
\left|   \int dx dy N^2 V(N (x-y)) \langle c_y , s_x \rangle \langle \psi ,
a^* (p_x) a^*_y \psi\rangle \right| &\leq C \| \varphi_t^{(N)} \|_{H^2}^2 \,
\| (\cN+1)^{1/2} \psi \|^2, \\ 
\left|   \int dx dy N^2 V(N (x-y)) \langle c_y , s_x \rangle \langle \psi , a^* (c_x) a^* (p_y) \psi\rangle \right| &\leq C  \| \varphi_t^{(N)} \|_{H^2}^2 \, \| (\cN+1)^{1/2} \psi \|^2. 
\end{split} \]
As for the first term on the r.h.s.\ of (\ref{eq:4-2}), we notice that 
\[ \langle c_y , s_x \rangle = k (x,y) + g (x,y) \]
where $g (x,y) = r(x,y) + \langle p_y , s_x \rangle$ is such that
\[ |g (x,y) | \leq C |\varphi_t^{(N)} (x)| \, | \varphi_t^{(N)} (y)|. \]
Therefore,
\[ \begin{split} 
\int dx dy N^2 V(N (x-y)) \langle c_y , s_x \rangle \langle \psi , a^*_x a^*_y \psi \rangle =\; & \int dx dy N^2 V(N (x-y)) k(x,y)  \langle \psi , a^*_x a^*_y \psi \rangle \\ &+ \int dx dy N^2 V(N (x-y)) g(x,y) \langle \psi , a^*_x a^*_y \psi \rangle.
\end{split} \]
The second term can be bounded by
\[ \begin{split} 
\Big| \int dx &dy  \, N^2 V(N (x-y)) g (x,y) \langle \psi , a^*_x a^*_y \psi \rangle \Big| \\  \leq \; &C \int dx dy N^2 V(N (x-y)) |\varphi_t^{(N)} (x)| |\varphi_t^{(N)} (y)| \|  a_x a_y \psi \| \, \| \psi \| \\ \leq \; &\delta \int dx dy N^2 V(N (x-y)) \| a_x a_y \psi \|^2 + C_\delta \int dx dy N^2 V(N(x-y)) | |\varphi_t^{(N)} (x)|^2 |\varphi_t^{(N)} (y)|^2  \\ \leq \; & \delta \int dx dy N^2 V(N(x-y)) \| a_x a_y \psi \|^2 + C_\delta \| \varphi_t^{(N)} \|_{H^2}^2.  \end{split} \]
Proceeding analogously to control the second quadratic term on the r.h.s.\ of (\ref{eq:TL4T}), we conclude that 
\[ \begin{split} 
\int dx dy \, N^2 V (N (x-y)) & \left[ \langle c_y, s_x \rangle \,  a^* (c_x) a^* (c_y)  + \langle s_x, c_y \rangle \, a (c_y)  a(c_x) \right] \\  = \; & \int dx dy \, N^2 V(N(x-y))  \left[ k(x,y) a^*_x a^*_y + \overline{k} (x,y) a_x a_y \right] \\ &+ \wt{\cE} (t) \end{split} \]
where the error $\wt{\cE} (t)$ is such that, for every $\delta >0$ there exists $C_\delta$ with 
\[ \pm \, \wt{\cE} (t) \leq  \delta \int dx dy N^2 V(N (x-y)) a_x^* a_y^* a_y a_x + C_\delta \| \varphi_t^{(N)}  \|_{H^2}^2 \,  (\cN + 1). \]
We then use \eqref{eq:hireg} to conclude the proof of the first bound in (\ref{eq:TL4T-cl}).

The proof of the second inequality in (\ref{eq:TL4T-cl}) is analogous, because commuting the terms contributing to $\cE_4 (t)$ with the number of particles operator $\cN$ either gives zero or leaves the terms essentially invariant (up to a constant and a possible sign change). Finally, also the third estimate in (\ref{eq:TL4T-cl}) can be proven similarly, because the time derivative of the terms contributing to $\cE_4 (t)$ can be expressed as linear combination of terms having the same form, just with one argument $c_x$, $c_y$, $s_x$ or $s_y$ replaced by its time-derivative. These terms can then be handled as above, using however the bounds for $\| \dot{\text{sh}}_{k_t} \|_2$ and $\| \dot{p}_{k_t} \|_2$ from Lemma \ref{lm:dotk} and the bound for $\| \dot{\varphi}_t^{(N)} \|_{H^2}$ from Proposition \ref{t:pdes}.
\end{proof}

\subsection{Analysis of $[i \partial_t T^*] T$}

We set 
\[ B = \frac{1}{2} \int dx dy \, \left( k_t (x,y) a_x^* a_y^* - \overline{k_t (x,y)} a_x a_y \right) \]
and 
\[ \dot{B} =  \frac{1}{2} \int dx dy \, \left( \dot{k}_t (x,y) a_x^* a_y^* - \overline{\dot{k}_t (x,y)} a_x a_y \right) \]
with 
\[ k_t (x,y) = -N w(N(x-y)) \varphi_t^{(N)} (x) \varphi_t^{(N)} (y) \]
and
\[ \dot{k}_t (x,y) = - N w (N (x-y))  \left( \dot{\varphi}_t^{(N)} (x) \varphi_t^{(N)} (y) + \varphi_t^{(N)} (x) \dot{\varphi}_t^{(N)} (y) \right). \]
Then $T = \exp (B)$. Using \eqref{eq:baker}, we find%(for $\ad^0_B (\dot B) = \dot B$ and $\ad^n_B (\dot B) = [B,\ad^{n-1}_B (\dot B)]$)
\begin{equation}\label{eq:BCH} \left(\partial_t T^* \right) T = \int_0^1 \di\lambda\, e^{-\lambda B(t)} \dot B(t) e^{\lambda B(t)} = \sum_{n \geq 0} \frac{1}{(n+1)!}\,  \text{ad}^n_{B} (\dot{B}) \end{equation}
up to an error term, which however converges to zero in expectation on the
domain $D(\cN)$ of the number of particles operator (this can be shown as in
Lemma \ref{l:bt}). Notice that, by the estimates in Prop. \ref{prp:dTT}, the series is absolutely convergent in expectations. Since $D(\cN)$ is invariant w.\,r.\,t.\ the fluctuation dynamics $\cU (t,s)$ (this is proven similarly to Prop.\ \ref{prop:apri})\done{}, we can use (\ref{eq:BCH}) to compute the expectation of $(\partial_t T^*) T$ in the state $\cU (t;0) \psi$ for any $\psi \in D(\cN)$. 

Next, we compute the terms on the r.h.s.\ of (\ref{eq:BCH}).

\begin{lem}
\label{lm:highercommutators}
 For each $n \in \Nbb$ there exist $f_{n,1}, f_{n,2} \in L^2(\Rbb^3 \times \Rbb^3)$ such that
\begin{equation}\label{eq:adnBB} \begin{split}
\ad^n_B(\dot B) &= \frac{1}{2} \int \di x\di y\left( f_{n,1}(x,y) a^\ast_y a^\ast_x + f_{n,2}(x,y) a_x a_y \right) \quad \text{for all even $n$ and }\\
\ad^n_B(\dot B) &= \frac{1}{2} \int \di x\di y\left( f_{n,1}(x,y) a^\ast_x a_y + f_{n,2}(x,y) a_x a^\ast_y \right) \quad \text{for all odd $n$}
\end{split} \end{equation}
where \begin{equation}\label{eq:est-fn}  \| f_{n,i} \|_2  \leq 2^n \norm{k_t}_2^n \| \dot{k}_t \|_2 , 
\end{equation}
for all $n \geq 0$ and $i=1,2$, 
\begin{equation}\label{eq:est-dotfn}
\| \dot{f}_{n,i} \|_2 \leq \left\{ \begin{array}{ll} \| \ddot k_t \|_2 \quad & \text{if } n = 0 \\ 
4^n \| k_t \|_2^{n-1} \, \left( \| \ddot{k}_t \|_2  \| k_t \|_2 + \|\dot{k}_t \|_2^{2} \right) \quad & \text{if } n \geq 1 \end{array} \right. 
\end{equation}
and
\begin{equation}\label{eq:est-fn2} \int dx  \, |f_{n,i} (x,x)|  \leq  2^n \norm{k_t}_2^n \| \dot{k}_t \|_2 , \quad \int dx \, | \dot{f}_{n,i} (x,x)| \leq  4^n \| k_t \|_2^{n-1} \, \left( \| \dot{k}_t \|_2^2 + \| \ddot{k}_t \|_2 \| k_t \|_2\right)
\end{equation}
for all $n \geq 1$.  
\end{lem}
\begin{proof} The proof is by induction in $n$. For $n = 0$,
\bd
\ad^0_B(\dot B) = \dot B = \frac{1}{2}\int \di x\di y\left( \dot k_t (x,y) a^\ast_x a^\ast_y - \cc{\dot k_t (x,y)} a_x a_y \right).
\ed
Hence $f_{0,1} (x,y) = \dot{k}_t (x,y)$ and $f_{0,2} (x,y) = \cc{\dot{k}_t (x,y)}$, and the estimates (\ref{eq:est-fn}) and (\ref{eq:est-dotfn}) are clearly satisfied. Suppose now the statement holds for some $n \in \bN$. 
We prove them for $(n+1)$. We assume first that $n$ is even. Then, using the canonical commutation relations (\ref{eq:ccr}), we find
\begin{align*}
& \ad^{n+1}_B(\dot B) \\ & = [B,\ad^n_B(\dot B)] \\
& = \left[\frac{1}{2} \int \di x\di y\left( k_t (x,y)a^\ast_x a^\ast_y - \cc{k_t (x,y)}a_x a_y \right), \frac{1}{2}\int \di x\di y\left( f_{n,1}(x,y) a^\ast_x a^\ast_y + f_{n,2}(x,y) a_x a_y \right)\right] \\
& = \frac{1}{2} \int \di x\di z \left(f_{n+1,1}(x,z) a^\ast_x a_z + f_{n+1,2}(x,z) a_x a^\ast_z \right),
\end{align*}
where
\begin{equation}
\label{eq:even}
\begin{split}
f_{n+1,1}(x,z) & = -\frac{1}{2} \int \di y \left( k_t (x,y) \left(
f_{n,2}(z,y) + f_{n,2}(y,z) \right) + \cc{k_t (y,z)}\left( f_{n,2}(x,y) +
f_{n,2}(y,x)\right) \right),\\
f_{n+1,2}(x,z) & = -\frac{1}{2} \int \di y \left( k_t (y,z) \left( f_{n,1}(x,y) + f_{n,1}(y,x) \right) + \cc{k_t (x,y)}\left( f_{n,1}(z,y) + f_{n,1}(y,z)\right) \right) .
\end{split}
\end{equation}
By Cauchy-Schwarz (similarly to (\ref{eq:CS-k})), we have
\be{normnorm}
\begin{split}
\norm{f_{n+1,1}}_{2} & \leq 2 \| k_t  \|_2 \| f_{n,2} \|_2 \leq 2^{n+1} \|
k_t \|^{n+1}_2 \| \dot k_t \|_2, \\
\| f_{n+1 ,2} \|_2 & \leq 2 \| k_t \|_2 \| f_{n,1} \|_2 \leq 2^{n+1} \| k_t
\|^{n+1}_2 \| \dot k_t \|_2,
\end{split} 
\end{equation}
where we used the induction assumption. Moreover, again by Cauchy-Schwarz,
\[ \int dx \, |f_{n+1 ,1} (x,x)|  \leq 2 \| k_t \|_2 \| f_{n,2} \|_2 \leq 2^{n+1} \| k_t \|^{n+1}_2 \| \dot k_t \|_2 \,. \]
As for the time-derivative of $f_{n+1,i}$, we find
\[ \begin{split}
\dot{f}_{n+1,1}(x,z) = &-\frac{1}{2} \int \di y \Big( \dot{k}_t (x,y) \left( f_{n,2}(z,y) + f_{n,2}(y,z) \right) + k_t (x,y) \left( \dot{f}_{n,2} (z,y) + \dot{f}_{n,2} (y,z) \right) \\ &\hspace{1cm} + \cc{\dot{k}_t (y,z)} \left( f_{n,2}(x,y) + f_{n,2}(y,x)\right) + \cc{k_t (y,z)} \left( \dot{f}_{n,2}(x,y) + \dot{f}_{n,2}(y,x)\right) \Big) \end{split}\]
and similarly for $\dot{f}_{n+1,2}$. Hence, we find
\[ \begin{split} 
\| \dot{f}_{n+1,1} \|_2 \leq \; &2 \left( \| \dot{k}_t \|_2 \| f_{n,2} \|_2 + \| k_t \|_2 \| \dot{f}_{n,2} \|_2 \right)  \\
\leq \; & 2 \left( 2^n \| k_t \|_2^n \| \dot{k}_t \|_2^2 + 4^n \| k_t \|_2^n \left( \| \ddot k_t \|_2 \| k_t \|_2 + \| \dot k_t \|_2^2 \right) \right) \\
\leq \; & (2^{n+1} + 4^n) \| k_t \|_2^n \| \dot k_t \|_2^2 + 2 \cdot 4^n \| k_t \|_2^{n+1} \| \ddot k_t \|_2 \\ \leq \; & 4^{n+1} \| k_t \|_2^n \left( \| \ddot k_t \|_2 \| k_t \|_2 + \| \dot k_t \|_2^2 \right),
\end{split} \]
proving (\ref{eq:est-dotfn}) for $i = 1$. The same bound for $i=2$ and the second bound in (\ref{eq:est-fn2}) for $i=1,2$ can be proven similarly. 

If $n$ is odd, we have, using again the canonical commutation relations, 
\begin{align*}
& \ad^{n+1}_B(\dot B)\\
& = \left[ \frac{1}{2}\int \di x \di y \left( k_t (x,y)a^\ast_x a^\ast_y - \cc{k_t (x,y)} a_x a_y \right) , \frac{1}{2}\int \di x\di y \left( f_{n,1}(x,y) a^\ast_x a_y + f_{n,2}(x,y) a_x a^\ast_y \right) \right] \\
& = \frac{1}{2} \int \di x\di z \big( a^\ast_x a^\ast_z f_{n+1,1}(x,z) + a_x a_z f_{n+1,2}(x,z) \big)
\end{align*}
where
\begin{equation}
\label{eq:odd}
\begin{split}
f_{n+1,1}(x,z) & = - \int \di y\, k_t (x,y)\left( f_{n,1}(z,y) +
f_{n,2}(y,z) \right), \\
f_{n+1,2}(x,z) & = - \int \di y\, \cc{k_t (x,y)}\left( f_{n,1}(y,z) + f_{n,2}(z,y) \right).
\end{split}
\end{equation}
The bounds (\ref{eq:est-fn}), (\ref{eq:est-dotfn}), (\ref{eq:est-fn2}) follow as above.
\end{proof}

Using Lemma \ref{lm:highercommutators}, we obtain the following properties of $(\partial_t T^*)T$.
\begin{proposition}\label{prp:dTT}
There exists a constant $C>0$ with 
\begin{equation}\begin{split}
\label{eq:prT-cl1}  \pm \, (i\partial_t T^*) T &\leq C \ech{\|
\varphi_t^{(N)} \|_{H^2}}{\ekt} \,  (\cN +1), \\
\pm \, \left[ \cN , (i\partial_t T^*) T \right]  &\leq C \ech{\|
\varphi_t^{(N)} \|_{H^2}}{\ekt} \, (\cN + 1), \\
\pm \partial_t \left[ (i\partial_t T^*) T \right] &\leq C \ech{(\| \varphi_t^{(N)} \|^2_{H^2} + \| \varphi_t^{(N)} \|_{H^4})}{\ekt} \, (\cN +1). 
\end{split}
\end{equation}
\end{proposition}

\begin{proof}
We observe first of all that, for all $f_1, f_2 \in L^2 (\bR^3 \times\bR^3)$ with $\int dx \, |f_2 (x,x)| < \infty$,
we have
\begin{equation}\label{eq:prT-1}  \left| \left\langle \psi, \int dx dy \left( f_1(x,y) a^\ast_x a^\ast_y + f_2(x,y) a_x a_y \right) \psi \right\rangle \right|  \leq (\| f_1 \|_2 + \| f_2 \|_2) \, \scal{\psi}{(\Ncal+1)\psi} 
\end{equation}
and
\begin{equation}\label{eq:prT-2} \begin{split} \Big|  \big\langle \psi , \int dx dy \, ( f_1(x,y) a^\ast_x
& a_y + f_2(x,y) a_x a^\ast_y ) \psi \big\rangle \Big|  \\ & \leq (\| f_1 \|_2 + \| f_2 \|_2) \scal{\psi}{\Ncal\psi} + \int |f_2(x,x)| \di x  \, \| \psi \|^2 .
\end{split} \end{equation}
In fact, (\ref{eq:prT-1}) follows because
\[ \begin{split} 
 \Big| \big\langle \psi, \int dx dy \, ( f_1(x,y) a^\ast_x &a^\ast_y + f_2(x,y) a_x a_y) \psi \big\rangle \Big| \\
 \leq \; & \int dx \, \left( \| a_x \psi \| \| a^* (f_1(x,.)) \psi \| +  \, \| a^* (f_2 (x,.)) \psi \| \| a_x \psi \| \right) \\
 \leq \; & \| (\cN +1)^{1/2} \psi \|  \int dx \, (\| f_1(x,.) \|_2 + \| f_2 (x,.)\|_2) \,  \| a_x \psi \| \\
 \leq \; & (\| f_1 \|_2 + \| f_2 \|_2) \| (\cN+1)^{1/2} \psi \|^2.
 \end{split} \]
 Eq.\ (\ref{eq:prT-2}) can be proven similarly. Combining the last estimates with Lemma \ref{lm:highercommutators} and with (\ref{eq:BCH}), we find
 \[ \begin{split} 
| \langle \psi , (\partial_t T^*)T \psi \rangle|  \leq \; &\sum_{n \geq 0} \frac{1}{(n+1)!} \left| \langle \psi , \text{ad}^n_B (\dot{B}) \psi \rangle \right| \\
\leq \; &\sum_{n \geq 0} \frac{1}{(n+1)!} \left( \| f_{n,1} \|_2 + \| f_{n,2} \|_2 \right)  \| (\cN+1)^{1/2} \psi \|^2 \\& + \sum_{n\geq 1} \frac{1}{(2n)!}  \left(\int dx \, |f_{2n-1,2} (x,x)| dx \right)  \| \psi \|^2 \\
\leq \; &C \sum_{n\geq 0} \frac{(2 \| k_t \|_2)^n}{(n+1)!} \, \| \dot{k}_t \|_2 \, \| (\cN+1)^{1/2} \psi \|^2 + %\\ &+ 
\sum_{n\geq 1} \frac{(2\| k_t \|_2)^n}{(2n)!}  \| \dot k_t \|_2 \| \psi \|^2 \\ 
\leq \; & C e^{2\| k_t \|_2} \| \dot{k}_t \|_2 \, \| (\cN+1)^{1/2} \psi \|^2 \\
\leq \; & C \ech{\| \varphi_t^{(N)} \|_{H^2}}{\ekt} \, \| (\cN+1)^{1/2} \psi \|^2 
\end{split}
\]
using also Lemma \ref{lm:dotk}. The second inequality in (\ref{eq:prT-cl1}) follows similarly because, essentially, the only consequence of taking the commutator with $\cN$ is to eliminate the terms $\text{ad}^n_B (\dot{B})$ for all odd $n$. Also the third bound in (\ref{eq:prT-cl1}) can be proven analogously, taking the time derivative of the expressions for $\text{ad}^n_B (\dot{B})$ given in (\ref{eq:adnBB}), using the bounds for $\| \dot f_{n,i} \|_2$ in (\ref{eq:est-dotfn}) and (\ref{eq:est-fn2}) and, finally, using the estimate for $\| \ddot k_t \|_2$ proven in Lemma \ref{lm:dotk}.
\end{proof}

\subsection{Proof of Theorem \ref{thm:L}}
\label{sub:thmL}

In this section, we combine the results of the previous subsection, to obtain a proof of Theorem \ref{thm:L}. {F}rom (\ref{eq:TL1T}) and Propositions \ref{prop:TKT}, \ref{prop:TwtL2T}, \ref{prop:TL3T}, \ref{prop:TL4T} it follows that
\begin{equation}\label{eq:cLN-can1} \begin{split} 
\cL_N (t) = \; &C_N (t) + \cK  + \frac{1}{2N}  \int dx dy N^3 V (N (x-y)) a_x^* a_y^* a_y a_x \\
&+ \Big[ N^3 \int dx dy (\Delta w) (N (x-y)) \varphi_t^{(N)} (x) \varphi_t^{(N)} (y)  a_x^* a_y^* \\ 
&\quad\;+ \frac{1}{2} \int dx dy N^3 V(N (x-y)) (1- w(N(x-y))) \varphi_t^{(N)} (x) \varphi_t^{(N)} (y)\, a_x^* a_y^*% \\&\hspace{.5cm}
+ \text{h.c.} \Big]  
\\ &+ \cE (t) 
\end{split} \end{equation}
where the constant $C_N (t)$ is defined in (\ref{eq:CNt}) and the error $\cE (t)$ is such that, for every $\delta > 0$ there exists $C_\delta > 0$ with
\begin{equation}\label{eq:cE-tot}
\begin{split} \pm \cE (t) &\leq \delta \left( \cK + \int dx dy N^2 V(N
(x-y)) a_x^* a_y^* a_y a_x \right) + C_\delta \frac{\cN^2}{N} + C_\delta
\ech{\| \varphi_t^{(N)} \|_{H^2}^2}{\ekt} \, \left( \cN + 1 \right), \\
\pm \left[ \cN ,  \cE (t)\right]  &\leq \delta \left( \cK + \int dx dy N^2
V(N (x-y)) a_x^* a_y^* a_y a_x \right) + C_\delta \frac{\cN^2}{N}  +
C_\delta \ech{\| \varphi^{(N)}_t \|_{H^2}^2}{\ekt} \left( \cN + 1 \right),\\
\pm \dot{\cE} (t) & \leq \delta \left( \cK + \int dx dy N^2 V(N (x-y)) a_x^* a_y^* a_y a_x \right) + \ech{C_\delta \| \varphi_t^{(N)} \|_{H^4}^2 \frac{\cN^2}{N} \\ &\hspace{.5cm} + C_\delta \left( \| \varphi_t^{(N)} \|_{H^4}  \| \varphi_t^{(N)} \|_{H^2} + \| \varphi_t^{(N)} \|_{H^2}^3 \right) \left(\cN+1 \right)}{C_\delta \ekt \left( \frac{\Ncal^2}{N}+\Ncal+1\right)}. 
\end{split} \end{equation}
In deriving (\ref{eq:cLN-can1}), we made use of the crucial cancellation between the linear contributions in (\ref{eq:TL1T}) and the linear terms in (\ref{eq:pr-TL3T}).  Next, we notice another crucial cancellation. The terms on the third and fourth line in (\ref{eq:cLN-can1}) can be written as
\[ \begin{split} 
N^3 \int dx dy \, a_x^* a_y^* \left[ \left(-\Delta +\frac{1}{2} V \right)(1-w) \right](N(x-y)) \varphi_t^{(N)} (x) \varphi_t^{(N)} (y) = 0 
\end{split}\]
since $f = 1-w$ is a solution of the zero-energy scattering equation $(-\Delta + (1/2) V)f = 0$.

We conclude that
\begin{equation}\label{eq:L-fin} \cL_N (t) = C_N (t) + \cK  + \frac{1}{2N}  \int dx dy N^3 V (N (x-y)) a_x^* a_y^* a_y a_x + \cE (t) \end{equation}
where the error $\cE (t)$ satisfies (\ref{eq:cE-tot}). Then (\ref{eq:thmL-1}) follows from the first bound in (\ref{eq:cE-tot}), taking $\delta = 1/2$. Also (\ref{eq:thmL-2}) and (\ref{eq:thmL-3}) follow from the second and third bounds in (\ref{eq:cE-tot}), since both $\cK$ and the quartic term on the r.h.s.\ of (\ref{eq:L-fin}) commute with $\cN$ and are time-independent. 

This concludes the proof of Theorem \ref{thm:L}.

\appendix

\section{Properties of the solution of the Gross-Pitaevskii equation}
\label{s:pde}

\begin{proof}[Proof of Proposition \ref{t:pdes}] 
(i) This part of the proposition is standard. One proves first local well-posedness of the two equations in $H^1 (\bR^3)$. The time of existence depends only on the $H^1$-norm of the initial data. Since $V,f \geq 0$, the $H^1$-norm is bounded by the energy, which is conserved.  
%\[ \cE (\ph) = \int dx |\nabla \varphi|^2 + \frac{1}{2} \int dx dy \, N^3 V(N(x-y)) f(N(x-y)) |\varphi (x)|^2 |
%\varphi (y)|^2 \]
%which is conserved. 
Hence one obtains global existence and a uniform bound on the $H^1$-norm. 

(ii) Also this part is rather standard, but since the non-linearity in (\ref{eq:mod-GP}) depends on $N$, and we need bounds uniform in $N$, we sketch the proof of the bound (\ref{eq:hireg}) for $\| \varphi_t^{(N)} \|_{H^n}$ (the bound for $\| \varphi_t \|_{H^n}$ can be proven analogously). We present the proof for the case $t >0$. We claim, first of all, that there exists $T>0$ depending only on $\| \varphi \|_{H^1}$ and $n \in \bN$ such that
  \begin{equation}\label{eq:preregularity}
    \sup_{t \in [0,T]} \| \varphi_t^{(N)} \|_{H^n} \le 2 \| \varphi_0 \|_{H^n} +
     \sup_{t \in [0,T]} \| \varphi_t^{(N)} \|_{H^{n-1}}^3.\done
 \end{equation}
 Introducing the short-hand notation $U_N (x)= N^3 V(Nx)f(Nx)$, we write the solution 
$\varphi_t^{(N)}$ of (\ref{eq:mod-GP}) as
\[
\varphi^{(N)}_t = e^{it\Delta} \varphi - i \int_0^t ds \, e^{i(t-s)\Delta} (U_N *
|\varphi^{(N)}_s|^2) \varphi^{(N)}_s.\]
Differentiating this equation w.r.t. the spatial variables we find that
\[
    \partial^\alpha \varphi^{(N)}_t = e^{it \Delta} \partial^\alpha \varphi - i
    \int_0^t ds \, e^{i(t-s) \Delta} \sum_{\beta \le \alpha} \sum_{\nu \le
    \beta} \binom{\alpha}{\beta} \binom{\beta}{\nu} \big( U_N * (\partial^\nu
    \overline{\varphi}^{(N)}_s \partial^{\beta - \nu} \varphi^{(N)}_s) \big)
    \partial^{\alpha - \beta} \varphi^{(N)}_s. \]
Here $\alpha$ is a three-dimensional multi-index of non-negative integers, with $|\alpha| \leq n$.
  
The $L_t^\infty ([0,T] , L_x^2)$-norm of the above expression can be controlled using
Strichartz estimates for the free Schr\"odinger evolution $e^{it\Delta}$
(see \cite[Theorem 1.2]{KT}). We find
\[\begin{split}
\| \partial^\alpha & \varphi^{(N)}_{(\cdot)} \|_{L_t^\infty L_x^2}  \\ \leq \; & \|
    \partial^\alpha \varphi \|_{L^2} +  \sum_{\beta \le \alpha} \sum_{\nu \le \beta} \binom{\alpha}{\beta}
    \binom{\beta}{\nu} \| \big( U_N * (\partial^\nu
    \overline{\varphi}^{(N)}_{(\cdot)} \partial^{\beta - \nu} \varphi^{(N)}_{(\cdot)})
    \big) \partial^{\alpha - \beta} \varphi^{(N)}_{(\cdot)} \|_{L_t^2 L_x^{6/5}} \\
     \leq \; & \|
    \partial^\alpha \varphi \|_{L^2} +  T^{1/2} \sum_{\beta \le \alpha} \sum_{\nu \le \beta} \binom{\alpha}{\beta}
    \binom{\beta}{\nu}  \sup_{t \in [0,T]} \| \big( U_N * (\partial^\nu
    \overline{\varphi}^{(N)}_{t} \partial^{\beta - \nu} \varphi^{(N)}_{t})
    \big) \partial^{\alpha - \beta} \varphi^{(N)}_{t} \|_{L_x^{6/5}}.
\end{split} \]    
By H\"older and Young inequality, we find 
\[ \begin{split}
\| \partial^\alpha & \varphi^{(N)}_{(\cdot)} \|_{L_t^\infty L_x^2} \\ \leq \; & \|
    \partial^\alpha \varphi \|_{L^2} + C T^{1/2}  \sum_{\beta \le \alpha} \sum_{\nu \le \beta}
    \binom{\alpha}{\beta} \binom{\beta}{\nu} \sup_{t \in [0,T]} \|
    \partial^\nu \varphi_t^{(N)} \|_{L^{p_1}} \| \partial^{\beta - \nu} \varphi_t^{(N)}
    \|_{L^{p_2}} \, \| \partial^{\alpha - \beta} \varphi_t^{(N)}
    \|_{L^{p_3}} 
  \end{split}\]
for $p_1, p_2, p_3 \geq 1$ with $p_1^{-1} + p_2^{-1} + p_3^{-1} = 5/6$. It is important to note that the indices $(p_1, p_2, p_3)$ can be chosen differently for each term in the summation. In some of the terms with $\lvert \alpha\rvert = n$, all $n$ derivatives hit the same $\varphi_t^{(N)}$. Since $1/6 + 1/6 + 1/2 = 5/6$, these terms can be bounded by
\begin{equation}\label{eq:multi-a}  \| \varphi_t^{(N)} \|^2_{L^6} \; \| \partial^\alpha \varphi_t^{(N)} \|_{L^2} \leq C \| \varphi_t^{(N)} \|_{H^n} \end{equation} for $C$ depending only on $\| \varphi \|_{H^1}$ (recall here that the $H^1$-norm is bounded uniformly in $t$, by part (i)). In some of the other terms, one $\varphi_t^{(N)}$ has $n - 1$ derivatives, one has at most one derivative and the last one has no derivatives. Since $\| \partial^{\gamma} \varphi_t^{(N)} \|_{L^6} \leq \| \varphi_t^{(N)} \|_{H^n}$, if $|\gamma| \leq n-1$, these terms are bounded by the r.h.s.\ of (\ref{eq:multi-a}). In all other terms, the three copies of $\varphi_t^{(N)}$ have at most $n-2$ derivatives. These terms are bounded by 
\[  \| \partial^{\gamma_1} \varphi_t^{(N)} \|_{L^6} \| \partial^{\gamma_2} \varphi_t^{(N)} \|_{L^6} \| \partial^{\gamma_3} \varphi_t^{(N)} \|_{L^2} \leq C \| \varphi_t^{(N)} \|_{H^{n-1}}^3 \]
for all multiindices $\gamma_1, \gamma_2, \gamma_3$ with $|\gamma_i| \leq n-2$, for $i=1,2,3$. We conclude that
\[ \| \partial^\alpha \varphi^{(N)}_{(\cdot)} \|_{L_t^\infty L_x^2} \leq  \|
    \partial^\alpha \varphi \|_{L^2} + C T^{1/2} \sup_{t \in [0,T]} \| \varphi_t^{(N)} \|_{H^n} + C T^{1/2} \sup_{t \in [0,T]} \| \varphi_t^{(N)} \|_{H^{n-1}}^3. \]   
Summing over all $\alpha$ with $|\alpha| \leq n$, we find
\[ \sup_{t \in [0,T]} \| \varphi^{(N)}_{t} \|_{H^n} \leq  \| \varphi \|_{H^n} + C T^{1/2} \sup_{t \in [0,T]} \| \varphi_t^{(N)} \|_{H^n} + C T^{1/2} \sup_{t \in [0,T]} \| \varphi_t^{(N)} \|_{H^{n-1}}^3. \]  
Choosing $T > 0$ so small that $CT^{1/2} \leq 1/2$, we find (\ref{eq:preregularity}).

To show (\ref{eq:hireg}), we iterate now (\ref{eq:preregularity}). We proceed by induction over $n$. For $n =1$, the claim follows from part (i). Suppose now that $\| \varphi^{(N)}_t \|_{H^{(n-1)}} \leq C_{n-1} \exp (K_{n-1} |t|)$, for constants $C_{n-1}, K_{n-1}$ depending on $\| \varphi \|_{H^{(n-1)}}$ and, respectively, on $\| \varphi \|_{H^1}$. Let $T$ be as in (\ref{eq:preregularity}). For an arbitrary $t >0$, there exists an integer $j \in \bN$ such that $(j-1) T < t \leq jT$. Then
  \begin{align*}
    \| \varphi_t^{(N)} \|_{H^n} & \le \sup_{s \in [(j-1)T, jT]} \| \varphi_s^{(N)}
    \|_{H^n}\done \\
    & \le 2 \| \varphi_{(j-1)T}^{(N)} \|_{H^n} + 2\sup_{s \in [(j-1)T, jT]} \|
    \varphi_s^{(N)} \|_{H^{n-1}}^3 \\
    & \le 2 \| \varphi_{(j-1)T}^{(N)} \|_{H^n} + 2 C_{n-1}^3 e^{3 K_{n-1} j T}.
  \end{align*}
Similarly we have
  \[
    \| \varphi_{(j-1)T}^{(N)} \|_{H^n} \le 2 \|
    \varphi_{(j-2)T}^{(N)} \|_{H^n} + 2 C_{n-1}^3  e^{3 K_{n-1} (j-1)T}.\done
  \]
Iterating $j$-times, we obtain
  \begin{align*}
    \| \varphi_t^{(N)} \|_{H^n} & \le 2^j \| \varphi \|_{H^n} + 2 C_{n-1}^3
    \sum_{\ell=0}^j 2^{\ell} e^{3 K_{n-1} (j-\ell)T} \le C_n e^{K_n t},\done
  \end{align*}
  for some constant $C_n$ depending on $\| \varphi \|_{H^n}$ and $K_n$ depending only on $\| \varphi \|_{H^1}$.

(iii) {F}rom the modified Gross-Pitaevski equation (\ref{eq:mod-GP}), letting $U_N (x) = N^3 V(Nx) f(Nx)$, we find
\[ \begin{split}
\norm{\phdot}_{2} \leq \; & \norm{\ph}_{H^2} + \left\| \left(U_N * \lvert \ph\rvert^2 \right)\ph \right\|_2 \\  \leq \; & \norm{\ph}_{H^2} + \left\| U_N * \lvert \ph\rvert^2 \right\|_2 \| \ph \|_\infty \\ \leq \; &
 \norm{\ph}_{H^2} + C \| U_N \|_1 \| \ph \|_4^2 \, \| \ph \|_\infty \\ \leq \;  &C \| \ph \|_{H^2}^3 \leq C e^{K |t|} 
\end{split} \]
for a constant $C$ depending only on $\| \varphi \|_{H^2}$ and $\| U_N \|_1$, and for $K >0$ depending only on $\| \varphi \|_{H^1}$. Here we used part (ii). Applying a gradient to (\ref{eq:mod-GP}), we find 
\begin{equation}\label{eq:nabladot} \begin{split}
i \nabla \phdot & = - \nabla \Delta \ph + \left( U_N \ast \lvert \ph \rvert^2 \right) \nabla \ph \\
& \quad + \left( U_N \ast \overline{\varphi}_t^{(N)} \nabla \ph \right) \ph + \left( U_N \ast  \nabla \overline{\varphi}_t^{(N)} \ph \right) \ph. 
\end{split} \end{equation}
Clearly, $\norm{\nabla \Delta \ph}_{2} \leq \norm{\ph}_{H^3}$. The second term on the first line is bounded in norm by
\[\begin{split}  \left\| (U_N * |\varphi_t^{(N)}|^2) \nabla \ph \right\|_2 &
\leq \| (U_N * |\varphi_t^{(N)}|^2)  \|_{\infty}   \| \nabla \ph \|_2 \\
&\leq  \|U_N\|_1 \|\varphi_t^{(N)}\|_\infty^2    \| \nabla \ph \|_2 \leq C
\| \ph \|_{H^2}^3. \end{split} \]
The terms on the second line of (\ref{eq:nabladot}) can be bounded similarly. {F}rom part (ii), we conclude that $\| \nabla \dot{\varphi}_t^{(N)} \|_2 \leq C \exp (K |t|)$. Analogously, we can also show that $\| \nabla^2 \dot{\varphi}_t^{(N)} \|_2 \leq C e^{K |t|}$. We conclude that $\| \dot{\varphi}^{(N)}_t\|_{H^2} \leq C e^{K |t|}$. Finally, (\ref{eq:mod-GP}) implies 
\[ \begin{split}
-\phddot  = \; &-\Delta i \phdot + \left( U_N \ast (\overline{\varphi}_t^{(N)} i \phdot)\right) \ph \\
&+ \left( U_N \ast (i\dot{\overline{\varphi}}_t^{(N)} \ph) \right) \ph + \left( U_N \ast \lvert\ph\rvert^2 \right) i \phdot.
\end{split}\]
Plugging in the r.h.s.\ of (\ref{eq:mod-GP}) for $i\phdot$, we arrive at 
\begin{align*}
- \phddot & = \Delta^2 \ph - \Delta \left( (U_N \ast \lvert \ph\rvert^2)\ph \right) + \left( U_N \ast \lvert \ph\rvert^2 \right)^2 \ph \\
& \quad + (U_N \ast \overline{\varphi}_t^{(N)} (-\Delta \ph))\ph + 2 \left[ U_N \ast \left( \lvert \ph\rvert^2 (U_N \ast \lvert \ph\rvert^2) \right) \right] \ph \\
& \quad + (U_N \ast (-\Delta \overline{\varphi}_t^{(N)}) \ph ) \ph + \left( U_N \ast \lvert \ph \rvert^2 \right)(-\Delta \ph).
\end{align*}
Proceeding similarly as above, we find that $\| \ddot \varphi_t^{(N)} \|_{2} \leq C \| \varphi_t^{(N)} \|_{H^4} \leq C \exp (K |t|)$. 

(iv) Using (\ref{eq:GP}) and (\ref{eq:mod-GP}), we find
\begin{equation}\label{ddt}
    \begin{aligned}
      \partial_t \| \varphi_t - \varphi_t^{(N)} \|_{2}^2 & = -2 \Im \left\langle
      \varphi_t, \big(U_N * |\varphi_t^{(N)}|^2 - 8 \pi a_0 |\varphi_t|^2\big)
      \varphi_t^{(N)} \right\rangle \\
      & = - 2 \Im \left\langle \varphi_t, \big(U_N * |\varphi_t|^2 - 8 \pi a_0
      |\varphi_t|^2\big) \varphi_t^{(N)} \right\rangle \\
      & \quad - 2 \Im \left\langle \varphi_t, \left( U_N * \big(|\varphi_t^{(N)}|^2 -
      |\varphi_t|^2\big) \right) \varphi_t^{(N)} \right\rangle.
    \end{aligned}
  \end{equation}
The second term on the r.h.s.\ can be written as
\[  \Im \left\langle \varphi_t, \left( U_N * \big(|\varphi_t^{(N)}|^2 -
      |\varphi_t|^2\big) \right) \varphi_t^{(N)} \right\rangle = \Im \left\langle \varphi_t, \left( U_N * \big(|\varphi_t^{(N)}|^2 -  |\varphi_t|^2\big) \right) (\varphi_t - \varphi_t^{(N)}) \right\rangle. \]
 Hence, by H\"older's and triangle's inequality, 
\begin{equation}\label{ddt-2} \begin{split} 
\Big|  \Im \Big\langle \varphi_t, \Big( U_N * \big(|\varphi_t^{(N)}|^2 &-
      |\varphi_t|^2\big) \Big) \varphi_t^{(N)} \Big\rangle \Big| \\ \leq \; & \| \varphi_t \|_\infty \left\| \left( U_N * \big(|\varphi_t^{(N)}|^2 -  |\varphi_t|^2\big) \right) (\varphi_t - \varphi_t^{(N)}) \right\|_1 \\ \leq \; & 
  \| \varphi_t \|_\infty  \left\| U_N * (|\varphi_t^{(N)}|^2 - |\varphi_t|^2 )\right\|_2 \| \varphi_t - \varphi_t^{(N)} \|_2 \\ \leq \; & \| U_N \|_1 \| \varphi_t \|_\infty \| \varphi_t - \varphi_t^{(N)} \|_2  \left\| |\varphi_t^{(N)}|^2 - |\varphi_t|^2 \right\|_2 \\ \leq \; &  \| U_N \|_1 \| \varphi_t \|_\infty \left( \| \varphi_t \|_\infty + \| \varphi_t^{(N)} \|_\infty \right) \,  \| \varphi_t - \varphi_t^{(N)} \|^2_2 \\ \leq \; & C  \left( \| \varphi_t \|^2_{H^2} + \| \varphi_t^{(N)}  \|^2_{H^2}  \right) \, \,  \| \varphi_t - \varphi_t^{(N)} \|^2_2.
\end{split} \end{equation}
As for the first term on the r.h.s.\ of (\ref{ddt}), we find (since $\int U_N (y) dy = 8 \pi a_0$) 
\[ \begin{split}  \Big| \Big\langle \varphi_t, \big(U_N * |\varphi_t|^2 -& 8 \pi a_0
      |\varphi_t|^2\big) \varphi_t^{(N)} \Big\rangle \Big|  \\ = \; & \left| \int dx \, \overline{\varphi}_t (x) \varphi_t^{(N)} (x) \int dy \, U_N (y) \left( |\varphi_t (x-y)|^2 - |\varphi_t (x)|^2\right) \right|   \\ \leq 
      \; & \int dx dy \, U_N (y) \, |\varphi_t (x)| |\varphi_t^{(N)} (x)| \left| |\varphi_t (x-y)|^2 - |\varphi_t (x)|^2 \right|. 
  \end{split} \]    
Writing $U_N (x) = N^3 U (Nx)$, with $U (x) = V(x) f(x)$ and changing integration variables, we find
\[ \begin{split}
      \Big| \Big\langle \varphi_t, \big(U_N * |\varphi_t|^2 -& 8 \pi a_0
      |\varphi_t|^2\big) \varphi_t^{(N)} \Big\rangle \Big| \\ \leq \; &\int dx dy U (y) |\varphi_t (x)| |\varphi_t^{(N)} (x)| \left| |\varphi_t (x-y/N)|^2 - |\varphi_t (x)|^2 \right|.  
  \end{split} \]
Using
\[ \begin{split} \left| |\varphi_t (x-y/N)|^2 - |\varphi_t (x)|^2 \right|  & = \left| \int_0^1 ds \frac{d}{ds} \, |\varphi_t (x-sy/N)|^2 \right| \\ &\leq 2 |y| N^{-1}\, \int_0^1ds \, |\nabla \varphi_t (x-sy/N)| |\varphi_t (x-sy/N)| \end{split} \]
we conclude that
\begin{equation}\label{ddt-1} \begin{split}
      \Big| \Big\langle \varphi_t, &\big(U_N * |\varphi_t|^2 - 8 \pi a_0
      |\varphi_t|^2\big) \varphi_t^{(N)} \Big\rangle \Big| \\ \leq\; &
   2 N^{-1}   \int dx dy \int_0^1 ds \, U (y) |y| |\varphi_t (x)| |\varphi_t^{(N)} (x)| |\nabla \varphi_t (x-sy/N)| |\varphi_t (x-sy/N)| \\ \leq \; & 2 N^{-1} \| \varphi_t \|^2_\infty \int dx dy \int_0^1 ds \, U(y) |y| \left( |\varphi_t^{(N)} (x)|^2 +  |\nabla \varphi_t (x-sy/N)|^2 \right) \\ \leq \; & CN^{-1} \| \varphi_t \|_\infty^2 \left( \| \varphi_t^{(N)} \|_2^2 + \| \nabla \varphi_t \|_2^2 \right)  \\ \leq \; & CN^{-1} \| \varphi_t \|_{H^2}^2 \end{split} \end{equation}
where the constant $C$ depends on $\int dy U(y) |y|$ and on $\| \varphi \|_{H^1}$. Inserting (\ref{ddt-1}) and (\ref{ddt-2}) into (\ref{ddt}), and using the estimate from part (ii) for $\| \varphi_t \|_{H^2}$, we find
\[\partial_t \| \varphi_t^{(N)} - \varphi_t \|_{2}^2 \le C e^{K|t|} \| \varphi_t^{(N)} - \varphi_t \|_{2}^2 + \frac{C}{N} e^{K|t|}.\]
The claim now follows from Gronwall's inequality, since $\varphi_{t=0} = \varphi^{(N)}_{t=0}$. 
\end{proof}

\section{Properties of the kernel $k_t$}
\label{sec:kernels}

This section is devoted to the proof of Lemma \ref{l:kernels} and Lemma \ref{lm:dotk}.

\begin{proof}[Proof of Lemma \ref{l:kernels}]
(i) We will make use of the bounds (\ref{eq:wN-bd}). The first bound implies immediately that
\begin{equation}\label{eq:k-pt} |k_t (x,y)| \leq \min \left(N |\ph (x)| |\ph (y)| , \frac{1}{|x-y|} |\ph (x)| |\ph (y)| \right) \end{equation}
and therefore, by Hardy's inequality 
\[ \| k_t \|_2^2 \leq C \int dx dy \frac{1}{|x-y|^2} |\ph (x)|^2 |\ph (y)|^2 \leq C \| \ph \|_{H^1}^2 \| \ph \|_2^2  \leq C \, . \]
As for the gradient of $k_t$, we have
\[ \nabla_1 k_t (x,y) = -N^2 \nabla w (N (x-y)) \ph (x) \ph (y) - N w (N (x-y)) \nabla \ph (x) \ph (y) \]
and thus, from the second bound in (\ref{eq:wN-bd}), 
\[ \begin{split} \| \nabla_1 k_t \|_2^2 & \leq C \int dx dy \, \frac{N^4}{(N^2 |x-y|^2 + 1)^2}  |\ph (x)|^2 |\ph (y)|^2  \\ & \leq C N \int dx \frac{N^3}{(N^2 |x|^2 + 1)^2} \| \ph \|_{H^1}^4 \leq C N \end{split} \]
where we used Young and then Sobolev inequalities\done{}. Next we compute
\[ \begin{split} 
\nabla_1(k_t \overline{k}_t) (x,y) =\; & \nabla_x \int dz k_t (x,z) \overline{k}_t (z,y) \\ 
= \; &\nabla_x \left[  \ph (x) \ph (y) \int dz N^2 w (N(x-z)) w(N(z-y)) |\ph (z)|^2 \right] \\
= \; &\nabla \ph (x) \ph (y) \int dz N^2 w (N(x-z)) w(N(z-y)) |\ph (z)|^2\\ &+ \ph (x) \ph (y) \int dz N^3 \nabla w (N (x-z)) w(N(z-y)) |\ph (z)|^2. \end{split} \] 
Using (\ref{eq:wN-bd}), we find
\[ \begin{split}  \| \nabla_1 (k_t \overline{k}_t) \|_2^2 \leq \; &C \int dx dy dz_1 dz_2 \, \frac{|\nabla \ph (x)|^2 |\ph (y)|^2 |\ph(z_1)|^2 |\ph (z_2)|^2}{|x-z_1||z_1 -y||x-z_2||z_2 -y|} \\
&+C \int dx dy dz_1 dz_2 \, \frac{|\ph (x)|^2 |\ph (y)|^2 |\ph(z_1)|^2 |\ph (z_2)|^2}{|x-z_1|^2 |z_1 -y| |x-z_2|^2 |z_2 -y|} \\  \leq \; &C \int dx dy dz_1 dz_2 \, \frac{|\nabla \ph (x)|^2 |\ph (y)|^2 |\ph(z_1)|^2 |\ph (z_2)|^2}{|x-z_1|^2 |z_2 -y|^2} \\
&+C \int dx dy dz_1 dz_2 \, \frac{|\ph (x)|^2 |\ph (y)|^2 |\ph(z_1)|^2 |\ph (z_2)|^2}{|x-z_1|^2 |z_1 - y|^2 |x-z_2|^2} \\ \leq \; &C \| \ph \|_{H^1}^3 \| \ph \|^2_2. \end{split}\]

(ii) The pointwise bound for $k_t(x,y)$ follows directly from (\ref{eq:wN-bd}), as noticed in (\ref{eq:k-pt}). To bound $\lvert r (k_t) (x,y)\rvert$, we observe that, by H\"older's inequality, (\ref{eq:wN-bd}) and part (i),
  \begin{align*}
    |( k_t &\overline{k}_t)^n k_t(x,y)| \\ & = | k_t \overline{k}_t ( k_t \overline{k}_t)^{n-1}
    k_t (x,y)| \\
    & = |\ph(x)| |\ph(y)| \\
& \quad \times \left| \int dz_1 dz_2 \, N^2 w(N(x-z_1))
    w(N(z_2-y)) \ph(z_1) \ph(z_2) \overline{k}_t (k_t
    \overline{k}_t)^{n-1}(z_1,z_2) \right| \\
    & \le |\ph(x)| |\ph(y)| \, \| \overline{k}_t (k_t \overline{k}_t)^{n-1}
    \|_{2} \\
    & \quad \times \left( \int dz_1 N^2 w(N(x-z_1))^2 |\ph(z_1)|^2 \int
    dz_2 \, N^2 w(N(z_2-y))^2 |\ph(z_2)|^2 \right)^{1/2} \\
    & \le C |\ph(x)| |\ph(y)|  \| \nabla \ph \|_{2}^2 \| k
    \|_{2}^{2n-1}.
  \end{align*}
  Thus
  \[
    |r(k_t) (x,y)| \le \sum_{n=1}^\infty \frac{1}{(2n+1)!} |( k_t \overline{k}_t)^n k
    (x,y)|\le C |\ph(x)| |\ph(y)| e^{\| k \|_2}.
  \]
  The pointwise estimate for $p (k_t) (x,y)$ can be proven similarly. 
  This completes the proof of part (ii). Part (iii) follows easily from the pointwise bounds in part (ii).
\end{proof}

\begin{proof}[Proof of Lemma \ref{lm:dotk}]
In the following proof we will use the bounds $\| \dot{\varphi}_t^{(N)} \|_{H^2} , \| \ddot{\varphi}_t^{(N)} \|_2 \leq C e^{K|t|}$ from Proposition \ref{t:pdes}. 

(i) {F}rom (\ref{eq:dtk}), we find
\begin{equation}\label{eq:dotk1}\norm{\dot k_t}^2 \leq 4 \int \di x\di y\, \frac{1}{|x-y|^2}  |\phdot(x)|^2 |\ph(y)|^2 \leq \, C\| \phdot \|_2^2 \| \ph \|_{H^1}^2 \leq C e^{K |t|} \end{equation}
by Hardy's inequality, and from Proposition \ref{t:pdes}, part (iii). Similarly, 
\[ \| \ddot k_t \|_2 \leq C \| \ddot \varphi_t^{(N)} \|_2 + C \| \dot \varphi_t^{(N)} \|_2 \| \nabla \dot \varphi_t^{(N)} \|_2 \leq C e^{K |t|} \]
by Proposition \ref{t:pdes}. Writing $p (k_t) = \sum_{n \geq 0} (k_t \overline{k}_t)^n / (2n!)$, we find immediately
\[ \| \dot p (k_t) \|_2 \leq \sum_{n \geq 0} \frac{1}{(2n)!} n \| \dot k_t \|_2 \| k_t \|_2^{2n-1} \leq \| \dot k_t \|_2 \, e^{ \|k_t \|_2} \leq C e^{K |t|} \]
applying (\ref{eq:dotk1}). The bound for $\dot r (k_t)$ can be proven analogously. 

(ii) {F}rom \ref{l:bt}, part~(v), we have
\[\| \nabla_1 \dot p (k_t) \|_2 \leq C e^{\| k_t \|_2} \left( \| \dot k_t
\|_2 \,  \| \nabla_1 (k_t \overline{k}_t) \|_2 + C \| \nabla_1 (k_t
\dot{\overline{k}}_t) \|_2  + C \| \nabla_1 (\dot k_t \overline{k}_t) \|_2
\right). \]
We are left with the task of estimating $\| \nabla_1 (k_t \dot{\overline{k}}_t) \|_2$ and $\| \nabla_1 (\dot k_t \overline{k}_t) \|_2$. We start by applying the product rule:
\begin{align}
& \norm{\nabla_1 (k_t \dot{\overline{k}}_t)}_2^2  \nonumber \\
& = \int \di x \di y \, \bigg\lvert \nabla_x \, \int dz \left( N w(N(x-z)) \dot{\cc{\varphi}}_t^{(N)} (x) \cc{\varphi}_t^{(N)} (z) + N w(N(x-z))\ \cc{\varphi}_t^{(N)} (x) \dot{\cc{\varphi}}_t^{(N)} (z) \right) \nonumber \\
& \qquad\qquad\qquad \times N w(N(z-y)) {\ph(z)} {\ph(y)} \bigg\rvert^2 \nonumber \\
& \leq 4 \int \di x \di y \left\lvert \int dz N^2 \nabla w (N(z-x)) \dot{\cc{\varphi}}_t^{(N)} (x)  |\ph(z)|^2  \ph(y) N w (N(y-z)) \right|^2 \label{eq:kk1} \\
& \quad + 4 \int \di x \di y \left| \int dz \, N w (N(z-x)) \nabla \dot{\cc{\varphi}}_t^{(N)} (x) |\ph(z)|^2 \ph(y) N w (N(y-z)) \right|^2 \label{eq:kk2}\\
& \quad + 4 \int \di x \di y \left| \int dz \, N^2 \nabla w (N(z-x)) \cc{\varphi}_t^{(N)} (x) \dot{\cc{\varphi}}_t^{(N)} (z) \right. \nonumber \\ & \left. \hspace{6cm} \times \ph(z)  \ph(y) \, N w (N(y-z)) \right|^2 \label{eq:kk3}\\
& \quad + 4 \int \di x \di y \left| \int dz \, N w(N(z-y)) \nabla \cc{\varphi}_t^{(N)} (x) \dot{\cc{\varphi}}_t^{(N)} (z) \ph(z) \ph(y) N w(N(y-z)) \right|^2. \label{eq:kk4}
\end{align}
Next, we estimate the four terms on the r.h.s.\ of the last equation. For the summands \eqref{eq:kk2} and \eqref{eq:kk4} we use that, from (\ref{eq:wN-bd}), $N w(Nx) \leq C |x|^{-1}$. Applying Hardy's inequality, both terms are bounded by $C \norm{\nabla \phdot}_{2}^2 \leq C \exp (K |t|)$, using Proposition \ref{t:pdes}. Since, again by (\ref{eq:wN-bd}), $N^2 \nabla w (Nx)\leq C |x|^{-2}$, the contribution (\ref{eq:kk1}) is bounded by 
\begin{align*}
C \int \di x \di y &\left( \int \di z \frac{1}{|x-z|^2 |z-y|} |\dot{\varphi}_t^{(N)} (x)| |\ph(z)|^2 |\ph(y)|  \right)^2 \\
& = C \int \di x \di y \di z_1 \di z_2 \frac{\lvert\phdot(x)\rvert^2 \lvert \ph(y)\rvert^2 \lvert \ph(z_1)\rvert^2 \lvert \ph(z_2)\rvert^2}{\lvert z_1-y\rvert \lvert z_2 -y\rvert \lvert x-z_1\rvert^2 \lvert x-z_2\rvert^2} \\
& = C \int \di x\, \lvert\phdot(x)\rvert^2 \int \di z_1 \di z_2 \frac{\lvert \ph(z_1)\rvert^2 \lvert \ph(z_2)\rvert^2}{\lvert x-z_1\rvert^2 \lvert x-z_2\rvert^2} \int \di y \, \frac{\lvert \ph(y)\rvert^2}{\lvert z_1-y\rvert \lvert z_2 - y\rvert} \\ &\leq C \| \phdot \|_2^2 \| \ph \|_{H^1}^6 \leq C e^{K|t|}
 \end{align*}
 by Proposition \ref{t:pdes}. Analogously, we can also bound the contribution \eqref{eq:kk3}. 
This shows the bound for $\| \nabla_1 \dot p (k_t) \|_2$. The bounds for $\| \nabla_2 \dot p (k_t) \|_2,\| \nabla_2 \dot p (k_t) \|_2, \| \nabla_2 \dot p (k_t) \|_2$ are proven similarly.

(iii)  {F}rom (\ref{eq:dtk}), using $N w(Nx) \leq C|x|^{-1}$, we find immediately that
\[ \sup_x \| \dot{k}_t (.,x) \|_2 \leq C \left(\| \nabla \phdot \|_2 \| \ph \|_\infty + \| \nabla \ph \|_2 \|\phdot\|_\infty \right) \leq C e^{K|t|} \]
by Proposition \ref{t:pdes}. To show the bound for $\dot{p} (k_t)$, we observe that
\begin{equation}\label{eq:px2} \begin{split}
\dot{p} (k_t) (x,y) = \; & \partial_t \sum_{n \geq 0}\frac{1}{(2n)!} \int dz_1 dz_2 \, k_t (x,z_1) (\overline{k}_t k_t)^{(n-1)} (z_1 ,z_2) \overline{k}_t (z_2 ,y ) \\ = \;& \sum_{n \geq 0}\frac{1}{(2n)!} \Big[  \int dz_1 dz_2 \, \dot{k}_t (x,z_1) (\overline{k}_t k_t)^{(n-1)} (z_1 ,z_2) \overline{k}_t (z_2 ,y ) \\ &\hspace{1.5cm} +
 \int dz_1 dz_2 \, k_t (x,z_1) (\partial_t (\overline{k}_t k_t)^{(n-1)}) (z_1 ,z_2) \overline{k}_t (z_2 ,y )\\
  &\hspace{1.5cm} + \int dz_1 dz_2 \, k_t (x,z_1) (\overline{k}_t k_t)^{(n-1)} (z_1 ,z_2) \dot{\overline{k}}_t (z_2 ,y ) \Big]. \end{split} \end{equation}
The first term in the parenthesis can be bounded in absolute value by
\[ \begin{split} 
\Big| \int &dz_1 dz_2 \, \dot{k}_t (x,z_1) (\overline{k}_t k_t)^{(n-1)} (z_1 ,z_2) \overline{k}_t (z_2 ,y ) \Big| \\ \leq \; &C \int dz_1 dz_2 \, \frac{1}{|x-z_1|} \left( |\phdot (x)| |\ph (z_1)|+ |\ph (x)| |\phdot (z_1)| \right)  \\ &\hspace{5cm} \times | (\overline{k}_t k_t)^{(n-1)} (z_1 ,z_2)| \, \frac{1}{|y-z_2|} |\ph (z_2)| |\ph (y)| \\
\leq \; &C |\ph (x)| |\ph (y)| \int dz_1 dz_2 \, \frac{1}{|x-z_1||y-z_2|} |\phdot (z_1)|  |\ph (z_2)| \,  | (\overline{k}_t k_t)^{(n-1)} (z_1 ,z_2)| \\
&+ C |\phdot (x)| |\ph (y)| \int dz_1 dz_2 \, \frac{1}{|x-z_1||y-z_2|} |\ph (z_1)|  |\ph (z_2)| \,  | (\overline{k}_t k_t)^{(n-1)} (z_1 ,z_2)| 
\\ \leq \; & C \| (\overline{k}_t k_t)^{n-1} \|_2  \, \left( |\ph (x)| |\ph (y)| \| \nabla \phdot \| \| \nabla \ph \| +   |\phdot (x)| |\ph (y)| \| \nabla \ph \|^2 \right).
\\ \leq \; & C \| k_t \|_2^{2(n-1)} (|\phdot (x)| + |\ph (x)|) |\ph (y)|. 
 \end{split}
\]
The last term in the parenthesis on the r.h.s.\ of (\ref{eq:px2}) can be bounded analogously. 
The middle term, on the other hand is bounded in absolute value by
\[ \begin{split}
\Big|  \int dz_1 dz_2 & \, k_t (x,z_1) (\partial_t (\overline{k}_t k_t)^{(n-1)}) (z_1 ,z_2) \overline{k}_t (z_2 ,y ) \Big|
 \\ \leq \; & C |\ph (x)| |\ph(y)| \int dz_1 dz_2 \, \frac{ |\ph (z_1)| |\ph (z_2)|}{|x-z_1| |y-z_2|} |\partial_t (\overline{k}_t k_t)^{n-1} (z_1, z_2)|  \\
\leq \; &C \| \partial_t (\overline{k}_t k_t)^{n-1} \|_2 \, \| \nabla \ph \|_2^2 \,   |\ph (x)| |\ph(y)| 
\\
\leq \; &C \| \dot{k}_t \|_2 \, \| k_t \|_2^{2n-3} |\ph (x)| |\ph (y)| .
\end{split}\] 
Inserting the last bounds in (\ref{eq:px2}), we find
\[ |\dot{p} (k_t) (x,y)| \leq C e^{K |t|}  e^{\| k_t \|_2} \, (|\ph (x)| + |\phdot (x)|) (|\ph (y)| +|\phdot (y)|).  \]
Integrating over $x$ and taking the supremum over $y$ gives, as before using \eqref{eq:hireg},
\[ \sup_y \| \dot{p} (k_t) (.,y) \|_2 \leq C e^{K |t|}. \]
The bound for $\dot{r} (k_t)$ can be proven analogously. Combining the bound for $\dot{r} (k_t)$ with the one for $\dot{k}_t$, we also obtain the bound for $\dot{\text{sh}} (k_t)$.
\end{proof}

\section{Convergence for $N$-particle states}
\label{sec:Npart}

In this section, we show how the result of Theorem \ref{thm:main}, stated there for initial data of the form $W(\sqrt{N} \varphi) T(k_0) \psi$ can be extended to a certain class of data with number of particles fixed to $N$. 
\begin{thm}\label{thm:Npart}
Let $\varphi \in H^4 (\bR^3)$ and suppose $\psi \in \cF$ with $\| \psi \|_\cF  =1$ is such that \begin{equation}\label{eq:ass1-N} \begin{split} \left\langle \psi , \left( \frac{\cN^2}{N} + \cN + \cH_N \right) \psi \right\rangle &\leq C  \end{split} \end{equation}
for a constant $C>0$. Let $P_N$ denote the projection onto the $N$-particle sector of the Fock space and assume that 
\begin{equation}\label{eq:ass2-N} \| P_N W(\sqrt{N} \varphi) T(k_0) \psi \| \geq C N^{-1/4} \end{equation}
for all $N \in \bN$ large enough. We consider the time evolution 
\[ \psi_{N,t} = e^{-i \cH_N t} \frac{P_N W(\sqrt{N} \varphi) T(k_0) \psi}{\| P_N W(\sqrt{N} \varphi) T(k_0) \psi \|} \] and we denote by $\gamma_{N,t}^{(1)}$ the one-particle reduced density associated with the $N$-particle state $\psi_{N,t}$. Then there exist constants $C, c_1 ,c_2 > 0$ with 
\[ \tr  \; \left| \gamma^{(1)}_{N,t} - | \varphi_t \rangle \langle \varphi_t| \right| \leq \frac{C\,  \exp \left( c_1 \exp (c_2 |t|) \right)}{N^{1/4}} \]
for all $t \in \bR$ and all $N$ large enough. Here $\varphi_t$ denotes the solution of the time-dependent Gross-Pitaevskii equation (\ref{eq:GP}), with initial data $\varphi_{t=0} = \varphi$. 
\end{thm}

{\it Remarks. } 
\begin{itemize}
\item[(i)] If we relax (\ref{eq:ass2-N}) to the weaker condition $\| P_N W (\sqrt{N} \varphi) T(k_0) \psi \| \geq C N^{-\alpha}$, for some $1/4 \leq \alpha < 1/2$, the proof below still implies the convergence $\gamma^{(1)}_{N,t} \to |\varphi_t \rangle \langle \varphi_t|$ but only with the slower rate $N^{-1/2+ \alpha}$. 
\item[(ii)] The assumption (\ref{eq:ass2-N}) and its weaker versions mentioned in the previous remark 
are very reasonable; let us explain why. The expected number of particles in the Fock state $W(\sqrt{N} \varphi) T(k_0) \psi$ is given by
\begin{equation}\label{eq:expecN}
\begin{split} &\left\langle W(\sqrt{N} \varphi) T(k_0) \psi , \cN \, W(\sqrt{N} \varphi) T(k_0) \psi \right\rangle \\ & \hspace{2cm}= N + \sqrt{N} \langle T(k_0) \psi , \phi (\varphi) T(k_0) \psi \rangle + \langle T(k_0) \psi , \cN T(k_0) \psi \rangle \end{split} \end{equation}
with the notation $\phi (\varphi) = a(\varphi) + a^* (\varphi)$. {F}rom Lemma \ref{l:a}, Lemma \ref{lm:TNT} and the assumption (\ref{eq:ass1-N}) we conclude that there exists a constant $C >0$ with 
\[ N - CN^{1/2} \leq \left\langle W(\sqrt{N} \varphi) T(k_0) \psi , \cN \,
W(\sqrt{N} \varphi) T(k_0) \psi \right\rangle \leq N + C N^{1/2}. \]
The expectation of $\cN^2$, on the other hand, is given by
\[ \begin{split} 
\big\langle W(\sqrt{N} \varphi) &T(k_0)  \psi , \cN^2 \, W(\sqrt{N} \varphi) T(k_0) \psi \big\rangle \\
= \; &\langle T(k_0) \psi , (\cN + \sqrt{N} \phi (\varphi) + N)^2 T(k_0)
\psi \rangle \\ = \; &N^2 + 2 N^{3/2} \langle T(k_0) \psi , \phi (\varphi)
T(k_0) \psi \rangle + 2N \langle T(k_0) \psi , \cN T(k_0) \psi \rangle \\
&+N \langle  T(k_0) \psi , \phi (\varphi)^2 T(k_0) \psi \rangle + \sqrt{N}
\langle   T(k_0) \psi ,  \left( \cN \phi (\varphi) + \phi (\varphi) \cN
\right) T(k_0) \psi \rangle \\ &+ \langle T(k_0) \psi, \cN^2 T(k_0) \psi
\rangle. \end{split} \]
 Subtracting the square of (\ref{eq:expecN}), and applying again Lemma \ref{l:a}, Lemma \ref{lm:TNT} and the assumption (\ref{eq:ass1-N}), we estimate the variance of the number of particles in the state $W(\sqrt{N} \varphi) T(k_0) \psi$ by
 \[ \left\langle W(\sqrt{N} \varphi) T(k_0) \psi, (\cN - \langle \cN \rangle)^2 W(\sqrt{N} \varphi) T(k_0) \psi \right\rangle \leq C N \]
for an appropriate constant $C >0$. Here $\langle \cN \rangle$ is a shorthand notation for the l.h.s.\ of (\ref{eq:expecN}). We conclude that
\[  \left\langle W(\sqrt{N} \varphi) T(k_0) \psi,  {\bf 1} \left(| \cN -
\langle \cN \rangle | \geq K \sqrt{N} \right)  W(\sqrt{N} \varphi) T(k_0)
\psi \right\rangle \leq C K^{-2}. \]
Choosing $K >0$ sufficiently large, we find
\[ \left\langle W(\sqrt{N} \varphi) T(k_0) \psi,  {\bf 1} \left(| \cN -
\langle \cN \rangle | \leq K \sqrt{N} \right)  W(\sqrt{N} \varphi) T(k_0)
\psi \right\rangle \geq 1/2. \] 
{F}rom (\ref{eq:expecN}), adjusting the value of $K$, we obtain
\[ \left\langle W(\sqrt{N} \varphi) T(k_0) \psi,  {\bf 1} \left(| \cN - N |
\leq K \sqrt{N} \right)  W(\sqrt{N} \varphi) T(k_0) \psi \right\rangle \geq
1/2. \] 
This means that
\[ \sum_{j= N - K \sqrt{N}}^{N+K \sqrt{N}} \left\| P_j \,  W(\sqrt{N} \varphi) T(k_0) \psi \right\|^2 \geq 1/2. \]
The average value of $\| P_j W (\sqrt{N} \varphi) T(k_0) \psi \|^2$ for $j$ between $N- K \sqrt{N}$ and $N+K \sqrt{N}$ is therefore larger or equal to $N^{-1/2}$, in accordance with the assumption (\ref{eq:ass2-N}). In fact, this argument shows that for every $N$ there exists an $M \in [N-K \sqrt{N} , N + K \sqrt{N}]$ with $\| P_M \, W(\sqrt{N} \varphi) T(k_0) \psi \| \geq N^{-1/4}$. Letting \[ \psi_{N,M,t} = e^{-i \cH_N t}  \frac{P_M \, W(\sqrt{N} \varphi) T(k_0) \psi}{\|  P_M \, W(\sqrt{N} \varphi) T(k_0) \psi \|} \]
and denoting by $\gamma^{(1)}_{N,M,t}$ the one-particle reduced density associated with $\psi_{N,M,t}$, one can show, similarly to Theorem \ref{thm:Npart}, that
\[ \tr \; \left| \gamma^{(1)}_{N,M,t} - |\varphi_t \rangle \langle \varphi_t| \right| \leq  \frac{C\,  \exp \left( c_1 \exp (c_2 |t|) \right)}{N^{1/4}} \]
The fact that the number of particles $M$ does not exactly match the parameter $N$ entering the Hamiltonian and the Weyl operator $W(\sqrt{N} \varphi)$ does not affect the analysis in any substantial way, since $|M - N| \leq C N^{1/2} \ll N$. 
\end{itemize}

 \begin{proof}[Proof of Theorem \ref{thm:Npart}]
We write the integral kernel of $\gamma^{(1)}_{N,t}$ as
\[ \begin{split} \gamma^{(1)}_{N,t} &(x;y) \\ = \; &\frac{1}{N \| P_N W(\sqrt{N} \varphi) T(k_0) \psi \|^2} \\ & \hspace{.5cm} \times \left\langle  
e^{-i\cH_N t}  P_N W(\sqrt{N} \varphi) T(k_0) \psi , a_y^* a_x  e^{-i\cH_N t}  P_N W(\sqrt{N} \varphi) T(k_0) \psi \right\rangle \\  = \; & \frac{1}{N \| P_N W(\sqrt{N} \varphi) T(k_0) \psi \|^2} \\ & \hspace{.5cm}  \times \left\langle  
e^{-i\cH_N t}  P_N W(\sqrt{N} \varphi) T(k_0) \psi , W (\sqrt{N} \varphi_t^{(N)}) \left(a_y^* + \sqrt{N} \overline{\varphi}_t^{(N)} (y) \right) \right. \\  & \hspace{2cm} \times \left. \left( a_x + \sqrt{N} \varphi_t^{(N)} (x)\right) W(\sqrt{N} \varphi_t^{(N)})  e^{-i\cH_N t}  W(\sqrt{N} \varphi) T(k_0) \psi \right\rangle. \end{split} \]
Hence, we find
\[ \begin{split} 
 \gamma^{(1)}_{N,t} &(x;y) - \overline{\varphi}_t^{(N)} (y)\varphi_t^{(N)} (x) \\  = \; & \frac{\varphi_t (x)}{\sqrt{N} \| P_N W(\sqrt{N} \varphi) T(k_0) \psi \|^2} 
\\ & \hspace{.1cm}  \times \left\langle  
e^{-i\cH_N t}  P_N W(\sqrt{N} \varphi) T(k_0) \psi , W (\sqrt{N} \varphi_t^{(N)}) a_y^* W(\sqrt{N} \varphi_t^{(N)} ) e^{-i\cH_N t}  W(\sqrt{N} \varphi) T(k_0) \psi \right\rangle \\ &+ \frac{\overline{\varphi}_t (y)}{\sqrt{N} \| P_N W(\sqrt{N} \varphi) T(k_0) \psi \|^2} \\ & \hspace{.1cm}  \times \left\langle  
e^{-i\cH_N t}  P_N W(\sqrt{N} \varphi) T(k_0) \psi , W (\sqrt{N} \varphi_t^{(N)}) a_x W(\sqrt{N} \varphi_t^{(N)} ) e^{-i\cH_N t}  W(\sqrt{N} \varphi) T(k_0) \psi \right\rangle\\ &+ \frac{1}{N \| P_N W(\sqrt{N} \varphi) T(k_0) \psi \|^2} \left\langle  e^{-i\cH_N t}  P_N W(\sqrt{N} \varphi) T(k_0) \psi , W (\sqrt{N} \varphi_t^{(N)}) \right. \\ & \left. \hspace{6cm}  \times  a^*_y a_x W(\sqrt{N} \varphi_t^{(N)})  e^{-i\cH_N t}  W(\sqrt{N} \varphi) T(k_0) \psi \right\rangle. 
\end{split} \]
Therefore, for any compact operator $J$ on $L^2 (\bR^3)$ we find 
\[ \begin{split} 
\tr \; J \, &\left( \gamma^{(1)}_{N,t} - |\varphi_t \rangle\langle \varphi_t| \right) \\ = \; &\frac{1}{\sqrt{N} \| P_N W(\sqrt{N} \varphi) T(k_0) \psi \|^2}  \left\langle  
e^{-i\cH_N t}  P_N W(\sqrt{N} \varphi) T(k_0) \psi , W (\sqrt{N} \varphi_t^{(N)}) \right. \\ & \left. \hspace{4cm} \times \left( a (J\varphi_t) + a^* (J \varphi_t) \right) W(\sqrt{N} \varphi_t^{(N)}  e^{-i\cH_N t}  W(\sqrt{N} \varphi) T(k_0) \psi \right\rangle
\\ &+  \frac{1}{N \| P_N W(\sqrt{N} \varphi) T(k_0) \psi \|^2} \left\langle  
e^{-i\cH_N t}  P_N W(\sqrt{N} \varphi) T(k_0) \psi , W (\sqrt{N} \varphi_t^{(N)}) \right. \\ &\hspace{4cm} \left.  \times d\Gamma (J) \, W(\sqrt{N} \varphi_t^{(N)}  e^{-i\cH_N t}  W(\sqrt{N} \varphi) T(k_0) \psi \right\rangle \end{split} \]
where $d\Gamma (J)$ denotes the second quantization of the operator $J$, defined by  
\( (d\Gamma (J) \psi)^{(n)} = \sum_{i=1}^n J^{(i)} \psi^{(n)} \)
for every $\psi = \{ \psi^{(n)} \}_{n\geq 0} \in \cF$ (here $J^{(i)}$ denotes the operator acting as $J$ on the $i$-th particle and as the identity on the other $(n-1)$ particles). Since $\| d\Gamma (J)  \psi \| \leq \| J \| \, \| \cN \psi \|$, we find, applying Lemma \ref{l:a}, 
\[ \begin{split}
\Big| \tr \; J \, &\left( \gamma^{(1)}_{N,t} - |\varphi_t \rangle\langle \varphi_t| \right)  \Big| \\ \leq \; &\frac{\| J \|}{\sqrt{N} \| P_N W(\sqrt{N} \varphi) T(k_0) \psi \|}   \| (\cN+1)^{1/2} \, W(\sqrt{N} \varphi_t^{(N)})  e^{-i\cH_N t}  W(\sqrt{N} \varphi) T(k_0) \psi \| 
\\ &+  \frac{\| J \|}{N \| P_N W(\sqrt{N} \varphi) T(k_0) \psi \|} \| \cN \, W(\sqrt{N} \varphi_t^{(N)})  e^{-i\cH_N t}  W(\sqrt{N} \varphi) T(k_0) \psi \|  \end{split} \]
where $\| J \|$ denotes the operator norm of $J$. {F}rom Lemma \ref{lm:TNT}, recalling the definition 
(\ref{eq:cU}) of the fluctuation dynamics, we find
\[  \begin{split}
\Big| \tr \; J \, &\left( \gamma^{(1)}_{N,t} - |\varphi_t \rangle\langle \varphi_t| \right)  \Big| \\  \leq \; &\frac{\| J \|}{\sqrt{N} \| P_N W(\sqrt{N} \varphi) T(k_0) \psi \|}  \| \cN^{1/2} \,T^* (k_t) \,  W^* (\sqrt{N} \varphi_t^{(N)})  e^{-i\cH_N t}  W(\sqrt{N} \varphi) T(k_0) \psi \|  
\\ &+  \frac{1}{N \| P_N W(\sqrt{N} \varphi) T(k_0) \psi \|} \| \cN \, T^* (k_t) \, W^* (\sqrt{N} \varphi_t^{(N)}  e^{-i\cH_N t}  W(\sqrt{N} \varphi) T(k_0) \psi \|  \\ \leq \:& \frac{ \| (\cN+1)^{1/2} \, \cU (t;0) \psi \|}{\sqrt{N} \| P_N W(\sqrt{N} \varphi) T(k_0) \psi \|}  + \frac{ \| \cN \, \cU (t;0) \psi \| 
}{N \|  P_N W(\sqrt{N} \varphi) T(k_0) \psi \|}.  \end{split} \]
Using Proposition \ref{prop:apri}, we conclude that
\[  \begin{split}
\left| \tr \; J \, \left( \gamma^{(1)}_{N,t} - |\varphi_t \rangle\langle \varphi_t| \right)  \right|  \leq \:& \frac{C\| J \| \, \| (\cN+1)^{1/2} \, \cU (t;0) \psi \|}{\sqrt{N} \| P_N W(\sqrt{N} \varphi) T(k_0) \psi \|}   + 
\frac{C \| (\cN+1)^{1/2} \, \psi \|  }{\sqrt{N} \|  P_N W(\sqrt{N} \varphi) T(k_0) \psi \|}  
\\ &+ \frac{C \| (\cN+1) \psi \| }{N \| P_N W(\sqrt{N} \varphi) T(k_0) \psi \|}.  \end{split}
\]
{F}rom the assumptions (\ref{eq:ass1-N}) and (\ref{eq:ass2-N}), we obtain
\[ \left| \tr \; J \, \left( \gamma^{(1)}_{N,t} - |\varphi_t \rangle\langle \varphi_t| \right)  \right| \leq \frac{C\| J \|}{N^{1/4}} \| (\cN+1)^{1/2} \, \cU (t;0) \psi \| + \frac{C}{N^{1/4}}. \]
Finally, Theorem \ref{thm:N} implies that
\begin{equation}\label{eq:last} \left| \tr \; J \, \left( \gamma^{(1)}_{N,t} - |\varphi_t \rangle\langle \varphi_t| \right)  \right| \leq \frac{C\| J \| \exp \left(c_1 \exp (c_2 |t|) \right)}{N^{1/4}}. \end{equation}
Since the Banach space $\cL^1 (L^2 (\bR^3))$ is the dual space to the space of compact operators, equipped with the operator norm, (\ref{eq:last}) implies the claim. 
\end{proof}

\end{document}